\documentclass[11pt,a4paper]{amsart}
\usepackage{amsmath}
\usepackage[foot]{amsaddr}
\usepackage{amsfonts}
\usepackage{amssymb}\label{key}
\usepackage{amsthm}
\usepackage{fullpage}
\usepackage[libertine]{newtxmath}
\usepackage[tt=false]{libertine} 
\usepackage{caption}
\usepackage{bbm}
\usepackage{hyperref, color}
\hypersetup{colorlinks=true,citecolor=blue, linkcolor=blue, urlcolor=blue}
\usepackage[linesnumbered,boxed,ruled,vlined]{algorithm2e}
\usepackage{bm}
\usepackage[numbers]{natbib}
\usepackage{xcolor}
\usepackage{enumerate} 
\usepackage{enumitem}
\usepackage{tabularx}
\usepackage{array}
\usepackage{cleveref}
\newcolumntype{L}[1]{>{\raggedright\arraybackslash}p{#1}}
\newcolumntype{C}[1]{>{\centering\arraybackslash}m{#1}}
\newcolumntype{R}[1]{>{\raggedleft\arraybackslash}p{#1}}

\newcommand{\dgraph}{d_\textsf{graph}}
\newcommand{\dham}{d_\textsf{Hamil}}  
\usepackage{makecell}
\usepackage{footnote}
\makeatletter
\def\paragraph{\@startsection{paragraph}{4}%
  \z@\z@{-\fontdimen2\font}%
  {\normalfont\bfseries}}
\makeatother

\makeatletter

\renewcommand{\tocsection}[3]{%
  \indentlabel{\@ifnotempty{#2}{\bfseries\ignorespaces#1 #2\quad}}\bfseries#3}
\renewcommand{\tocsubsection}[3]{%
  \indentlabel{\@ifnotempty{#2}{\ignorespaces#1 #2\quad}}#3}

\newcommand\@dotsep{4.5}
\def\@tocline#1#2#3#4#5#6#7{\relax
  \ifnum #1>\c@tocdepth 
  \else
    \par \addpenalty\@secpenalty\addvspace{#2}%
    \begingroup \hyphenpenalty\@M
    \@ifempty{#4}{%
      \@tempdima\csname r@tocindent\number#1\endcsname\relax
    }{%
      \@tempdima#4\relax
    }%
    \parindent\z@ \leftskip#3\relax \advance\leftskip\@tempdima\relax
    \rightskip\@pnumwidth plus1em \parfillskip-\@pnumwidth
    #5\leavevmode\hskip-\@tempdima{#6}\nobreak
    \leaders\hbox{$\m@th\mkern \@dotsep mu\hbox{.}\mkern \@dotsep mu$}\hfill
    \nobreak
    \hbox to\@pnumwidth{\@tocpagenum{\ifnum#1=1\bfseries\fi#7}}\par
    \nobreak
    \endgroup
  \fi}
\AtBeginDocument{%
\expandafter\renewcommand\csname r@tocindent0\endcsname{0pt}
}
\def\l@subsection{\@tocline{2}{0pt}{2.5pc}{5pc}{}}
\makeatother
\usepackage{lipsum}

\newcommand{\todo}[1]{\typeout{TODO: \the\inputlineno: #1}\textbf{{\color{red}[[[ #1 ]]]}}}

\newcommand{\Update}{\textsf{Update}}

\newcommand{\DTV}[2]{d_{\mathrm{TV}}\left({#1},{#2}\right)}
\newcommand{\E}[1]{\mathbb{E}\left[{#1}\right]}

\newcommand{\one}[1]{\mathbf{1}\left[#1\right]}

\newcommand{\I}{\mathcal{I}}
\newcommand{\D}{\mathcal{D}}

\newcommand{\X}{\boldsymbol{X}}
\newcommand{\Y}{\boldsymbol{Y}}

\newcommand{\margin}[2]{\mu_{{#2},{#1}}}

\newcommand{\Brac}[1]{\left\langle\,{#1}\,\right\rangle}

\newcommand{\T}{\mathcal{T}}

\newcommand{\hatT}{\widehat{\T}}

\newcommand{\exelog}[3]{\left\langle{#2}, #1_t({#2})\right\rangle_{t=1}^{{#3}}}

\newcommand{\Exelog}{\mathsf{Exe\text{-}Log}}

\newcommand{\Imid}{\mathcal{I}_{\mathsf{mid}}}
\newcommand{\UpdateGraph}{\textsf{UpdateGraph}}
\newcommand{\UpdateHamiltonian}{\textsf{UpdateHamiltonian}}

\newcommand{\AddVertex}{\textsf{AddVertex}}
\newcommand{\UpdateEdge}{\textsf{UpdateEdge}}
\newcommand{\DeleteVertex}{\textsf{DeleteVertex}}
\newcommand{\opt}{\mathsf{opt}}
\newcommand{\Rham}{R_{\mathsf{Hamil}}}
\newcommand{\Rgra}{R_{\mathsf{graph}}}
\newcommand{\Tpre}{T_{\mathsf{preparation}}^{\mathsf{single}}}
\newcommand{\Tupd}{T_{\mathsf{update}}^{\mathsf{single}}}
\newcommand{\Tprem}{T_{\mathsf{preparation}}^{\mathsf{multi}}}
\newcommand{\Tupdm}{T_{\mathsf{update}}^{\mathsf{multi}}}
\newcommand{\Trefine}{T_{\mathsf{completion}}^{\mathsf{multi}}}
\newcommand{\Lham}{L_{\mathsf{Hamil}}}
\newcommand{\Lgra}{L_{\mathsf{graph}}}

\newtheorem{theorem}{Theorem}[section]
\newtheorem{observation}[theorem]{Observation}
\newtheorem{claim}[theorem]{Claim}

\newtheorem{lemma}[theorem]{Lemma}
\newtheorem{proposition}[theorem]{Proposition}
\newtheorem{corollary}[theorem]{Corollary}

\theoremstyle{definition}
\newtheorem{definition}[theorem]{Definition}
\newtheorem{condition}[theorem]{Condition}

\newtheorem*{remark*}{Remark}

\newcommand{\datainsert}{\textsf{Insert}}

\newcommand{\dataremove}{\textsf{Remove}}
\newcommand{\dataupdate}{\textsf{Change}}
\newcommand{\dataeval}{\textsf{Eval}}

\newcommand{\datasuccessor}{\textsf{Succ}}
\newcommand{\LengthFix}{\textsf{LengthFix}}

\newcommand{\pup}{p^{\mathsf{up}}}
\newcommand{\Isingset}{\mathcal{P}}

\renewcommand{\emptyset}{\varnothing}

\newcommand{\norm}[1]{\left\Vert#1\right\Vert}
\newcommand{\abs}[1]{\left\vert#1\right\vert}

 \newcommand{\tuple}[1]{\left(#1\right)} \newcommand{\eps}{\varepsilon}
 
 \newcommand{\tp}{\tuple}

\renewcommand{\d}{\,\-d}

\def\*#1{\boldsymbol{#1}} 
\def\+#1{\mathcal{#1}} 
\def\-#1{\mathrm{#1}} 

\title{Dynamic inference in probabilistic graphical models}
\thanks{Weiming Feng and Yitong Yin are supported by the National Key R\&D Program of China 2018YFB1003202 and the National Natural Science Foundation of China under Grant Nos. 61722207 and 61672275.
Kun He and Xiaoming Sun are supported by the National Natural Science Foundation of China Grants No. 61832003, 61433014 and K.C.Wong Education Foundation.}

\author{Weiming Feng}
\author{Kun He}
\author{Xiaoming Sun}
\author{Yitong Yin}

\address[Weiming Feng, Yitong Yin]{State Key Laboratory for Novel Software Technology, Nanjing University. \textnormal{E-mail: \url{fengwm@smail.nju.edu.cn, yinyt@nju.edu.cn}}}

\address[Kun He]{Shenzhen University; Shenzhen Institute of Computing Sciences.
\textnormal{E-mail: \url{hekun.threebody@foxmail.com}}}

\address[Xiaoming Sun]{CAS Key Lab of Network Data Science and Technology, Institute of Computing Technology, Chinese Academy of Sciences. \textnormal{E-mail: \url{sunxiaoming@ict.ac.cn}}}
\date{}

\begin{document}
\begin{abstract}
Probabilistic graphical models, such as Markov random fields (MRFs), are useful for describing high-dimensional distributions in terms of local dependence structures. 
The {probabilistic inference} is a fundamental problem related to graphical models, and sampling is a main approach for the problem.
In this paper, we study probabilistic inference problems when the graphical model itself is changing dynamically with time.
Such dynamic inference problems arise naturally in today's application, e.g.~multivariate time-series data analysis and practical learning procedures.

We give a dynamic algorithm for sampling-based probabilistic inferences in MRFs, where each dynamic update can change the underlying graph and all parameters of the MRF simultaneously, as long as the total amount of changes is bounded.
More precisely, suppose that the MRF has $n$ variables and polylogarithmic-bounded maximum degree, and $N(n)$ independent samples are sufficient for the inference for a polynomial function $N(\cdot)$.
Our algorithm dynamically maintains an answer to the inference problem using $\widetilde{O}(n N(n))$ space cost,
and $\widetilde{O}(N(n) + n)$ incremental time cost upon each update to the MRF, as long as the Dobrushin-Shlosman condition is satisfied by the MRFs.
This well-known condition has long been used for guaranteeing the efficiency of Markov chain Monte Carlo (MCMC) sampling in the traditional static setting.
Compared to the static case, which requires $\Omega(n N(n))$ time cost for redrawing all $N(n)$ samples whenever the MRF changes,
our dynamic algorithm gives a $\widetilde\Omega(\min\{n, N(n)\})$-factor speedup.
Our approach relies on a novel dynamic sampling technique, which transforms  local Markov chains (a.k.a. single-site dynamics) to dynamic sampling algorithms, and an ``algorithmic Lipschitz'' condition that we establish for sampling from graphical models, namely, when the MRF changes by a small difference, samples can be modified to reflect the new distribution, with cost proportional to the difference on MRF.

%

%


\end{abstract}
\maketitle

\setcounter{page}{0} \thispagestyle{empty} \vfill
\pagebreak

\setcounter{tocdepth}{2}
\tableofcontents{}
\setcounter{page}{0} \thispagestyle{empty} \vfill
\pagebreak

\section{Introduction}\label{sec:introduction}
The probabilistic graphical models provide a rich language for describing high-dimensional distributions in terms of the dependence structures between random variables. 
The \emph{Markov random filed} (MRF) is a basic graphical model that encodes pairwise interactions of complex systems.  
Given a graph $G=(V,E)$, each vertex $v\in V$ is associated with a function $\phi_v:Q\to\mathbb{R}$, called the \emph{vertex potential}, on a finite domain $Q=[q]$ of $q$ \emph{spin states}, and each edge $e\in E$ is associated with a symmetric function $\phi_e:Q^2\to\mathbb{R}$, called the \emph{edge potential}, which describes a pairwise interaction.
Together, these induce a probability distribution $\mu$ over all configurations $\sigma\in Q^V$:
\begin{align*}
\mu(\sigma)	 \propto \exp(H(\sigma)) = \exp\Big( \sum_{v \in V}\phi_v(\sigma_v) + \sum_{e=\{u,v\} \in E}\phi_e(\sigma_u,\sigma_v)\Big).	
\end{align*}
This distribution $\mu$ is known as the Gibbs distribution and $H(\sigma)$ is the \emph{Hamiltonian}.
It arises naturally from various physical models, statistics or learning problems, and combinatorial problems in computer science~\cite{mezard2009information,koller2009probabilistic}.

The \emph{probabilistic inference} is one of the most  fundamental computational problems in graphical model.
Some basic inference problems ask to calculate the marginal distribution, conditional distribution, 
or maximum-a-posteriori probabilities of one or several random variables~\cite{wainwright2008graphical}.
{Sampling} is perhaps the most widely used approach for probabilistic inference.
Given a graphical model, independent samples are drawn from the Gibbs distribution and certain statistics are computed using the samples to give estimates for the inferred quantity. 
%
For most typical inference problems, such statistics are easy to compute once the samples are given, 
for instance, for estimating the marginal distribution on a  variable subset $S$, the statistics is the frequency of each configuration in $Q^S$ among the samples,
thus the cost for inference is dominated by the cost for generating random samples~\cite{JVV86,vstefankovivc2009adaptive}.

The classic probabilistic inference assumes a static setting, where the input graphical
model is fixed. In today's application, dynamically changing graphical models naturally arise in many scenarios.
In various practical algorithms for learning graphical models,
e.g.~the contrastive divergence algorithm for learning the restricted Boltzmann machine~\cite{hinton2012practical} and the iterative proportional fitting algorithm for maximum likelihood estimation of graphical models~\cite{wainwright2008graphical},
the optimal model $\I^*$ is obtained by updating the parameters of the graphical model iteratively (usually by gradient descent),
which generates a sequence of graphical models $\I_1,\I_2,\cdots,\I_{M}$,
with the goal that $\I_{M}$ is a good approximation of $\I^*$. 
Also in the study of the multivariate time-series data, the dynamic Gaussian graphical models~\cite{carvalho2007dynamic}, multiregression dynamic model~\cite{queen1993multiregression}, dynamic graphical model~\cite{feng2019dynamic},
and dynamic chain graph models~\cite{anacleto2017dynamic}, are all dynamically changing graphical models and have been used in a variety of applications. 
Meanwhile, with the advent of Big Data, scalable machine learning systems need to deal with continuously evolving graphical models (see e.g.~\cite{renggli2019continuous} and~\cite{smyth2009asynchronous}).

The theoretical studies of probabilistic inference in dynamically changing graphical models are lacking.
In the aforementioned scenarios in practice, it is common that a sequence of graphical models is presented with time, where any two consecutive graphical models can differ from each other in all potentials but by a small total amount.
Recomputing the inference problem from scratch at every time when the graphical model is changed, can give the correct solution, but is very wasteful.
A fundamental question is whether probabilistic inference can be solved dynamically and efficiently.

In this paper,  we study the problem of probabilistic inference in an MRF when the MRF itself is changing dynamically with time.
At each time, the whole graphical model, including all vertices and edges as well as their potentials, are subject to changes.
Such \emph{non-local} updates are very general and cover all applications mentioned above.
The problem of \emph{dynamic inference} then asks to maintain a correct answer to the inference in a dynamically changing MRF
with low incremental cost proportional to the amount of changes made to the graphical model at each time. 
 


\subsection{Our results}
We give a dynamic algorithm for sampling-based probabilistic inferences.
Given an MRF instance with $n$ vertices, suppose that $N(n)$ independent samples are sufficient to give an approximate solution to the inference problem, where 
$N: \mathbb{N}^+ \to \mathbb{N}^+$ is a polynomial function. 
We give dynamic algorithms for general inference problems on dynamically changing MRF. 

Suppose that the current MRF has $n$ vertices and polylogarithmic-bounded maximum degree, 
and each update to the MRF may change the underlying graph and/or all vertex/edge potentials, as long as the total amount of changes is bounded.
Our algorithm maintains an approximate solution to the inference with $\widetilde{O}(n N(n))$ space cost, and with $\widetilde{O}(N(n) + n)$ incremental time cost upon each update,
assuming that the MRFs satisfy the Dobrushin-Shlosman condition~\cite{dobrushin1985completely,dobrushin1985constructive,dobrushin1987completely}. 
The condition has been widely used to imply the efficiency of Markov chain Monte Carlo (MCMC) sampling (e.g.~see \cite{hayes2006simple,dyer2008dobrushin}).
Compared to the static algorithm, which requires $\Omega(n N(n))$ time for redrawing all $N(n)$ samples each time,
our dynamic algorithm significantly improves the time cost with an $\widetilde\Omega(\min\{n, N(n)\})$-factor speedup.

On specific models,
the Dobrushin-Shlosman condition has been established in the literature,  which directly gives us following efficient dynamic inference algorithms, with $\widetilde{O}\left(nN(n)\right)$ space cost and 
$\widetilde{O}\left(N(n) +n \right)$ time cost per update, on graphs with $n$ vertices and maximum degree $\Delta = O(1)$:
\begin{itemize}
\item for Ising model with temperature $\beta$ satisfying $\mathrm{e}^{-2|\beta|}> 1-\frac{2}{\Delta+1}$, which is close to the uniqueness threshold $\mathrm{e}^{-2|\beta_c|}= 1-\frac{2}{\Delta}$, beyond which the static versions of sampling or marginal inference problem for anti-ferromagnetic Ising model is intractable~\cite{galanis2016inapproximability,galanis2015inapproximability};
\item for hardcore model with fugacity $\lambda<\frac{2}{\Delta-2}$, which matches the best bound known for sampling algorithm with near-linear running time on general graphs with bounded maximum degree~\cite{vigoda1999fast,luby1999fast,dyer2000markov};
\item for proper $q$-coloring with $q>2\Delta$, which matches the best bound known for sampling algorithm with near-linear running time on general graphs with bounded maximum degree~\cite{jerrum1995very}.
\end{itemize}

Our dynamic inference algorithm is based on a dynamic sampling algorithm, which efficiently maintains $N(n)$ independent samples for the current MRF while the MRF is subject to changes. 
More specifically, we give a dynamic version of the \emph{Gibbs sampling} algorithm, a local Markov chain for sampling from the Gibbs distribution that has been studied extensively.
%
%
Our techniques are based on: (1)~couplings for dynamic instances of graphical models; and (2)~dynamic data structures for representing single-site Markov chains so that the couplings can be realized algorithmically in sub-linear time.
Both these techniques are of independent interest, and can be naturally extended to more general settings with multi-body interactions.

Our results show that on dynamically changing graphical models,
sampling-based probabilistic inferences can be solved significantly faster than  
rerunning the static algorithm at each time.
This has practical significance in speeding up the iterative procedures for learning graphical models.

\subsection{Related work} 
The problem of dynamic sampling from graphical models was introduced very recently in~\cite{feng2019dynamic}. 
There, a dynamic sampling algorithm was given for graphical models with soft constraints, and can only deal with local updates that change a single vertex or edge at each time.
%
%
The regimes for such dynamic sampling algorithm to be efficient are much more restrictive than the conditions for the rapid mixing of Markov chains.
Our algorithm greatly improves the regimes for efficient dynamic sampling for the Ising and hardcore models in~\cite{feng2019dynamic}, and for the first time, can handle non-local updates that change all vertex/edge potentials simultaneously.
Besides, the dynamic/online sampling from log-concave distributions was also studied in~\cite{narayanan2017efficient,lee2019online}.

Another related topic is the dynamic graph problems, which ask to maintain a solution (e.g.~spanners~\cite{forster2019dynamic, nanongkai2017dynamic, wulff2017fully} or shortest paths~\cite{bernstein2016deterministic, henzinger2016dynamic, henzinger2014decremental}) while the input graph is dynamically changing.
More recently, important progress has been made on dynamically maintaining structures that are related to graph random walks, such as spectral sparsifier~\cite{durfee2019dynamicsoectral,abraham2016fully} or effective resistances~\cite{durfee2018fully,goranci2018dynamic}.
Instead of one particular solution, dynamic inference problems ask to maintain an estimate of a statistics, such statistics comes from an exponential-sized probability space described by a dynamically changing graphical model.

\subsection{Organization of the paper.}
In~\Cref{sec:dynamic-inference-problem}, we formally introduce the dynamic inference problem.
In~\Cref{sec:main-results}, we formally state the main results.
Preliminaries are given in~\Cref{sec:preliminary}. 
In~\Cref{sec:outline-algorithm}, we outline our dynamic inference algorithm.
In~\Cref{sec:agorithm-dynamic-Gibbs}, we present the algorithms for dynamic Gibbs sampling.
The analyses of these dynamic sampling algorithms are given in~\Cref{sec-proof-dynamic}. 
The proof of the main theorem on dynamic inference is given in~\Cref{sec-proof-inference}.
The conclusion is given in Section~\ref{sec:conclusion}.

\section{Dynamic inference problem}\label{sec:dynamic-inference-problem}
\subsection{Markov random fields.} 
%
An instance of \emph{Markov random field (MRF)}  is specified by a tuple $\I = (V, E, Q,\Phi)$, where $G=(V,E)$ is an undirected simple graph; $Q$ is a domain of $q=|Q|$ \emph{spin states}, for some finite $q>1$; and  $\Phi = (\phi_a)_{a \in V \cup E}$ associates each $v\in V$ a  \emph{vertex potential} $\phi_v:Q\to \mathbb{R}$  and each $e\in E$ an \emph{edge potential} $\phi_e:Q^2\to \mathbb{R}$, where $\phi_e$ is symmetric.

A \emph{configuration} $\sigma \in Q^V$ maps each vertex $v\in V$ to a spin state in $Q$, so that each vertex can be interpreted as a variable.
And the \emph{Hamiltonian} of a configuration $\sigma \in Q^V$ is defined as:
\begin{align*}
H(\sigma) \triangleq \sum_{v \in V}\phi_v(\sigma_v) + \sum_{e=\{u,v\}\in E}\phi_e(\sigma_u,\sigma_v). 
\end{align*}
This defines the \emph{Gibbs distribution} $\mu_{\I}$, which is a probability distribution over $Q^V$ such that
 \begin{align*}
\forall\sigma\in Q^V,\quad \mu_{\I}(\sigma)  = \frac{1}{Z}\exp(H(\sigma)),  
\end{align*}
where the normalizing factor $Z \triangleq \sum_{\sigma \in Q^V}\exp(H(\sigma))$ is called the \emph{partition function}.

The Gibbs measure $\mu(\sigma)$ can be $0$ as the  functions $\phi_v, \phi_e$ can take the value $-\infty$. A configuration $\sigma$ is called \emph{feasible} if $\mu(\sigma) > 0$. 
To trivialize the problem of constructing a feasible configuration, we further assume the following natural condition for the MRF instances considered in this paper:\footnote{This condition guarantees that the marginal probabilities are always well-defined, and the problem of constructing a feasible configuration $\sigma$, where $\mu_{\I}(\sigma)>0$, is trivial. The condition holds for all MRFs with soft constraints, or with hard constraints where there is a permissive spin, e.g.~the hardcore model. For MRFs with truly repulsive hard constraints such as proper $q$-coloring, the condition may translate to the  condition $q\ge \Delta+1$  where $\Delta$ is the maximum degree of graph $G$, which is necessary for the irreducibility of local Markov chains for $q$-colorings.}
\begin{align}
\label{eq-assume-MRF}
\forall\, v \in V,\,\, \forall \sigma \in Q^{\Gamma_G(v)}:\quad  \sum_{c \in Q}\exp\left(\phi_v(c)+ \sum_{u \in \Gamma_v}\phi_{uv}(\sigma_u,c)\right) > 0.
\end{align}
where $\Gamma_G(v)\triangleq\{u\in V\mid \{u,v\}\in E\}$ denotes the neighborhood of $v$ in graph $G=(V,E)$.

%
Some well studied typical MRFs include:
\begin{itemize}
\item \emph{Ising model}: The domain of each spin is $Q=\{-1,+1\}$. Each edge $e\in E$ is associated with a \emph{temperature} $\beta_e\in\mathbb{R}$; and each vertex $v\in V$ is associated with a \emph{local field} $h_v\in\mathbb{R}$. For each configuration $\sigma\in\{-1,+1\}^V$, $\mu_{\I}(\sigma)\propto\exp\left(\sum_{\{u,v\}\in E}\beta_e \sigma_u\sigma_v+\sum_{v\in V}h_v\sigma_v\right)$.

%
\item \emph{Hardcore model}: The domain is $Q=\{0,1\}$. 
Each configuration $\sigma\in Q^V$ indicates an independent set in $G=(V,E)$, and $\mu_{\I}(\sigma)\propto \lambda^{\norm{\sigma}}$, where $\lambda > 0$ is the \emph{fugacity} parameter.
\item \emph{proper $q$-coloring}: 
uniform distribution over all proper $q$-colorings of $G=(V,E)$.
\end{itemize}

\subsection{Probabilistic inference and sampling} 
In graphical models, the task of probabilistic inference is to derive the probabilities regarding one or more random variables of the model.
Abstractly, this is described by a function $\*\theta: \mathfrak{M} \rightarrow \mathbb{R}^K$ that 
maps each graphical model $\I \in \mathfrak{M}$ to a target $K$-dimensional probability vector,
where $\mathfrak{M}$ is the class of  graphical models containing the random variables we are interested in and the $K$-dimensional vector describes the probabilities we want to derive.
Given $\boldsymbol{\theta}(\cdot)$ and an MRF instance $\+I \in \mathfrak{M}$, the inference problem asks to estimate the probability vector $\*\theta(\I)$. 

Here are some fundamental inference problems~\cite{wainwright2008graphical} for MRF instances.
Let $\I=(V,E,Q,\Phi)$ be an MRF instance and $A,B\subseteq V$ two disjoint sets where $A \uplus B \subseteq V$. 
\begin{itemize}
\item \emph{Marginal inference}: 
 estimate the marginal distribution $\mu_{A,\I}(\cdot) $ of the variables in $A$, where
 \begin{align*}
 \forall \sigma_A \in Q^A,\quad
\mu_{A,\I}(\sigma_A)  \triangleq \sum_{\tau \in Q^{V\setminus A }}\mu_{\I}(\sigma_A, \tau).
\end{align*}

\item \emph{Posterior inference}: 
given any $\tau_B \in Q^B$, estimate the posterior distribution $\mu_{A,\I}(\cdot \mid \tau_B)$ for the variables in $A$, where
 \begin{align*}
  \forall \sigma_A \in Q^A,\quad
\mu_{A,\I}(\sigma_A\mid \tau_B)  \triangleq \frac{\mu_{A \cup B,\I}(\sigma_A, \tau_B)}{\mu_{B,\I}(\tau_B)}.
\end{align*}

\item \emph{Maximum-a-posteriori (MAP) inference}: 
find the maximum-a-posteriori (MAP) probabilities $P_{A,\I}^{\ast}(\cdot)$
for the configurations over $A$, where 
\begin{align*}
\forall \sigma_A \in Q^A,\quad
P^{\ast}_{A,\I}(\sigma_A)  \triangleq \max_{\tau_B \in Q^B} \mu_{A \cup B,\I}(\sigma_A, \tau_B).
\end{align*}
\end{itemize}
\noindent
All these fundamental inference problems can be described abstractly by a function $\*\theta: \mathfrak{M} \rightarrow \mathbb{R}^K$.
For instances, for marginal inference, 
$\mathfrak{M}$ contains all MRF instances where $A$ is a subset of the vertices, 
$K = \abs{Q}^{|A|}$, and $\*\theta(\I) = (\mu_{A,\I}(\sigma_A))_{\sigma_{A}\in Q^A}$; and for posterior or MAP inference, $\mathfrak{M}$ contains all MRF instances where $A\uplus B$ is a subset of the vertices, $K = \abs{Q}^{|A|}$ and $\*\theta(\I) = (\mu_{A,\I}(\sigma_A\mid \tau_B))_{\sigma_{A}\in Q^A}$ (for posterior inference) or $\*\theta(\I) = (P^{\ast}_{A,\I}(\sigma_A))_{\sigma_{A}\in Q^A}$ (for MAP inference).


One canonical approach for probabilistic inference is by sampling: 
sufficiently many independent samples are drawn (approximately) from the Gibbs distribution of the MRF instance and an estimate of the target probabilities is calculated from these samples.
%
%
Given a probabilistic inference problem $\*\theta(\cdot)$, we use $\+E_{\*\theta}(\cdot)$  to denote an estimating function that approximates $\*\theta(\I)$ using independent samples drawn approximately from $\mu_{\I}$.
For the aforementioned problems of marginal, posterior and MAP inferences, such estimating function $\+E_{\*\theta}(\cdot)$ simply counts the frequency of the samples that satisfy certain properties.

The sampling cost of an estimator is captured in two aspects: the number of samples it uses and the accuracy of each individual sample it requires. 


\begin{definition}[\textbf{$(N,\epsilon)$-estimator for $\*\theta$}]
\label{condition-consistent-estimator}
Let $\*\theta: \mathfrak{M} \to \mathbb{R}^K$ be a probabilistic inference problem and $\+E_{\*\theta}(\cdot)$ an estimating function for $\*\theta(\cdot)$ that for each instance $\I=(V, E, Q, \Phi) \in\mathfrak{M}$, maps samples in $Q^V$ to an estimate of $\*\theta(\I)$.
%
Let $N: \mathbb{N}^+\to\mathbb{N}^+$  and $\epsilon: \mathbb{N}^+ \to (0,1)$. 
For any instance $\I=(V,E,Q,\Phi) \in \mathfrak{M}$ where $n = |V|$,
the random variable $\+E_{\*\theta}(\X^{(1)},\ldots,\X^{(N(n))})$ is said to be an \emph{$(N,\epsilon$)-estimator} for $\*\theta(\I)$ 
if $\X^{(1)},\ldots,\X^{(N(n))} \in Q^V$ are $N(n)$ independent samples drawn approximately from $\mu_\I$ such that 
$\DTV{\X^{(j)}}{\mu_{\I}}\leq \epsilon(n)$ for all $1\le j\le N(n)$.
\end{definition}

In \Cref{condition-consistent-estimator}, an estimator is viewed as a black-box algorithm specified by two functions $N$ and $\epsilon$.
Usually, the estimator is more accurate if more independent samples are drawn and each sample provides a higher level of accuracy.  
Thus, one can choose some large $N$ and small $\epsilon$ to achieve a desired quality of estimate.

\subsection{Dynamic inference problem}
We consider the inference problem where the input graphical model is changed dynamically:
at each step, the current MRF instance $\I = (V,E,Q,\Phi)$ is updated to a new instance $\I'=(V',E',Q,\Phi')$.
We consider general update operations for MRFs that can change both the \textbf{underlying graph} and \textbf{all edge/vertex potentials} simultaneously, where the update request is made by a \emph{non-adaptive adversary} independently of the randomness used by the inference algorithm.
Such updates are general enough and
cover many applications, e.g.~analyses of time series network data~\cite{carvalho2007dynamic,queen1993multiregression,feng2019dynamic,anacleto2017dynamic}, and learning algorithms for graphical  models~\cite{hinton2012practical,wainwright2008graphical}.
%

The difference between the original and the updated instances is measured as follows.

\begin{definition}[\textbf{difference between MRF instances}]
The difference between two MRF instances $\I = (V,E,Q,\Phi)$ and $\I'=(V',E',Q,\Phi')$, where $\Phi=(\phi_a)_{a \in V \cup E}$ and $\Phi'=(\phi'_a)_{a \in V' \cup E'}$, is defined as
 \begin{equation}
 \label{eq-def-dg-dh}
 d(\I,\I') \triangleq\sum_{v\in V\cap V^{'}}\norm{\phi_v-\phi'_v}_1 + \sum_{e\in E\cap E^{'}}\norm{\phi_e- \phi'_e}_1
 +
 |V \oplus V' | + |E\oplus E'|,
\end{equation}
where $A \oplus B = (A \setminus B ) \cup (B \setminus A)$ stands for the symmetric difference between two sets $A$ and $B$, $\norm{\phi_v-\phi'_v}_1 \triangleq \sum_{c \in Q}\abs{\phi_v(c) - \phi'_v(c)}$, and $\norm{\phi_e- \phi'_e}_1 \triangleq \sum_{c,c'\in Q}\abs{\phi_e(c,c')-\phi'_e(c,c')}$.
\end{definition}

Given a probability vector specified by the function $\*\theta:\mathfrak{M} \to \mathbb{R}^K$, the \emph{dynamic inference problem} asks to maintain an estimator $\hat{\*\theta}(\I)$ of $\boldsymbol{\theta}(\I)$ for the current MRF instance $\I = (V, E, Q, \Phi) \in \mathfrak{M}$, with a data structure, such that when $\I$ is updated to $\I'=(V',E',Q,\Phi')\in \mathfrak{M}$, 
the algorithm updates $\hat{\*\theta}(\I)$ to an estimator $\hat{\boldsymbol{\theta}}(\I')$ of the new vector $\boldsymbol{\theta}(\I')$, or equivalently, outputs the difference between the estimators $\hat{\boldsymbol{\theta}}(\I)$ and $\hat{\boldsymbol{\theta}}(\I')$.

It is desirable to have the dynamic inference algorithm which maintains
an $(N,\epsilon)$-estimator for $\*\theta(\I)$ for the current instance $\I$.
However, the dynamic algorithm cannot be efficient 
if $N(n)$ and $\epsilon(n)$ change drastically with $n$ 
(so that significantly more samples or substantially more accurate samples may be needed when a new vertex is added), 
or if recalculating the estimating function $\+E_{\*\theta}(\cdot)$ itself is expensive.
We introduce a notion of \emph{dynamical efficiency} for the estimators that are suitable for dynamic inference.

\begin{definition}[\textbf{dynamical efficiency}]
\label{definition-estimator-dynamic}
Let $N: \mathbb{N}^+\to\mathbb{N}^+$ and $\epsilon: \mathbb{N}^+ \to (0,1)$. 
Let $\+E(\cdot)$ be an estimating function for some $K$-dimensional probability vector of MRF instances.
An tuple $(N,\epsilon,\+E)$ is said to be \emph{dynamically efficient} if it satisfies:
\begin{itemize}
\item \textbf{(bounded difference)} there exist constants $C_{1},C_{2} > 0$ such that for any $n \in \mathbb{N}^+$, 
\begin{align*}
\abs{N(n+1)-N(n)} \leq \frac{C_{1} \cdot N(n)}{n}
\quad\text{ and }\quad
\abs{\epsilon(n+1)-\epsilon(n)} \leq \frac{C_{2} \cdot \epsilon(n)}{n};
\end{align*}
\item \textbf{(small incremental cost)} 
there is a deterministic algorithm that maintains $\+E(\X^{(1)},\ldots, \X^{(m)})$ using $(m n+K)\cdot \mathrm{polylog}(m n)$ bits where $\X^{(1)},\ldots,\X^{(m)} \in Q^V$ and $n=|V|$, 
such that  when $\X^{(1)},\ldots,\X^{(m)} \in Q^V$ are updated to  $\Y^{(1)},\ldots, \Y^{(m')} \in Q^{V'}$, where $n'=|V'|$, the algorithm updates $\+E(\X^{(1)},\ldots, \X^{(m)})$ to  $\+E(\Y^{(1)},\ldots ,\Y^{(m')})$ within time cost $\+D\cdot\mathrm{polylog}(mm'nn') +O(m+m')$, where $\+D$ is the size of the difference between two sample sequences defined as:
\begin{align}
\label{eq-def-diff-sample}
\+D \triangleq \sum_{i \leq \max\{m,m'\}} \sum_{v \in V\cup V'}\one{\X^{(i)}(v) \neq \Y^{(i)}(v)},
\end{align}
where an unassigned $\X^{(i)}(v)$ or $\Y^{(i)}(v)$ is not equal to any assigned spin.
\end{itemize}
\end{definition}

The dynamic efficiency basically asks $N(\cdot), \epsilon(\cdot)$, and $\+E(\cdot)$  to have some sort of ``Lipschitz'' properties.
To satisfy the bounded difference condition, $N(n)$ and $1/\epsilon(n)$ are necessarily polynomially bounded, and they can be any constant, polylogarithmic, or polynomial functions, or multiplications of such functions.
The condition with  small incremental cost also holds very commonly. 
In particular, it is satisfied by the estimating functions for all the aforementioned problems for the marginal, posterior and MAP inferences as long as the sets of variables have sizes $\abs{A},\abs{B} = O(\log n)$.
We remark that the $O(\log n)$ upper bound is somehow necessary for the efficiency of inference, because otherwise the dimension of $\*\theta(\I)$ itself (which is at least $q^{|A|}$) becomes super-polynomial in $n$.

\section{Main results}\label{sec:main-results}

Let $\I = (V, E, Q,\Phi)$ be an MRF instance, where $G=(V,E)$. Let $\Gamma_G(v)$ denote the neighborhood of $v$ in $G$.  
For any vertex $v \in V$ and any configuration $\sigma \in Q^{\Gamma_G(v)}$, we use $\mu_{v,\I}^\sigma(\cdot)=\mu_{v,\I}(\cdot \mid \sigma)$ to denote the marginal distribution on $v$ conditional on $\sigma$:
\begin{align*}
\forall c \in Q:\quad \mu_{v,\I}^\sigma(c) = \mu_{v,\I}(c \mid \sigma) \triangleq  \frac{\exp\left(\phi_v(c)+ \sum_{u \in \Gamma_G(v)}\phi_{uv}(\sigma_u,c)\right)}{\sum_{a \in Q}\exp\left(\phi_v(a)+ \sum_{u \in \Gamma_G(v)}\phi_{uv}(\sigma_u,a)\right)}.
\end{align*}
Due to the assumption in~\eqref{eq-assume-MRF}, the marginal distribution is always well-defined.
The following condition is the \emph{Dobrushin-Shlosman condition}~\cite{dobrushin1985completely,dobrushin1985constructive,dobrushin1987completely, hayes2006simple,dyer2008dobrushin}.
%
%

%
%
%



\begin{condition}[\textbf{Dobrushin-Shlosman condition}]
\label{condition-Dobrushin}
Let $\mathcal{I} = (V, E, Q, \Phi)$ be an MRF instance with Gibbs distribution $\mu=\mu_\I$.
Let $A_{\I} \in \mathbb{R}_{\geq 0}^{V \times V}$ be the \emph{influence matrix} which is defined as 
\begin{align*}
A_{\I}(u,v) \triangleq 
\begin{cases}
\max_{(\sigma,\tau) \in B_{u,v}}\DTV{\mu^\sigma_{v}}{\mu_{v}^\tau},
& \{u,v\}\in E,\\
0 & \{u,v\}\not\in E,
\end{cases}
\end{align*}
where the maximum is taken over the set $B_{u,v}$ of all  $(\sigma,\tau)\in Q^{\Gamma_G(v)}\times Q^{\Gamma_G(v)}$ that differ only at $u$, and $\DTV{\mu^\sigma_{v}}{\mu_{v}^\tau} \triangleq \frac{1}{2}\sum_{c \in Q}\abs{\mu^\sigma_v(c) - \mu^\tau_v(c)}$ is the total variation distance between $\mu^\sigma_{v}$ and $\mu^\tau_{v}$.
An MRF instance $\mathcal{I}$ is said to satisfy the \emph{Dobrushin-Shlosman condition} if there is a constant $\delta>0$ such that 
\begin{align*}
\max_{u \in V} \sum_{v \in V} A_{\I}(u, v) \leq 1 - \delta. 
\end{align*}
\end{condition}

Our main theorem assumes the following setup:
Let $\*\theta: \mathfrak{M} \to \mathbb{R}^K$ be a probabilistic inference problem that maps each  MRF instance in $\mathfrak{M}$ to a $K$-dimensional probability vector,  and let $\+E_{\*\theta}$ be its estimating function.
Let $N:\mathbb{N}^+ \to \mathbb{N}^+$ and $\epsilon:\mathbb{N}^+ \to (0,1)$.
We use $\I = (V, E, Q, \Phi) \in \mathfrak{M}$, where $n=|V|$,  to denote the current  instance and $\I'= (V', E', Q, \Phi')\in\mathfrak{M}$, where $n'=|V'|$, to denote  the updated instance. 

\begin{theorem}[\textbf{dynamic inference algorithm}]
\label{theorem-infer-MRF}
Assume that $(N,\epsilon, \+E_{\*\theta})$ is dynamically efficient, both $\I$ and $\I'$ satisfy the Dobrushin-Shlosman condition, and $d(\I,\I') \leq L = o(n)$.

There is an algorithm 
that maintains an $(N,\epsilon)$-estimator $\hat{\boldsymbol{\theta}}(\I)$ of the probability vector $\boldsymbol{\theta}(\I)$ for the current MRF instance $\I$,  
 using   $\widetilde{O}\left(nN(n)+K\right)$ bits,
such that when $\I$ is updated  to $\I'$, the algorithm updates $\hat{\boldsymbol{\theta}}(\I)$ to 
an $(N,\epsilon)$-estimator $\hat{\boldsymbol{\theta}}(\I')$ of $\boldsymbol{\theta}(\I')$ for the new instance $\I'$,
within expected time cost
$$\widetilde{O}\left(\Delta^2 L N(n) +\Delta n \right),$$
where $\widetilde{O}(\cdot)$ hides a $\mathrm{polylog}(n)$ factor,
$\Delta= \max\{\Delta_G,\Delta_{G'}\}$, where $\Delta_G$ and $\Delta_{G'}$ denote the maximum degree of $G=(V,E)$ and $G'=(V',E')$ respectively.  
\end{theorem}

\noindent
Note that the extra $O(\Delta n)$ cost is necessary for editing the current MRF instance $\I$ to $\I'$.

Typically,  the difference between two MRF instances $\I,\I'$ is small\footnote{In multivariate time-series data analysis, the MRF instances of two sequential times are similar.
In the iterative algorithms for learning graphical models,
the difference between two sequential MRF instances generated by gradient descent are bounded 
to prevent oscillations.
Specifically, the difference is very small when the iterative algorithm approaches to the convergence state~\cite{hinton2012practical,wainwright2008graphical}.}, and the underlying graphs are sparse~\cite{sa2016ensuring}
, that is, $L,\Delta \leq \mathrm{polylog}(n)$. 
In such cases, our algorithm updates the estimator within time cost $\widetilde{O}(N(n) + n)$, which significantly outperforms static sampling-based inference algorithms that require time cost $\Omega(n'N(n')) = \Omega(nN(n))$
for redrawing all $N(n')$ independent samples.

\bigskip
\paragraph{Dynamic sampling}
The core of our dynamic inference algorithm is a dynamic sampling algorithm: 
Assuming the Dobrushin-Shlosman condition, the algorithm can maintain a sequence of $N(n)$ independent samples $\X^{(1)},\ldots,\X^{(N(n))} \in Q^V$ that are $\epsilon(n)$-close to $\mu_{\I}$ in total variation distance, and when $\I$ is updated to $\I'$ with difference $d(\I,\I')\le L=o(n)$, the algorithm can update the maintained samples to $N(n')$ independent samples $\Y^{(1)},\ldots,\Y^{(N(n'))}\in Q^{V'}$ that are $\epsilon(n')$-close to $\mu_{\I'}$ in total variation distance, using a time cost $\widetilde{O}\left(\Delta^2L N(n) +\Delta n \right)$ in expectation.
This shows an ``algorithmic Lipschitz'' condition holds for sampling from Gibbs distributions: when the MRF changes insignificantly, a population of samples can be modified to reflect the new distribution, with cost proportional to the difference on MRF.
We show that such property is guaranteed by the Dobrushin-Shlosman condition.
This {dynamic sampling} algorithm is formally described in 
Theorem~\ref{theorem-sample-MRF}
and is of independent interest~\cite{feng2019dynamic}.

\bigskip
\paragraph{Applications on specific models}
On specific models, 
we have the following results, where $\delta>0$ is an arbitrary constant. 
\begin{table}[h]
\begin{tabular}{|C{1.9cm}|C{2.9cm}|C{3.1cm}|C{5cm}|}
\hline  {\small model}& regime & space cost & time cost for each update \\
\hline 
Ising  & $\mathrm{e}^{-2|\beta|} \geq 1 - \frac{2-\delta}{\Delta + 1}$ & $\widetilde{O}\left(nN(n)+K\right)$ & $\widetilde{O}\left(\Delta^2L N(n) +\Delta n \right)$ \\
\hline 
hardcore  & $\lambda\le\frac{2-\delta}{\Delta-2}$ & $\widetilde{O}\left(nN(n)+K\right)$ & $\widetilde{O}\left(\Delta^3LN(n)  +\Delta n \right)$ \\ 
\hline 
$q$-coloring & $q\ge(2+\delta)\Delta$ & $\widetilde{O}\left(nN(n)+K\right)$ &  $\widetilde{O}\left(\Delta^2LN(n)  +\Delta n \right)$\\
\hline
\end{tabular}
\caption{Dynamic inference for specific models.}\label{table:results-specific-models}
\end{table}

\noindent
The results for Ising model and $q$-coloring are corollaries of \Cref{theorem-infer-MRF}. The regime for hardcore model is better than the Dobrushin-Shlosman condition (which is $\lambda \leq \frac{1-\delta}{\Delta - 1}$), because we use the coupling introduced by Vigoda~\cite{vigoda1999fast} to analyze the algorithm.

\section{Preliminaries}\label{sec:preliminary}

%
\paragraph{Total variation distance and coupling}
Let $\mu$ and $\nu$ be two distributions over $\Omega$. The \emph{total variation distance} between $\mu$ and $\nu$ is defined as
\begin{align*}
\DTV{\mu}{\nu} \triangleq \frac{1}{2}\sum_{x \in \Omega}\abs{\mu(x)-\nu(x)}.	
\end{align*}
A \emph{coupling} of  $\mu$ and $\nu$ is a joint distribution $(X, Y) \in \Omega \times \Omega$ such that marginal distribution of $X$ is $\mu$ and the marginal distribution of $Y$ is $\nu$. The following coupling lemma is well-known.

\begin{proposition}[\textbf{coupling lemma}]
For any coupling $(X,Y)$ of $\mu$ and $\nu$, it holds that
\begin{align*}
\Pr[X \neq Y] \geq \DTV{\mu}{\nu}.	
\end{align*}
Furthermore, there is an \emph{optimal coupling} that achieves equality.
\end{proposition}

\paragraph{Local neighborhood}
Let $G = (V,E)$ be a graph.
For any vertex $v \in V$, let $\Gamma_G(v)\triangleq\{u\in V\mid \{u,v\}\in E\}$ denote the neighborhood of $v$, and $\Gamma^+_G(v) \triangleq \Gamma_G(v) \cup \{v\}$ the inclusive neighborhood of $v$.
We simply write $\Gamma_v=\Gamma(v)=\Gamma_G(v)$ and $\Gamma_v^+=\Gamma^+(v)=\Gamma_G^+(v)$ for short when $G$ is clear in the context.
We use $\Delta=\Delta_G\triangleq\max_{v\in V}|\Gamma_v|$ to denote the maximum degree of graph $G$.

A notion of \textbf{local neighborhood for MRF} is frequently used.
Let $\I=(V,E,Q,\Phi)$ be an MRF instance.
For $v\in V$, we denote by $\I_v\triangleq\I[\Gamma_v^+]$ the restriction of $\I$ on the inclusive neighborhood $\Gamma_v^+$ of $v$, i.e.~$\I_v=(\Gamma^+_v,E_v,Q,\Phi_v)$, where $E_v=\{\{u,v\}\in E\}$ and $\Phi_v=(\phi_a)_{a\in \Gamma_v^+\cup E_v}$.

\vspace{1em}
\paragraph{Gibbs sampling}
The \emph{Gibbs sampling} (a.k.a.~\emph{heat-bath}, \emph{Glauber dynamics}), is a classic Markov chain for sampling from Gibbs distributions.
Let $\I = (V, E, Q, \Phi)$ be an MRF instance and  $\mu=\mu_{\I}$ its Gibbs distribution.
The chain of Gibbs sampling (Algorithm~\ref{alg-Gibbs}) is on the space $\Omega \triangleq Q^V$, and has the stationary distribution $\mu_\I$~\cite[Chapter~3]{levin2017markov}. 

\begin{algorithm}[ht]
\SetKwInOut{Initialization}{Initialization}
\Initialization{an initial state $\X_0 \in \Omega$ (not necessarily feasible);}
\For{$t=1,2,\ldots,T$}{
  pick $v_t \in V$ uniformly at random\;
  draw a random value $c\in Q$ from the marginal distribution $\mu_{v_t}(\cdot\mid X_{t-1}(\Gamma_{v_t}))$\;\label{step-resample}
  $X_t(v_t)\gets c$ and $X_t(u) \gets X_{t-1}(u)$ for all $u \in V \setminus \{v_t\}$\;
}
\caption{Gibbs sampling}\label{alg-Gibbs}
\end{algorithm}
%
\paragraph{\textbf{Marginal distributions}}
Here $\mu_{v}(\cdot \mid \sigma(\Gamma_v))=\mu_{v,\I}(\cdot \mid \sigma(\Gamma_v))$ denotes the marginal distribution at $v\in V$ conditioning on $\sigma(\Gamma_v)\in Q^{\Gamma_v}$, which is computed as:
\begin{align}
\label{eq-marginal}
\forall c \in Q:\quad \mu_{v}(c \mid \sigma(\Gamma_v)) = \frac{\phi_{v}(c) \prod_{u \in \Gamma_v}\phi_{uv}(\sigma_u, c)}{\sum_{c'\in Q}\phi_{v}(c') \prod_{u \in \Gamma_v}\phi_{uv}(\sigma_u, c') }.
\end{align}
Due to the assumption~\eqref{eq-assume-MRF}, this marginal distribution is always well defined, and its computation uses only the information of $\I_v$. 

%
%

\vspace{1em}
\paragraph{Coupling for mixing time}
Consider a chain $(\X_t)_{t=0}^\infty$ on space $\Omega$ with stationary distribution $\mu_\I$ for MRF instance $\I$.
The \emph{mixing rate} is defined as: for $\epsilon>0$,
\begin{align*}
\tau_{\mathsf{mix}}(\I, \epsilon) \triangleq \max_{\X_0}\min\left\{t \mid \DTV{\X_t}{\mu_{\I}} \leq \epsilon \right\},
\end{align*}
where $\DTV{\X_t}{\mu_\I}$ denotes the \emph{total variation distance} between $\mu_\I$ and the distribution of $\X_t$.

A coupling of a Markov chain is a joint process $(\X_t,\Y_t)_{t\geq 0}$ such that $(\X_t)_{t\geq 0}$ and $(\Y_t)_{t\geq 0}$ marginally follow the same transition rule as the original chain.
%
%
Consider the following type of couplings.
\begin{definition}[\textbf{one-step optimal coupling for Gibbs sampling}]
\label{def-one-step-coupling-static-Gibbs}
A coupling $(\X_t,\Y_t)_{t\geq 0}$ of Gibbs sampling on an MRF instance $\I= (V, E, Q, \Phi)$ is a \emph{one-step optimal coupling} if it is constructed as follows: 
For $t=1,2,\ldots$,
\begin{enumerate}
\item pick the same random $v_t \in V$, and  let $(X_t(u),Y_t(u)) \gets (X_{t-1}(u),Y_{t-1}(u))$ for all $u\neq v_t$;
\item sample $(X_t(v_t),Y_t(v_t))$ from an optimal coupling $D_{\opt,\I_{v_t}}^{\sigma,\tau}(\cdot,\cdot)$
of the marginal distributions $\mu_{v_t}(\cdot\mid \sigma)$ and $\mu_{v_t}(\cdot\mid \tau)$ where $\sigma=X_{t-1}(\Gamma_{v_t})$ and $\tau=Y_{t-1}(\Gamma_{v_t})$. 
\end{enumerate}
\end{definition}
\noindent
The coupling $D_{\opt,\I_{v_t}}^{\sigma,\tau}(\cdot,\cdot)$ is an \emph{optimal coupling} of $\mu_{v_t}(\cdot\mid \sigma)$ and $\mu_{v_t}(\cdot\mid \tau)$ that attains the maximum $\Pr[\boldsymbol{x}=\boldsymbol{y}]$ for all couplings $(\boldsymbol{x},\boldsymbol{y})$ of $\boldsymbol{x}\sim\mu_{v_t}(\cdot\mid \sigma)$ and $\boldsymbol{y}\sim\mu_{v_t}(\cdot\mid \tau)$. The coupling   $D_{\opt,\I_{v_t}}^{\sigma,\tau}(\cdot,\cdot)$  is determined by the local information $\I_v$ and $\sigma,\tau\in Q^{\mathrm{deg}(v)}$.

With such a coupling, we can establish the following relation between the Dobrushin-Shlosman condition and the rapid mixing of the Gibbs sampling~\cite{dobrushin1985completely,dobrushin1985constructive,dobrushin1987completely,bubley1997path, hayes2006simple,dyer2008dobrushin}.
\begin{proposition}[\cite{bubley1997path,hayes2006simple}]
\label{proposition-mixing-Gibbs}
Let $\I = (V, E, Q, \Phi)$ be an MRF instance with $n = |V|$, and $\Omega = Q^V$ the state space.
Let $H(\sigma,\tau) \triangleq \abs{\{v \in V \mid \sigma_v \neq \tau_v\}}$ denote the Hamming distance between $\sigma \in \Omega$ and $\tau \in \Omega$.
If $\I$ satisfies the Dobrushin-Shlosman condition (\Cref{condition-Dobrushin}) with constant $\delta > 0$, then the one-step optimal coupling $(\X_t,\Y_t)_{t\geq 0}$  for Gibbs sampling (\Cref{def-one-step-coupling-static-Gibbs}) satisfies
\begin{align*}
\forall\,\sigma, \tau \in \Omega:\quad \E{\,H(\X_t,\Y_t)\mid \X_{t-1} = \sigma \land \Y_{t-1}=\tau\,} \leq \left(1-\frac{\delta}{n}\right)\cdot H(\sigma, \tau),
\end{align*}
and hence the mixing rate of Gibbs sampling on $\I$ is bounded as
$\tau_{\mathsf{mix}}(\I, \epsilon) \leq \left\lceil\frac{n}{\delta}\log \frac{n}{\epsilon}\right\rceil$.
\end{proposition}

\section{Outlines of algorithm}
\label{sec:outline-algorithm}
Let $\*\theta: \mathfrak{M} \to \mathbb{R}^K$ be a probabilistic inference problem that maps each  MRF instance in $\mathfrak{M}$ to a $K$-dimensional probability vector,  and let $\+E_{\*\theta}$ be its estimating function.
%
Le $\I = (V, E, Q, \Phi) \in \mathfrak{M}$ be the current  instance, where $n=|V|$.
Our dynamic inference algorithm maintains a sequence of $N(n)$ independent samples $\*X^{(1)},\ldots,\*X^{(N(n))} \in Q^V$ which are $\epsilon(n)$-close to the Gibbs distribution $\mu_{\I}$ in total variation distance and 
an $(N,\epsilon)$-estimator $\hat{\*\theta}(\I)$ of $\*\theta(\I)$ such that 
\begin{align*}
\hat{\*\theta}(\I) = \+E_{\*\theta}(\X^{(1)}, \X^{(2)},\ldots,\X^{(N(n))}).	
\end{align*}
Upon an update request that modifies $\I$ to a new instance $\I'=(V',E',Q,\Phi')\in \mathfrak{M}$, where $n' = |V'|$, our algorithm does the followings:
\begin{itemize}
\item \emph{Update the sample sequence.} 
Update $\*X^{(1)},\ldots,\*X^{(N(n))}$ to a new sequence of $N(n')$ independent samples $\*Y^{(1)},\ldots,\*Y^{(N(n'))} \in Q^{V'}$ that are $\epsilon(n')$-close to $\mu_{\I'}$ in total variation distance, and output the difference between two sample sequences.
\item \emph{Update the estimator.} 
Given the difference between the two sample sequences, 
update $\hat{\*\theta}(\I)$ to $\hat{\*\theta}(\I') =\mathcal{E}_{\*\theta}(\Y^{(1)},\ldots,\Y^{(N(n'))})$ by accessing the oracle in \Cref{definition-estimator-dynamic}.
\end{itemize}
Obviously, the updated estimator $\hat{\*\theta}(\I')$ is an $(N,\epsilon)$-estimator for $\*\theta(\I')$.

Our main technical contribution is to give an algorithm that dynamically maintains a sequence of $N(n)$ independent samples for $\mu_\I$, while $\I$ itself is dynamically changing.
The dynamic sampling problem was recently introduced in~\cite{feng2019dynamic}. 
The dynamical sampling algorithm given there only handles update of a single vertex or edge and works only for graphical models with soft constraints.

In contrast, our dynamic sampling algorithm maintains a sequence of $N(n)$ independent samples for $\mu_\I$ within total variation distance $\epsilon(n)$, while the entire specification of the graphical model $\I$ is subject to dynamic update (to a new $\I'$ with difference $d(\I,\I')\le L=o(n)$).
Specifically, the algorithm updates the sample sequence within expected time $O(\Delta^2 N(n)L \log^3 n + \Delta n)$. 
Note that the extra $O(\Delta n)$ cost is necessary for just editing the current MRF instance $\I$ to $\I'$ because a single update may change all the vertex and edge potentials simultaneously.
This incremental time cost dominates the time cost of the dynamic inference algorithm, and is efficient for maintaining $N(n)$ independent samples, especially when $N(n)$ is sufficiently large, e.g.~$N(n)=\Omega(n/L)$, in which case the average incremental cost for updating each sample is $O(\Delta^2 L \log^3n + {\Delta n}/{N(n)})=O(\Delta^2 L \log^3n)$.

We illustrate the main idea by explaining how to maintain one sample.
The idea is to represent the trace of the Markov chain for generating the sample by a dynamic data structure, and when the MRF instance is changed, this trace is modified to that of the new chain for generating the sample for the updated instance.
This is achieved  by both a set of efficient dynamic data structures and the coupling between the two Markov chains.

Specifically, let $(\X_t)_{t=0}^T$ be the Gibbs sampler chain for distribution $\mu_{\I}$.
When the chain is rapidly mixing, starting from an arbitrary initial configuration $\X_0\in Q^V$, after suitably many steps $\X=\X_T$ is an accurate enough sample for $\mu_{\I}$.
At each step, $\X_{t-1}$ and $\X_{t}$ may differ only at a vertex $v_t$ which is picked from $V$ uniformly and independently at random. 
The evolution of the chain is fully captured by the initial state $\X_0$ and the sequence of pairs $\Brac{v_{t},X_t(v_t)}$, from $t=1$ to $t=T$, which is called the \emph{execution log} of the chain in the paper.

Now suppose that the current instance $\I$ is updated to $\I'$.
We construct such a coupling between the original chain $(\X_t)_{t=0}^T$ and the new chain $(\Y_t)_{t=0}^T$, such that $(\Y_t)_{t=0}^T$ is a faithful Gibbs sampling chain for the updated instance $\I'$ given that $(\X_t)_{t=0}^T$ is a faithful chain for $\I$, and the difference between the two chains is small, in the sense that they have almost the same execution logs except for about $O(TL/n)$ steps, where $L$ is the difference between $\I$ and $\I'$.

To simplify the exposition of such coupling, for now we restrict ourselves to the cases where the update to the instance $\I$ does not change the set of variables. 
Without loss of generality, 
we only consider the following two basic update operations that modifies $\I$ to $\I'$.
\begin{itemize}
\item \emph{Graph update}. The update only adds or deletes some edges, while all vertex potentials and the potentials of unaffected edges are not changed.
\item \emph{Hamiltonian update}.	The update changes (possibly all) potentials of vertices and edges, while the underlying graph remains unchanged.
\end{itemize}
The general update of graphical model can be obtained by combining these two basic operations.

Then the new chain $(\Y_t)_{t=0}^T$ can be coupled with $(\X_t)_{t=0}^T$ by using the same initial configuration $\Y_0=\X_0$ and the same sequence $v_1,v_2,\ldots,v_T\in V$ of randomly picked vertices. 
And for $t=1,2,\ldots, T$, the transition $\Brac{v_{t},Y_t(v_t)}$ of the new chain can be generated using the same vertex $v_t$ as in the original $(\X_t)_{t=0}^T$ chain, and a random $Y_t(v_t)$ generated according to a coupling of the marginal distributions of $X_t(v_t)$ and $Y_t(v_t)$, conditioning respectively on the current states of the neighborhood of $v_t$ in $(\X_t)_{t=0}^T$ and $(\Y_t)_{t=0}^T$.
%
Note that these two marginal distributions must be identical unless (\textbf{I}) $\X_{t-1}$ and $\Y_{t-1}$ differ from each other over the neighborhood of $v_t$ or (\textbf{II}) the  $v_t$ itself is incident to where the models $\I$ and $\I'$ differ.
The event (\textbf{II}) occurs rarely due to the following reasons.
\begin{itemize}
\item For graph update, the event (\textbf{II}) occurs only if $v_t$ is incident to an updated edge. Since only $L$ edges are updated, the event occurs in at most  $O(TL/n)$ steps in expectation.
\item For Hamiltonian update, all the potentials of vertices and edges can be changed, thus $\I,\I'$ may differ everywhere. The key observation is that, as the total difference between the current and updated potentials is bounded by $L$, we can apply a filter to first select all candidate steps where the coupling may actually fail due to the difference between $\I$ and $\I'$,  which can be as small as $O(TL/n)$, and the actual coupling between $(\X_t)_{t=0}^{\infty}$  and $(\Y_t)_{t=0}^{\infty}$ is constructed with such prior.     
\end{itemize}
Finally, when $\I$ and $\I'$ both satisfy the Dobrushin-Shlosman condition, the percolation of disagreements between $(\X_t)_{t=0}^T$ and $(\Y_t)_{t=0}^T$ is bounded, and we show that the two chains are almost always identically coupled as $\Brac{v_{t},X_t(v_t)}=\Brac{v_{t},Y_t(v_t)}$, with exceptions at only $O(TL/n)$ steps.
The original chain $(\X_t)_{t=0}^T$ can then be updated to the new chain $(\Y_t)_{t=0}^T$ by only editing these $O(TL/n)$ local transitions $\Brac{v_{t},Y_t(v_t)}$ which are different from $\Brac{v_{t},X_t(v_t)}$.
This is aided by the dynamic data structure for the execution log of the chain, which is of independent interest.

\section{Dynamic Gibbs sampling}
\label{sec:agorithm-dynamic-Gibbs}
In this section, we give the dynamic sampling algorithm that updates the sample sequences.

In the following theorem, we use $\I = (V, E, Q, \Phi)$, where $n=|V|$,  to denote the current MRF instance and $\I'= (V', E', Q, \Phi')$, where $n'=|V'|$, to denote  the updated MRF instance. 
And define
\begin{align*}
\dgraph(\I,\I') 
&\triangleq |V \oplus V' | + |E\oplus E'|\\
\dham(\I,\I') 
&\triangleq \sum_{v\in V\cap V^{'}}\norm{\phi_v-\phi'_v}_1 + \sum_{e\in E\cap E^{'}}\norm{\phi_e- \phi'_e}_1.
\end{align*}
Note that $d(\I,\I')=\dgraph(\I,\I')+\dham(\I,\I')$, where $d(\I,\I')$ is defined in~\eqref{eq-def-dg-dh}.

\begin{theorem}[\textbf{dynamic sampling algorithm}]
\label{theorem-sample-MRF}
Let $N:\mathbb{N}^+ \to \mathbb{N}^+$ and $\epsilon: \mathbb{N}^+ \to (0,1)$ be two functions satisfying the bounded difference condition in Definition~\ref{definition-estimator-dynamic}.
Assume that  $\I$ and $\I'$ both satisfy Dobrushin-Shlosman condition,  $\dgraph(\I,\I')\leq \Lgra =o(n)$ and $\dham(\I,\I') \leq \Lham$.

There is an algorithm that maintains a sequence of $N(n)$ independent samples $\X^{(1)},\ldots,\X^{(N(n))} \in Q^V$ where $\DTV{\mu_{\I}}{\X^{(i)}}\leq \epsilon(n)$ for all $ 1\leq i \leq N(n)$,
using  $O\left(nN(n)\log n\right)$  memory words, each of $O(\log n)$ bits,
such that when $\I$ is updated  to $\I'$, the algorithm updates the sequence to $N(n')$ independent samples
$\Y^{(1)},\ldots,\Y^{(N(n'))}\in Q^{V'}$ where $\DTV{\mu_{\I'}}{\Y^{(i)}}\leq \epsilon(n')$ for all $ 1\leq i \leq N(n')$,
within expected time cost
\begin{align}
\label{eq-time-sampling}
O\left(\Delta^2( \Lgra + \Lham)N(n)\log ^ 3 n +\Delta n \right),
\end{align}
where $\Delta= \max\{\Delta_G,\Delta_{G'}\}$, and $\Delta_G,\Delta_{G'}$ denote the maximum degree of $G=(V,E)$ and $G'=(V',E')$.  
\end{theorem}

Our algorithm is based on the Gibbs sampling algorithm.
Let $N: \mathbb{N}^+\to \mathbb{N}^+$ and $\epsilon: \mathbb{N}^+\to (0,1)$ be two functions in \Cref{theorem-sample-MRF}.
We first give the \emph{single-sample dynamic Gibbs sampling algorithm} (\Cref{alg-dynamic-Gibbs}) that maintains a single sample $\X \in Q^V$ for the current MRF instance $\I=(V,E,Q,\Phi)$ where $n = |V|$ such that $\DTV{\X}{\mu_{\I}} \leq \epsilon(n)$.
We then use this algorithm to obtain the \emph{multi-sample dynamic Gibbs sampling algorithm} that maintains $N(n)$ independent samples for the current instance. 

Given the error function $\epsilon: \mathbb{N}^+\to (0,1)$,
suppose that $T(\I)$ is an easy-to-compute integer-valued function that upper bounds the mixing time on instance $\I$, such that 
\begin{align}
T(\I)\ge\tau_{\textsf{mix}}(\I,\epsilon(n)),\label{eq:mixing-time-bound}
\end{align}
where $\tau_{\textsf{mix}}(\I,\epsilon(n))$ denotes the mixing rate for the Gibbs sampling chain $(\X_t)_{t\ge 0}$ on instance $\I$.
By Proposition~\ref{proposition-mixing-Gibbs}, if the Dobrushin-Shlosman condition is satisfied, we can set
\begin{align}
\label{eq-upper-bound-T-mix}
T(\I) =  \left\lceil\frac{n}{\delta}\log \frac{n}{\epsilon(n)}\right\rceil.
\end{align}

Our algorithm for single-sample dynamic Gibbs sampling maintains a random process $(\X_t)_{t=0}^T$, which is a Gibbs sampling chain on instance $\I$ of length $T=T(\I)$, where $T(\I)$ satisfies~\eqref{eq:mixing-time-bound}.
Clearly $\X_T$ is a sample for $\mu_\I$ with $\DTV{\X_T}{\mu_{\I}}\le \epsilon(n)$.

When the current instance $\I$ is updated to a new instance $\I' =(V',E',Q,\Phi')$ where $n'=|V'|$, the original process $(\X_t)_{t=0}^T$ is transformed to a new process $(\Y_t)_{t=0}^{T'}$ such that the following holds as an invariant: $(\Y_t)_{t=0}^{T'}$ is a Gibbs sampling chain on $\I'$ with $T' = T(\I')$. Hence $\Y_T$ is a sample for the new instance $\I'$ with $\DTV{\Y_T}{\mu_{\I'}}\le \epsilon(n')$.
This is achieved through the following two steps:
\begin{enumerate}
\item
We construct couplings between $(\X_t)_{t=0}^{T}$ and $(\Y_t)_{t=0}^{T'}$, so that the new process $(\Y_t)_{t=0}^{T'}$ for $\I'$ can be obtained by making small changes to the original process $(\X_t)_{t=0}^{T}$ for $\I$.
\item 
We give a data structure which represents $(\X_t)_{t=0}^{T}$ incrementally and supports various updates and queries to $(\X_t)_{t=0}^{T}$ so that the above coupling can be generated efficiently.
\end{enumerate}


\subsection{Coupling for dynamic instances}\label{section-coupling}
%
The Gibbs sampling chain $(\X_t)_{t=0}^T$ can be uniquely and fully recovered from: the initial state $\X_0\in Q^V$, and the pairs $\exelog{X}{v_t}{T}$ that record the transitions. 
We call $\exelog{X}{v_t}{T}$ the \emph{execution-log} for the chain $(\X_t)_{t=0}^T$, and denote it with
\[
\Exelog(\I,T)\triangleq \exelog{X}{v_t}{T}.
\]
%

The following invariants are assumed for the random execution-log with an initial state.

\begin{condition}[\textbf{invariants for Exe-Log}]\label{exe-log-invariant}
Fixed an initial state $\X_0\in Q^V$, the followings hold for the random execution-log $\Exelog(\I,T)= \exelog{X}{v_t}{T}$ for the Gibbs sampling chain $(\X_t)_{t=0}^T$ on instance $\I=(V,E,Q,\Phi)$:
\begin{itemize}
\item $T=T(\I)$ where $T(\I)$ satisfies~\eqref{eq:mixing-time-bound};
\item the random process $(\X_t)_{t=0}^T$ uniquely recovered from the transitions $\exelog{X}{v_t}{T}$ and the initial state $\X_0$, is 
identically distributed as the Gibbs sampling (Algorithm~\ref{alg-Gibbs}) on instance $\I$ starting from initial state $\X_0$ with $v_t$ as the vertex picked at the $t$-th step.
\end{itemize}
\end{condition}
\noindent
Such invariants guarantee that $\X_T$ provides a sample for $\mu_{\I}$ with $\DTV{\X_T}{\mu_{\I}}\le \epsilon(|V|)$.

Suppose the current instance $\I$ is updated to a new instance $\I'$.
We construct couplings between the execution-log $\Exelog(\I,T)=\exelog{X}{v_t}{T}$ with initial state $\X_0\in Q^V$ for $\I$ and the execution-log $\Exelog(\I', T')=\exelog{Y}{v_t'}{T'}$ with initial state $\Y_0\in Q^{V'}$ for $\I'$.
Our goal is as follows: assuming Condition~\ref{exe-log-invariant}  for $\X_0$ and $\Exelog(\I,T)$, the same condition should hold invariantly for $\Y_0$ and $\Exelog(\I', T')$.

Unlike traditional coupling of Markov chains for the analysis of mixing time, where the two chains start from arbitrarily distinct initial states but proceed by the same transition rule, here the two chains $(\X_t)_{t=0}^T$ and $(\Y_t)_{t=0}^T$ start from similar states but have to obey different transition rules due to differences between instances $\I$ and $\I'$.

Due to the technical reason, we divide the update from $\I=(V,E,Q,\Phi)$ to $\I'=(V',E',Q,\Phi')$ into two steps: we first update $\I=(V,E,Q,\Phi)$ to 
\begin{align}
\label{eq-def-Imid}
\I_{\mathsf{mid}}=(V,E,Q,\Phi^{\mathsf{mid}}), 
\end{align}
where the potentials $\Phi^{\mathsf{mid}}=(\phi^{\mathsf{mid}}_a)_{a \in V \cup E}$ in the middle instance $\I_{\mathsf{mid}}$ are defined as
\begin{align*}
\forall a \in V \cup E,\quad  \phi^{\mathsf{mid}}_a \triangleq \begin{cases}
 \phi'_a &\text{if } a \in V' \cup E' \\
 \phi_a &\text{if } a \not\in V' \cup E';\end{cases}  
\end{align*}
then we update $\I_{\mathsf{mid}}=(V,E,Q,\Phi^{\mathsf{mid}})$ to  $\I'=(V',E',Q,\Phi')$. In other words, the update from $\I$ to $\Imid$ is only caused by updating the potentials of vertices and edges, while the underlying graph remains unchanged; and the update from $\Imid$ to $\I'$ is only caused by updating the underlying graph,~i.e. adding vertices, deleting vertices, adding edges and deleting edges.  

The dynamic Gibbs sampling algorithm can be outlined as follows.
\begin{itemize}
\item \UpdateHamiltonian: update $\*X_0$ and $\exelog{X}{v_t}{T}$ to a new initial state $\*Z_0$ and a new execution log $\Exelog(\Imid, T) =  \exelog{Z}{u_t}{T}$ such that the random process $(\*Z_t)_{t=0}^{T}$ 
	is the Gibbs sampling on instance $\Imid$.
\item \UpdateGraph: update $\*Z_0$ and $\exelog{Z}{u_t}{T}$ to a new initial state $\*Y_0$ and  a new execution log $\Exelog(\I',T) = \exelog{Y}{v_t'}{T}$ 
    such that the random process $(\*Y_t)_{t=0}^{T}$ 
    is the Gibbs sampling on instance $\I'$.
\item \LengthFix: change the length of the execution log $\exelog{Y}{v_t'}{T}$ from $T$ to $T'$, where $T' = T(\I')$ and $ T(\I')$ satisfies~\eqref{eq:mixing-time-bound}.
\end{itemize}
The dynamic Gibbs sampling algorithm is given in~\Cref{alg-dynamic-Gibbs}.

\begin{algorithm}[ht]
\SetKwInOut{Data}{Data}
\SetKwInOut{Update}{Update}
\Data{$\X_0\in Q^V$ and $\Exelog(\I,T)=\exelog{X}{v_t}{T}$ for current $\I=(V,E,Q,\Phi)$.}
\Update{an update that modifies $\I$ to $\I'=(V',E',Q,\Phi')$.}
compute $T'=T(\I')$ satisfying~\eqref{eq:mixing-time-bound} and construct $\Imid=(V',E',Q,\Phi^\mathsf{mid})$ as in~\eqref{eq-def-Imid}\;
$\left(\*Z_0, \exelog{Z}{u_t}{T}\right) \gets \UpdateHamiltonian \left(\I,\Imid, \*X_0, \exelog{X}{v_t}{T} \right)$\;\tcp{update the potentials: $\I \to \Imid$}
$\left(\*Y_0, \exelog{Y}{v_t'}{T}\right) \gets \UpdateGraph \left(\Imid,\I', \*Z_0, \exelog{Z}{u_t}{T} \right)$\;\tcp{update the underlying graph: $\Imid \to \I'$}
$\left(\*Y_0, \exelog{Y}{v_t'}{T'}\right) \gets \LengthFix \left(\I',\Y_0,\exelog{Y}{v'_t}{T},T'\right)$, where $T' = T(\I')$ \;
\tcp{change the length of the execution log from $T$ to $T'=T(\I')$}
update the data to $\Y_0$ and $\Exelog(\I', T')=\exelog{Y}{v'_t}{T'}$\;
\caption{Dynamic Gibbs sampling}\label{alg-dynamic-Gibbs}
\end{algorithm}
\begin{algorithm}[ht]
\SetKwInOut{Input}{Input}
\SetKwInOut{Output}{Output}
\SetKwInOut{Data}{Data}
\Data{$\X_0\in Q^V$ and $\Exelog(\I,T)=\exelog{X}{v_t}{T}$ for current $\I=(V,E,Q,\Phi)$.}
\Input{
the new length $T'>0$.}
\If{$T'<T$}{truncate $\exelog{X}{v_t}{T}$ to $\exelog{X}{v_t}{T'}$\;}
\Else{extend $\exelog{X}{v_t}{T}$ to $\exelog{X}{v_t}{T'}$ by simulating the Gibbs sampling chain on $\I$ for $T-T'$ more steps\;}
update the data to $\X_0$ and $\Exelog(\I, T')=\exelog{X}{v_t}{T'}$
\caption{$\textsf{LengthFix}\left(\I,\X_0,\exelog{X}{v_t}{T},T'\right)$}\label{alg-change-length}
\end{algorithm}

The subroutine \LengthFix{} is given in \Cref{alg-change-length}. We then describe $\UpdateHamiltonian$ (\Cref{section-update-ham}) and $\UpdateGraph$ (\Cref{section-update-graph}).

\subsubsection{Coupling for Hamiltonian update}
\label{section-update-ham}
We consider the update of changing potentials of vertices and edges.
The update do not change the underlying graph.
Let $\I = (V,E,Q,\Phi)$ be the current MRF instance. 
Let $\*X_0$ and $\exelog{X}{v_t}{T}$ be the current initial state and execution log such that the random process $(\X_t)_{t=0}^T$  is the Gibbs sampling on instance $\I$.
Upon such an update, the new instance becomes $\I'=(V,E,Q,\Phi')$. The algorithm $\UpdateHamiltonian(\I,\I',\*X_0, \exelog{X}{v_t}{T})$ updates the data to $\*Y_0$ and $\exelog{Y}{v_t'}{T}$ such that the random process $(\Y_t)_{t=0}^T$  is the Gibbs sampling on instance $\I'$.

We transform the pair of $\X_0\in Q^{V}$ and $\exelog{X}{v_t}{T}$ to a new pair of $\Y_0\in Q^{V}$ and $\exelog{Y}{v_t}{T}$ for $\I'$.
This is achieved as follows: the vertex sequence $(v_t)_{t=1}^{T}$ is identically coupled and the chain $(\X_t)_{t=0}^{T}$ is transformed to $(\Y_t)_{t=0}^{T}$ by the following one-step local coupling between $\X$ and $\Y$.

\begin{definition}[\textbf{one-step local coupling for Hamiltonian update}]\label{def-one-step-local-coupling-dynamic}
The two chains $(\X_t)_{t=0}^{\infty}$ on instance $\I=(V,E,Q,\Phi)$ and $(\Y_t)_{t=0}^{\infty}$ on instance $\I'=(V,E,Q,\Phi')$ are coupled as:
\begin{itemize}
\item Initially $\X_0=\Y_0 \in Q^V$; 
\item for $t=1,2,\ldots$, the two chains $\X$ and $\Y$ jointly do:
\begin{enumerate}
\item pick the same $v_t \in V$, and let $(X_t(u),Y_t(u)) \gets (X_{t-1}(u),Y_{t-1}(u))$ for all $u \in V \setminus \{v_t\}$;
\item sample $(X_t(v_t),Y_t(v_t))$ 
 from a coupling $D_{\I_{v_t},\I'_{v_t}}^{\sigma,\tau}(\cdot,\cdot)$
of the marginal distributions $\margin{\I}{v_t}(\cdot\mid \sigma)$ and $\margin{\I'}{v_t}(\cdot\mid \tau)$ with $\sigma=X_{t-1}(\Gamma_G({v_t}))$ and $\tau=Y_{t-1}(\Gamma_{G}({v_t}))$, where $G =(V,E)$. 
\end{enumerate}
\end{itemize}
\end{definition}
\noindent
The local coupling $D_{\I_{v},\I'_v}^{\sigma,\tau}(\cdot,\cdot)$ for Hamiltonian update is specified as follows.
\begin{definition}[\textbf{local coupling $D_{\I_{v},\I'_v}^{\sigma,\tau}(\cdot,\cdot)$ for Hamiltonian update}]
\label{definition-Ising-coupling-step}
Let $v \in V$ be vertex and $\sigma, \tau \in Q^{\Gamma_G(v)}$  two configurations, where $G=(V, E)$. We say a random pair $(c,c')\in Q^2$ is drawn from the coupling $D_{\I_{v},\I'_v}^{\sigma,\tau}(\cdot,\cdot)$ if $(c,c')$ is generated by the following two steps:
\begin{itemize}
\item \textbf{sampling step:} sample $(c,c') \in Q^2$ jointly from an optimal coupling $D^{\sigma,\tau}_{\opt,\I_v}$
of the marginal distributions $\mu_{v,\I}(\cdot \mid \sigma)$ and $\mu_{v, \I}(\cdot \mid \tau)$, such that $c \sim \mu_{v,\I}(\cdot \mid \sigma)$ and $c' \sim \mu_{v, \I}(\cdot \mid \tau)$;

\item \textbf{resampling step:} flip a coin independently with the probability of HEADS being
\begin{align}
\label{eq-correct-prob}
p_{\I_{v},\I'_{v}}^{\tau}(c') \triangleq \begin{cases}
0 & \text{if } \mu_{v,\I}(c' \mid \tau) \leq \mu_{v, \I'}(c' \mid \tau ),\\
\frac{\mu_{v,\I}(c' \mid \tau) - \mu_{v, \I'}(c' \mid \tau )}{\mu_{v, \I}(c' \mid \tau )} & \text{otherwise };
\end{cases}
\end{align}
if the outcome of coin flipping is HEADS, resample $c'$ from the distribution $\nu_{\I_v,\I'_v}^\tau$ independently, where the distribution $\nu_{\I_v,\I'_v}^\tau$ is defined as
 \begin{align}
 \label{eq-correct-dist}
\forall b\in Q:\quad \nu_{\I_{v},\I'_{v}}^{\tau}(b) \triangleq 
\frac{\max \left\{0, \mu_{v, \I'}(b \mid \tau ) - \mu_{v,\I}(b \mid \tau)\right\}}{\sum_{x\in Q}\max \left\{0, \mu_{v, \I}(x \mid \tau ) - \mu_{v,\I'}(x \mid \tau)\right\}}.
 \end{align}
\end{itemize}
\end{definition}
\begin{lemma}
\label{lemma-valid-coupling}
$D_{\I_v,\I'_v}^{\sigma,\tau}(\cdot,\cdot)$ in \Cref{definition-Ising-coupling-step} is a valid coupling between $\mu_{v, \I}(\cdot\mid \sigma)$ and $\mu_{v, \I'}(\cdot \mid \tau)$.
\end{lemma}
By \Cref{lemma-valid-coupling}, the resulting $(\Y_t)_{t=0}^T$ is a faithful copy of the Gibbs sampling on instance $\I'$, assuming that $(\X_t)_{t=0}^T$ is such a chain on instance $\I$.

Next we give an upper bound for the probability $p_{\I_{v},\I'_{v}}^{\tau}(\cdot)$ defined in~\eqref{eq-correct-prob}.
\begin{lemma}
\label{lemma-correct-up-bound}
For any two instances $\I=(V,E,Q,\Phi)$ and $\I'=(V,E,Q,\Phi')$ of MRF model, and any $v\in V,c\in Q$ and $\sigma \in Q^{\Gamma_G(v)}$, it holds that
\begin{align}
\label{eq-linear-decay}
p_{\I_{v},\I'_{v}}^{\tau}(c) \leq   2\left( \Vert\phi_v-\phi'_v\Vert_1 + \sum_{e=\{u,v\}\in E}\Vert \phi_e-\phi'_e \Vert_1 \right),
\end{align}
where $\Vert\phi_v-\phi'_v\Vert_1 = \sum_{c \in Q}|\phi_v(c)-\phi'_v(c)|$ and $\Vert \phi_e-\phi'_e \Vert_1 = \sum_{c,c' \in Q}|\phi_e(c,c')-\phi'_e(c,c')|$.
\end{lemma}
\noindent
By Lemma~\ref{lemma-correct-up-bound}, for each vertex $v \in V$, we define an upper bound of the probability $p^{\cdot}_{\I_v,\I'_v}(\cdot)$ as
\begin{align}
\label{eq-def-Ising-up}
\pup_v \triangleq \min\left\{2\left( \Vert\phi_v-\phi'_v\Vert_1 + \sum_{e=\{u,v\}\in E}\Vert \phi_e-\phi'_e \Vert_1 \right),1\right\}. 
\end{align}

With $\pup_v$, we can implement the one-step local coupling in \Cref{def-one-step-local-coupling-dynamic} as follows.
We first sample each $v_i \in V$ for $1\leq i \leq T$ uniformly and independently.
For each vertex $v \in V$, let $T_v \triangleq\{1\leq t \leq T \mid v_t = v \}$ be the set of all the steps that pick the vertex $v$.
We select each $t \in T_v$ independently with probability $\pup_v$ to construct a random subset $\Isingset_v \subseteq T_v$, and let 
\begin{align}
\label{eq-def-Ising-set}
  \Isingset \triangleq \bigcup_{v \in V}\Isingset_v.
 \end{align}
We then couple the two chains $(\X_t)_{t=0}^T$ and $(\Y_t)_{t=0}^T$. 
First set $\X_0 = \Y_0$.
For each $1\leq t \leq T$, we set  $(X_t(u),Y_t(u)) \gets (X_{t-1}(u),Y_{t-1}(u))$ for all $u \in V \setminus \{v_t\}$; 
then generate the random pair $(X_t(v_t),Y_t(v_t))$ by the following procedure. 

\begin{itemize}
\item \textbf{sampling step:} Let  $\sigma = X_{t-1}(\Gamma_G(v_t))$ and $\tau = Y_{t-1}(\Gamma_G(v_t))$. 
	We draw a random pair $(c,c') \in Q^2$ from the optimal coupling $D^{\sigma,\tau}_{\opt,\I_v}$ of the marginal distributions $\mu_{v,\I}(\cdot\mid \sigma)$ and $\mu_{v,\I}(\cdot\mid \tau )$ such that $c \sim \mu_{v,\I}(\cdot \mid \sigma)$ and $c' \sim \mu_{v,\I}(\cdot\mid \tau)$;

\item \textbf{resampling step:} If $t \notin \Isingset$, set $X_t(v_t) = c$ and $Y_t(v_t) = c'$.
Otherwise, set $X_t(v_t) = c$ and 
\begin{align}
\label{eq-couple-resample}
Y_t(v_t) = \begin{cases}
 b \sim \nu_{\I_{v_t},\I'_{v_t}}^{\tau} &\text{with probability } p^\tau_{\I_{v_t},\I'_{v_t}}(c') / \pup_{v_t}\\
 c' &\text{with probability } 1-p^\tau_{\I_{v_t},\I'_{v_t}}(c')/\pup_{v_t}.
 \end{cases}
 \end{align}
\end{itemize}
Note that $\pup_{v_t} > 0$ if $t \in \Isingset$.
By Lemma~\ref{lemma-correct-up-bound}, it must hold that $p^\tau_{\I_{v_t},\I'_{v_t}}(c') \leq \pup_{v_t}$. Hence, the probability $p^\tau_{\I_{v_t},\I'_{v_t}}(c') / \pup_{v_t}$ is valid.
Note that the probability that $Y_t(v_t)$ is set as $b$ is
\begin{align*}
\Pr[Y_t(v_t) \text{ is set as } b] = \Pr\left[ t \in \Isingset  \right] \cdot \frac{p^\tau_{\I_{v_t},\I'_{v_t}}(c')}{\pup_{v_t}} =  \pup_{v_t} \cdot \frac{p^\tau_{\I_{v_t},\I'_{v_t}}(c')}{\pup_{v_t}} = p^\tau_{\I_{v_t},\I'_{v_t}}(c').
\end{align*}
Hence, our implementation perfectly simulates the coupling in Definition~\ref{def-one-step-local-coupling-dynamic}.

Let $\D_t$ denote the \emph{set of disagreements} between $\X_t$ and $\Y_t$. Formally,
\begin{align}\label{eq:disagreement-set}
\D_t\triangleq\{v \in V \mid X_t(v) \neq Y_t(v)\}.
\end{align}
Note that if $v_t \notin \Gamma_G(\D_{t-1})$, the random pair $(c,c')$ drawn from the coupling $D^{\sigma,\tau}_{\opt,\I_v}$ must satisfy $c = c'$.
Thus it is easy to make the following observation for the $(\X_t)_{t=0}^{T}$ and $(\Y_t)_{t=0}^{T}$ coupled as above.

\begin{observation}
\label{observation-coupling-Ising}
For any integer $t \in [1,T]$, if $v_t \notin \Gamma_G^+(\D_{t-1})$ and $t \notin \Isingset$, then $X_t(v_t) = Y_t(v_t)$ and $\D_{t} = \D_{t-1}$.  
\end{observation}
With this observation, the new $\Y_0$ and $\Exelog(\I',T) = \exelog{Y}{v_t}{T} $ can be generated from $\X_0$ and $\Exelog(\I,T) = \exelog{X}{v_t}{T}$ as Algorithm~\ref{alg-ham-update}.
\begin{algorithm}[h]
\SetKwInOut{Data}{Data}
\SetKwInOut{Update}{Update}
\SetKwIF{Wp}{}{}{with probability}{do}{}{}{}
\Data{$\X_0\in Q^V$ and $\Exelog(\I,T)=\exelog{X}{v_t}{T}$ for $\I=(V,E,Q,\Phi)$.}
\Update{an update that modifies $\I$ to $\I'=(V,E,Q,\Phi')$.}
$t_0\gets 0$, $\D\gets \emptyset$, and {construct a $\boldsymbol{Y}_0 \gets \boldsymbol{X}_0$}\;
for each $v \in V$, construct a random subset $\Isingset_v \subseteq T_v \triangleq \{1\leq t \leq T \mid v_t = v \}$ such that each element in $T_v$ is selected independently with probability $\pup_v$ defined in~\eqref{eq-def-Ising-up}\label{line-array-1}\;
construct the set  $\Isingset \gets \bigcup_{v \in V}\Isingset_v$\label{line-array-2}\;
\While{$\exists\,t_0<t\le T$ such that $v_{t} \in \Gamma_G^+(\D)$ or $t \in \Isingset$}{
find the smallest $t > t_0$ such that $v_{t} \in \Gamma_G^+(\D)$ or $t \in \Isingset$\label{line-next-ham}\;
for all $t_0< i < t$, let $Y_i(v_i)=X_i(v_i)$\;
sample $Y_t(v_t) \in Q$ conditioning on $X_{t}(v_t)$ according to the optimal coupling between $\mu_{v_t,\I}(\cdot \mid X_{t-1}(\Gamma_G(v_t)))$ and $\mu_{v_t,\I}(\cdot \mid Y_{t-1}(\Gamma_G(v_t)))$\label{line-sample-ham-1}\;
\If{$t \in \Isingset$\label{line-flip-1}}{
\Wp{$p^\tau_{\I_{v_t},\I'_{v_t}}(Y_t(v_t)) / \pup_{v_t}$ where $\tau = Y_{t-1}(\Gamma_G(v_t))$\label{line-flip-2} }{
resample $Y_t(v_t) \sim \nu_{\I_{v_t},\I'_{v_t}}^{\tau}$, where $\nu_{\I_{v_t},\I'_{v_t}}^{\tau}$ is defined in~\eqref{eq-correct-dist} \label{line-flip-ham} \;
}
}
\textbf{if} {$X_t(v_t)\neq Y_t(v_t)$} \textbf{then} {$\D\gets \D\cup\{v_t\}$} \textbf{else} {$\D\gets \D\setminus\{v_t\}$}\;
$t_0\gets t$\;
}
for all remaining $t_0< i\le T$: let  
$Y_i(v_i)=X_i(v_i)$\;
update the data to $\Y_0$ and $\Exelog(\I',T)=\exelog{Y}{v_t}{T}$\;
\caption{$\UpdateHamiltonian\left(\I,\I',\*X_0, \exelog{X}{v_t}{T} \right)$}\label{alg-ham-update}
\end{algorithm}

Observation~\ref{observation-coupling-Ising} says that the nontrivial coupling between $X_t(v_t)$ and $Y_t(v_t)$ is only needed when $v_t\in \Gamma_G^+(\D_{t-1})$ or $t \in \Isingset$, which occurs rarely as long as $\D_{t-1}$ and $\Isingset$ are small. 
This is a key to ensure the small incremental time cost of Algorithm~\ref{alg-ham-update}.
For the $(\X_t)_{t=0}^{T}$ and $(\Y_t)_{t=0}^{T}$ coupled as above and any $1\leq t \leq T$, 
let $\gamma_t$ indicate whether the event $t\in  \Isingset \lor v_t\in\Gamma_G^+(\D_{t-1})$ occurs:

\begin{align}
\label{label-gamma-Ising}
\gamma_t \triangleq \one{t\in  \Isingset \lor v_t\in\Gamma_G^+(\D_{t-1})},
\end{align}
and $\Rham$ denote the number of occurrences of such bad events:
\begin{align}
\label{eq-def-R-Ising}
\Rham \triangleq \sum_{t=1}^{T}\gamma_t.
\end{align}
The following lemma bounds the expectation of $\Rham$.
\begin{lemma}[\textbf{cost of the coupling for \UpdateHamiltonian}]\label{lemma-upper-bound-R-Ham}
Let $\I=(V,E,Q,\Phi)$ be the current MRF instance and $\I'=(V,E,Q,\Phi')$ the updated instance.
Assume that $\I$ satisfies Dobrushin-Shlosman condition (\Cref{condition-Dobrushin}) with constant $\delta>0$, and $\dham(\I,\I') =  \sum_{v\in V}\norm{\phi_v-\phi'_v}_1 + \sum_{e\in E}\norm{\phi_e- \phi'_e}_1 \leq L$.  
It holds that $\E{\Rham} = O\left( \frac{\Delta T L}{n\delta}  \right)$,
where $n= |V|$, $\Delta$ is the maximum degree of graph $G=(V,E)$.
\end{lemma}

\subsubsection{Coupling for graph update}
\label{section-update-graph}
Let $\I = (V,E,Q,\Phi)$ be an MRF instance, where $\Phi=(\phi_a)_{a \in V \cup E}$. 
Let $\*X_0$ and $\exelog{X}{v_t}{T}$ be the current initial state and execution log such that the random process $(\X_t)_{t=0}^T$  is the Gibbs sampling on instance $\I$.
Let $\I'=(V',E',Q,\Phi')$ be the new instance obtained by updating the underlying graph, where $\Phi'=(\phi_a)_{a \in V' \cup E'}$ satisfies 
\begin{align*}
\forall a \in 	(V \cap V') \cap (E \cap E'), \quad \phi_a = \phi'_a.
\end{align*}
Given the update from $\I$ to $\I'$,
the subroutine $\UpdateGraph\left(\I,\I',\*X_0, \exelog{X}{v_t}{T} \right)$ updates the data to a new initial state $\*Y_0$ and a new execution-log $\exelog{Y}{v'_t}{T}$ such that the random process $(\*Y_t)_{t=0}^{T}$ 
 is the Gibbs sampling on instance $\I'$.
 
The subroutine \UpdateGraph{} does as the following three steps. 
 \begin{itemize}
 \item \AddVertex: add isolated vertices in $V' \setminus V$ with potentials $(\phi_v)_{v \in V' \setminus V}$, and update the instance $\I = (V,E,Q,\Phi)$ to a new instance 
 \begin{align}
 \label{eq-def-I-1}
 \I_1 = \I_1(\I,\I') \triangleq \left(V \cup V', E , Q, \Phi \cup (\phi_v)_{v \in V' \setminus V}\right);
 \end{align}
then update $\X_0$ and $ \exelog{X}{v_t}{T}$ to $\*Z^{}_0$ and $\Exelog(\I_1,T) = \exelog{Z^{}}{u_t}{T}$ such that the random process $(\*Z^{}_t)_{t=0}^{T}$ 
	is the Gibbs sampling on instance $\I_1$.
 \item \UpdateEdge: add new edges in $E' \setminus E$ with potentials $(\phi_e)_{e \in E' \setminus E}$, delete edges in $E \setminus E'$ , and update the instance $\I_1$ to a new instance
  \begin{align}
  \label{eq-def-I-2}
 \I_2 = \I_2(\I,\I')	&\triangleq \left(V \cup V', E' , Q, \Phi \cup (\phi_v)_{v \in V' \setminus V}  \cup (\phi_e)_{e \in E' \setminus E} \setminus (\phi_e)_{e \in E \setminus E'} \right)\notag\\
 &= \left(V \cup V', E' , Q, \Phi' \cup (\phi_v)_{v \in V \setminus V'}\right);
 \end{align}
  then update $\*Z^{}_0$ and $\exelog{Z^{}}{u_t}{T}$ to $\*Z^{'}_0$ and $\Exelog(\I_2,T) =\exelog{Z^{'}}{w_t}{T}$ such that the random process $(\*Z^{'}_t)_{t=0}^{T}$ 
	is the Gibbs sampling on instance $\I_2$.
\item \DeleteVertex: delete isolated vertices in $V \setminus V'$, and update the instance $\I_2$ to $\I'=(V',E',Q,\Phi')$; then update $\*Z^{'}_0$ and $\exelog{Z^{'}}{w_t}{T}$ to $\*Y_0$ and $\Exelog(\I',T) = \exelog{Y}{v'_t}{T}$ such that the random process $(\*Y_t)_{t=0}^{T}$ is the Gibbs sampling on instance $\I'$.
 \end{itemize}
The algorithm \UpdateGraph{} is given in \Cref{alg-update-graph}.
 
 \begin{algorithm}[ht]
\SetKwInOut{Data}{Data}
\SetKwInOut{Update}{Update}
\Data{$\X_0\in Q^V$ and $\Exelog(\I,T)=\exelog{X}{v_t}{T}$ for current $\I=(V,E,Q,\Phi)$.}
\Update{an update of the underlying graph that modifies $\I$ to $\I'=(V',E',Q,\Phi')$.}
construct instances $\I_1$ and $\I_2$ as in~\eqref{eq-def-I-1} and~\eqref{eq-def-I-2}\;
$\left(\*Z_0, \exelog{Z}{u_t}{T}\right) \gets \AddVertex \left(\I, \I_1 ,\*X_0, \exelog{X}{v_t}{T} \right)$\label{alg-graph-add-vtx1}\;
\tcp{add isolated vertices to update $\I$ to $\I_1$}
$\left(\*Z'\*_0, \exelog{Z'}{w_t}{T}\right) \gets \UpdateEdge \left(\I_1, \I_2,\*Z_0, \exelog{Z}{u_t}{T} \right)$\;
\tcp{add and delete edges to update $\I_1$ to $\I_2$}
$\left(\*Y_0, \exelog{Y}{v_t'}{T}\right) \gets \DeleteVertex \left(\I_2,\I', \*Z'\*_0, \exelog{Z'}{w_t}{T}  \right)$\;\tcp{delete isolated vertices to update $\I_2$ to $\I'$}
update the data to $\Y_0$ and $\Exelog(\I')=\exelog{Y}{v'_t}{T}$\;
\caption{$\UpdateGraph\left(\I,\I',\*X_0, \exelog{X}{v_t}{T} \right)$}\label{alg-update-graph}
\end{algorithm}


The subroutines \AddVertex{} and \DeleteVertex{} are simple, because they only deal with isolated variables. We first describe the main subroutine \UpdateEdge{}, then describe \AddVertex{} and \DeleteVertex{}.

\vspace{1em}
\paragraph{The coupling for \UpdateEdge}
We first consider the update of adding and deleting edges.
The update does not change the set of variables.
Let $\I = (V,E,Q,\Phi)$ be the current MRF instance.
Let $\*X_0$ and $\exelog{X}{v_t}{T}$ be the current initial state and execution log such that the random process $(\X_t)_{t=0}^T$  is the Gibbs sampling on instance $\I$.
Upon such an update, the new instance becomes $\I'=(V,E',Q,\Phi')$, where $\phi'_a = \phi_a$ for all $a \in V \cup (E \cap E')$. The subroutine $\UpdateEdge(\I,\I',\*X_0, \exelog{X}{v_t}{T})$ updates the data to $\*Y_0$ and $\exelog{Y}{v_t'}{T}$ such that the random process $(\Y_t)_{t=0}^T$  is the Gibbs sampling on instance $\I'$.
We use $\mathcal{S}\subseteq V$ to denote the set of vertices affected by the update from $\I$ to $\I'$:
\begin{align}
\label{eq-def-S}
\mathcal{S} \triangleq \bigcup_{(u,v) \in E \oplus E'}\{u,v\},
\end{align}
where $E \oplus E'$ is the symmetric difference between $E$ and $E'$.

We transform this pair of $\X_0\in Q^{V}$ and $\exelog{X}{v_t}{T}$ to a new pair of $\Y_0\in Q^{V}$ and $\exelog{Y}{v_t}{T}$ for $\I'$.
This is achieved as follows: the vertex sequence $(v_t)_{t=1}^{T}$ is identically coupled and the chain $(\X_t)_{t=0}^{T}$ is transformed to $(\Y_t)_{t=0}^{T}$ by the following one-step local coupling between $\X$ and $\Y$.
\begin{definition}[\textbf{one-step local coupling for \UpdateEdge}]\label{def-one-step-local-coupling-graph}
The two chains $(\X_t)_{t=0}^{\infty}$ on instance $\I = (V,E,Q,\Phi)$ and $(\Y_t)_{t=0}^{\infty}$ on instance $\I' = (V,E',Q,\Phi')$ are coupled as:
\begin{itemize}
\item Initially $\X_0=\Y_0 \in Q^V$; 
\item for $t=1,2,\ldots$, the two chains $\X$ and $\Y$ jointly do:
\begin{enumerate}
\item pick the same $v_t \in V$, and let $(X_t(u),Y_t(u)) \gets (X_{t-1}(u),Y_{t-1}(u))$ for all $u \in V \setminus \{v_t\}$;
\item sample $(X_t(v_t),Y_t(v_t))$ 
 from a coupling $D_{\I_{v_t},\I'_{v_t}}^{\sigma,\tau}(\cdot,\cdot)$
of the marginal distributions $\margin{\I}{v_t}(\cdot\mid \sigma)$ and $\margin{\I'}{v_t}(\cdot\mid \tau)$ with $\sigma=X_{t-1}(\Gamma_G({v_t}))$ and $\tau=Y_{t-1}(\Gamma_{G'}({v_t}))$, where $G=(V,E)$  and $G'=(V,E')$. 
\end{enumerate}
\end{itemize}
\end{definition}
\noindent
The local coupling $D_{\I_{v},\I'_{v}}^{\sigma,\tau}(\cdot,\cdot)$ for \UpdateEdge~is specified as follows.
\begin{align}
\forall \sigma\in Q^{\Gamma_G(v)},\tau\in Q^{\Gamma_{G'}(v)}:\quad
D_{\I_{v},\I'_{v}}^{\sigma,\tau}(\cdot,\cdot)=\begin{cases}
D_{\opt,\I_v}^{\sigma,\tau}(\cdot,\cdot) & \text{if } v\not\in\mathcal{S},\\
\mu_{v,\I}(\cdot\mid \sigma)\times \mu_{v,\I'}(\cdot\mid \tau) & \text{if }v\in\mathcal{S},
\end{cases}
\label{eq:local-coupling-D-for-R}
\end{align}
where $D^{\sigma,\tau}_{\opt,\I_{v}}$ is an optimal coupling of marginal distributions $\mu_{v,\I}(\cdot\mid \sigma)$ and $\mu_{v,\I}(\cdot\mid \tau)$.
Recall $\I_v = (\Gamma^+_v,E_v,Q,\Phi_v)$ where $E_v=\{\{u,v\}\in E\}$ and $\Phi_v=(\phi_a)_{a\in \Gamma_v^+\cup E_v}$.
Obviously, $D_{\I_{v},\I'_{v}}^{\sigma,\tau}$ is  a valid coupling of $\mu_{v,\I}(\cdot\mid \sigma)$ and $\mu_{v,\I'}(\cdot\mid \tau)$. Because for any $v\not\in\mathcal{S}$, we have $\I_v=\I_{v'}$ and hence $\mu_{v,\I}(\cdot\mid \sigma)$ and $\mu_{v,\I'}(\cdot\mid \tau)$ are the same, both defined by ~\eqref{eq-marginal} on $\I_v$. Thus they can be coupled by $D_{\opt,\I_v}^{\sigma,\tau}$. 

%
%

Obviously  the resulting $(\Y_t)_{t=0}^{T}$ is a faithful copy of the Gibbs sampling on instance $\I'$, assuming that $(\X_t)_{t=0}^{T}$ is such a chain on instance $\I$.

Recall $\D_t\triangleq \{v \in V \mid X_t(v) \neq Y_t(v) \}$ is set of disagreements between $\X_t$ and $\Y_t$. 
The following observation is easy to make for the $(\X_t)_{t=0}^{T}$ and $(\Y_t)_{t=0}^{T}$ coupled as above. 
\begin{observation}\label{observation:skip}
For any $t\in[1, T]$, if $v_t\not\in \mathcal{S}\cup\Gamma_G^+(\D_{t-1})$ then $\X_t(v_t)=\Y_t(v_t)$ and $\D_t=\D_{t-1}$.
\end{observation}

With this observation, the new $\Y_0$ and $\Exelog(\I',T)=\exelog{Y}{v_t}{T}$ can be generated from $\X_0$ and $\Exelog(\I, T)=\exelog{X}{v_t}{T}$ as in Algorithm~\ref{alg-constraint-update}. 


\begin{algorithm}[ht]
\SetKwInOut{Data}{Data}
\SetKwInOut{Update}{Update}
\Data{$\X_0\in Q^V$ and $\Exelog(\I,T)=\exelog{X}{v_t}{T}$ for current $\I=(V,E,Q,\Phi)$.}
\Update{an update of adding and deleting edges that modifies $\I$ to $\I'=(V,E',Q,\Phi')$.}
$t_0\gets 0$, $\D\gets \emptyset$, $\Y_0 \gets \X_0$ and construct  $\mathcal{S} \gets \bigcup_{(u,v) \in E \oplus E'}\{u,v\}$ \label{step-init-2}\;
\While{$\exists\,t_0<t\le T$ such that $v_{t} \in  \mathcal{S} \cup \Gamma_G^+(\D)$\label{step-while-condition}}{
find the smallest $t > t_0$ such that $v_{t} \in  \mathcal{S} \cup  \Gamma_G^+(\D)$\label{step-find-smallest-t}\;
for all $t_0< i < t$, let $Y_i(v_i)=X_i(v_i)$\label{step-assign-mid-spin}\;
sample $Y_{t}(v_t)$ conditioning on $X_{t}(v_t)$ according to the coupling $D_{v_t}^{\sigma,\tau}(\cdot,\cdot)$ (constructed in \eqref{eq:local-coupling-D-for-R}), where $\sigma=X_{t-1}(\Gamma_G({v_{t}}))$ and $\tau=Y_{t-1}(\Gamma_{G'}({v_{t}}))$\label{step-optimal-coupling}\;
\textbf{if} {$X_t(v_t)\neq Y_t(v_t)$} \textbf{then} {$\D\gets \D\cup\{v_t\}$} \textbf{else} {$\D\gets \D\setminus\{v_t\}$}\label{step-update-D}\;
$t_0\gets t$\;
}
for all remaining $t_0< i\le T$: let  
$\Y_i(v_i)=X_i(v_i)$\label{step-assign-end-spin}\;
update the data to $\Y_0$ and $\Exelog(\I',T)=\exelog{Y}{v_t}{T}$\label{step-push-update}\;
\caption{$\UpdateEdge(\I,\I',\*X_0, \exelog{X}{v_t}{T})$}\label{alg-constraint-update}
\end{algorithm}

Observation~\ref{observation:skip} says that the nontrivial coupling between $X_t(v_t)$ and $Y_t(v_t)$ is only needed when $v_t\in \mathcal{S}\cup\Gamma_G^+(\D_{t-1})$, which occurs rarely as long as $\D_{t-1}$ remains small. 
This is a key to ensure the small incremental time cost of Algorithm~\ref{alg-constraint-update}.
Formally, for the $(\X_t)_{t=0}^{T}$ and $(\Y_t)_{t=0}^{T}$ coupled as above, for any $1\leq t \leq T$, let $\gamma_t$ indicate whether this bad event occurs:
\begin{align}
\label{label-gamma}
\gamma_t \triangleq \one{v_t \in \mathcal{S} \cup \Gamma_{G}^+(\D_{t-1})},
\end{align}
and let $\Rgra$ denote the number of occurrences of such bad events:
\begin{align}
\label{eq-def-R}
\Rgra \triangleq \sum_{t=1}^{T}\gamma_t.
\end{align}
We will see that $\Rgra$ dominates the cost of Algorithm~\ref{alg-constraint-update}, 
once a data structure is given to encode the execution-log and resolve the updates in Line~\ref{step-push-update} and various queries (in Lines~\ref{step-while-condition}, \ref{step-find-smallest-t} and~\ref{step-optimal-coupling}) to the data.
%
%
%



\begin{lemma}[\textbf{cost of the coupling for \UpdateEdge}]\label{lemma-upper-bound-R}
Let $\I=(V,E,Q,\Phi)$ be the current MRF instance and $\I'=(V,E',Q,\Phi')$ the updated instance.
Assume that $\I'$ satisfies Dobrushin-Shlosman condition (\Cref{condition-Dobrushin}) with constant $\delta > 0$, and $|E \oplus E'| \leq L$.  
It holds that $\E{\Rgra} = O\left( \frac{\Delta T L}{n\delta}  \right)$,
where $n= |V|$, $\Delta= \max\{\Delta_G,\Delta_{G'}\}$, and $\Delta_G,\Delta_{G'}$ denote the maximum degree of $G=(V,E)$ and $G'=(V,E')$.
\end{lemma}

\vspace{1em}
\paragraph{Coupling for \AddVertex}
Let $\I = (V,E,Q,\Phi)$ be the current MRF instance.
Let $\*X_0$ and $\exelog{X}{v_t}{T}$ be the current initial state and execution log such that the random process $(\X_t)_{t=0}^T$  is the Gibbs sampling on instance $\I$.
The update adds a set of \emph{isolated} vertices $S$ with potentials $(\phi_a)_{a \in S}$.
Upon such an update, the new instance becomes
\begin{align*}
\I'=(V',E, Q, \Phi') = (V \cup S, E, Q, \Phi \cup (\phi_a)_{a \in S}).	
\end{align*}
The subroutine $\AddVertex(\I,\I',\*X_0, \exelog{X}{v_t}{T})$ updates the data to $\*Y_0$ and $\exelog{Y}{v_t'}{T}$ such that the random process $(\Y_t)_{t=0}^T$  is the Gibbs sampling on instance $\I'$.
%

Since the new instance $\I'$ is the same as $\I$ except the isolated vertices in $S$, 
we can construct $\Y_0(V)=\X_0$ and $\Y_0(S)\in  Q^S$ is arbitrary, and $\Exelog(\I',T)=\exelog{Y}{v_t'}{T}$ can be constructed by inserting random appearances of vertices in $S$ into $(v_t)_{t=1}^{T}$, 
while for any $v \in S$, the $Y_t(v)$ at the inserted steps $t$ are sampled i.i.d.~from the marginal distribution $\margin{\I'}{v}(\cdot)$, which is just a distribution over $Q$ proportional to $\exp(\phi_v(\cdot))$ in the case of Gibbs sampling, since $v$ is an isolated vertex.
Let $[T]\triangleq\{1,2,\ldots,T\}$.
Formally:
\begin{enumerate}
\item
Let $P\subseteq [T]$ be a random subset such that each $t\in[T]$ is selected into $P$ independently with probability $\frac{|S|}{|S \cup V|}$. 
Let $h = |P|$ and enumerate all elements in $P$ as $r_1 < r_2 <\ldots < r_h$.
Let $m =  T-h$ and enumerate all elements in $[T]\setminus P$ as $\ell_1<\ell_2<\cdots<\ell_m$.
\item
For each $1\leq i \leq h$, sample $u_i \in S$ uniformly and independently.
\item
Let $\exelog{X}{v_t}{m}\gets\textsf{LengthFix}\left(\I,\X_0,\exelog{X}{v_t}{T},m\right)$.
\item
Construct $\exelog{Y}{v_t'}{T'}$ as follows:
\begin{align*}
\forall\, t = r_k \in P&:
\quad v'_t = u_k
\quad \text{and }\quad 
Y_t(v_t')\sim\margin{\I'}{u_k}(\cdot), \text{ where } \margin{\I'}{u_k}(c) \propto \exp(\phi_{u_k}(c));\\
\forall\, t = \ell_k\in [T'] \setminus P&:
\quad v'_t = v_{k}
\quad \text{and }\quad 
Y_t(v_t') = X_k(v_{t}')=X_k(v_k).
\end{align*}
\end{enumerate}

It is easy to see that $(\Y_t)_{t=0}^{T'}$ is a faithful copy of the Gibbs sampling on instance $\I'$.

\vspace{1em}
\paragraph{Coupling for \DeleteVertex}
Let $\I = (V,E,Q,\Phi)$ be the current MRF instance. The update deletes 
a set of \emph{isolated} variables $S \subseteq V$.
Let $\*X_0$ and $\exelog{X}{v_t}{T}$ be the current initial state and execution log such that the random process $(\X_t)_{t=0}^T$  is the Gibbs sampling on instance $\I$.
Upon such update, the instance is updated to $\I'=(V',E,Q,\Phi')$, where $V'=V\setminus S$ and $\Phi'=\Phi\setminus(\phi_{v})_{v \in S}$. 
The subroutine $\DeleteVertex(\I,\I',\*X_0, \exelog{X}{v_t}{T})$ updates the data to $\*Y_0$ and $\exelog{Y}{v_t'}{T}$ such that the random process $(\Y_t)_{t=0}^T$  is the Gibbs sampling on instance $\I'$.
%

We can simply construct $\Y_0=X_0(V')$.
The new execution-log $\Exelog(\I',\epsilon)=\exelog{Y}{v_t'}{T}$ can be constructed from the original $\Exelog(\I,T)=\exelog{X}{v_t}{T}$ by simply deleting all appearances of vertices $v \in S$ in $(v_t)_{t=1}^T$ and the corresponding trivial transitions  $X_t(v)$, followed by calling $\textsf{LengthFix}$ on instance $\I'$ to properly append the chain to the length $T$.

It is easy to see that $(\Y_t)_{t=0}^{T}$ is a faithful copy of the Gibbs sampling on instance $\I'$.

\subsection{Data structure for Gibbs sampling}
\label{section-DS}
We now describe an efficient data structure for Gibbs sampling $(\X_t)_{t=0}^T$.
Let $\I = (V, E, Q, \Phi)$ be an MRF instance.
The data structure should provide the following functionalities.

\begin{itemize}
\item \textbf{Data:} an initial state $\X_0\in Q^V$ and an execution-log $\exelog{X}{v_t}{T}\in (V\times  Q)^T$ that records the $T$ transitions of the Gibbs sampling $(\X_t)_{t=0}^T$;
\item \textbf{updates:}   
\begin{itemize}
\item $\datainsert(t,v,c)$, which inserts a transition $\Brac{v,c}$ after the $(t-1)$-th transition $\Brac{v_{t-1},X_{t-1}(v_{t-1})}$;
\item $\dataremove(t)$, which deletes the $t$-th transition $\Brac{v_t,X_t(v_t)}$;
\item $\dataupdate(t,c)$, which changes the $t$-th transition $\Brac{v_t,X_t(v_t)}$ to $\Brac{v_t, c}$;
\end{itemize}
Note that the updates $\datainsert(t,v,c)$ and $\dataremove(t)$ change the length $T$ of the chain, as well as the order-numbers of all transitions after the inserted/deleted transition.
\item \textbf{queries:}  
\begin{itemize}
\item
$\dataeval(t,v)$, which returns the value of $X_t(v)$  for arbitrary $t$ and $v$ (not necessarily $=v_t$);
\item
$\datasuccessor(t,v)$, which returns $i$ for the smallest $i>t$ such that $v_i=v$ if such $i$ exists, or returns $\perp$ if otherwise.
\end{itemize}
\end{itemize}

It is not difficult to realize that the query $\dataeval(t,v)$ can actually be solved by a predecessor search defined symmetrically to $\datasuccessor(t,v)$.
This data structure problem for Gibbs sampling is quite natural and is of independent interest.

\begin{theorem}[\textbf{data structure for Gibbs sampling}]
\label{theorem-DS}
There exists a deterministic dynamic data structure which stores an arbitrary initial state $\X_0\in Q^V$ and an execution-log  $\exelog{X}{v_t}{T}\in (V\times  Q)^T$   for Gibbs sampling
using $O(T+|V|)$ memory words, each of $O(\log T+\log |V| + \log| Q|)$ bits,
such that each operation among \datainsert{}, \dataremove{}, \dataupdate{}, \dataeval{} and \datasuccessor{} can be resolved in time $O(\log^2 T+\log |V|)$.
\end{theorem}

\begin{proof}
The initial state and execution-log  are stored by separate data structures.

The initial state $\X_0\in  Q^V$ is maintained by a deterministic dynamic dictionary, with $(v,X_0(v))$ for vertices $v\in V$ as the key-value pairs.
%
%
Such a deterministic data structure answers queries of $X_0(v)$ given any $v\in V$ while $V$ is dynamically changing.

The execution-log $\exelog{X}{v_t}{T} \in (V \times  Q)^T$ is stored by $|V|$ balanced search trees $(\T_v)_{v \in V}$ (e.g.,~red-black trees).
In each tree $\T_v$, each node in $\T_v$ stores a distinct transition $\Brac{v_t,X_t(v_t)}$ with $v_t=v$, such that the in-order tree walk of $\T_v$ prints all $\Brac{v_t,X_t(v_t)}$ with $v_t=v$ in the order they appear in the execution-log $\exelog{X}{v_t}{T}$.
Altogether these trees $(\T_v)_{v \in V}$  have $T$ nodes in total.
Besides, these trees $(\T_v)_{v \in V}$ are indexed by another deterministic dynamic dictionary, with $(v,p_v)$ for vertices $v\in V$ as key-value pairs, where each $p_v$ is the pointer to the root of tree $\T_v$.
This dictionary provides random accesses to the trees $\T_v$ for all $v\in V$, while $V$ is dynamically changing.

Given any $t$, we want to answer predecessor (or successor) search for the largest $i\le t$ (or smallest $i>t$) such that $v_i=v$.
This is achieved with assistance from another data structure, an \emph{order-statistic tree} (or \emph{OS-tree}) $\hatT$~\cite[Section~14]{cormen2009introduction}.
In $\hatT$, each node stores the ``identity'' of an individual transition $\exelog{X}{v_t}{T}$ (which is actually a pointer  to the node storing the transition $\Brac{v_t,X_t(v_t)}$ in the tree $\T_v$ with $v_t=v$).
In particular, the in-order tree walk of $\hatT$ prints all $\exelog{X}{v_t}{T}$ in that order.
Such a data structure supports two query functions: (1)~Select: given any $t$, returns the identity of the $t$-th transition $\Brac{v_t,X_t(v_t)}$; and (2)~Rank: given the identity of any transition $\Brac{v_t,X_t(v_t)}$, returns its rank $t$ in the sequence $\exelog{X}{v_t}{T}$.
Besides, the OS-tree $\hatT$ also supports standard insertion (of a new transition $\Brac{v,c}$ to a given rank $t$) and deletion (of the transition $\Brac{v_t,X_t(v_t)}$ at a given rank $t$).
As a balanced tree, all these queries and updates for the OS-tree $\hatT$ can be resolved in $O(\log T)$ time.

The successor and predecessor searches mentioned above for any $v\in T$ and $t$, can then be resolved by binary searches in the balanced search tree $\T_v$ while querying the OS-tree $\hatT$ as an oracle for ordering, which takes time at most $O(\log^2 T+\log |V|)$ in total, 
where the $\log |V|$ cost is used for accessing the root of $\T_v$ via the dynamic dictionary that indexes the trees $(T_v)_{v\in V}$.

This solves the successor query $\datasuccessor(t, v)$ as well as the evaluation query $\dataeval(t, v)$ for Gibbs sampling, both within time cost $O(\log^2 T+\log |V|)$, where the latter is actually solved by the predecessor search for the largest $i\le t$ such that $v_i=v$ and returning the value of $X_i(v_i)$ recorded in the $i$-th transition $\Brac{v_i,X_i(v_i)}$ or returning the value of $X_0(v)$ if no such $i$ exists.

It is also easy to verify that with the above dynamic data structures, all updates, including: $\datainsert(t,v,c)$, $\dataremove(t)$ and $\dataupdate(t, c)$, can be implemented with cost at most $O(\log^2 T+\log |V|)$, and the data structures together use $O(T+|V|)$ words in total, where each word consists of $O(\log T+\log|V|+\log| Q|)$ bits.
\end{proof}

\subsection{Single-sample dynamic Gibbs sampling algorithm}

With the data structure for Gibbs sampling stated in Theorem~\ref{theorem-DS}, 
the couplings constructed in Section~\ref{section-coupling} can be implemented as the algorithm for dynamic Gibbs sampling. Recall $\dgraph(\cdot,\cdot)$ and $\dham(\cdot,\cdot)$ are defined in~\eqref{eq-def-dg-dh}.


\begin{lemma}[\textbf{single-sample dynamic Gibbs sampling algorithm}]\label{lemma-dynamic-gibbs-sampling}
Let $\epsilon:\mathbb{N}^+ \to (0,1)$ be an error function. 
Let $\I = (V, E, Q, \Phi)$ be an MRF instance with $n=|V|$ and $\I'= (V', E', Q, \Phi')$ the updated instance with $n' = |V'|$.
Denote $T=T(\I)$, $T' = T(\I')$ and $T_{\max}=\max\{T,T'\}$.
Assume $\dgraph(\I,\I')\leq \Lgra = o(n)$, $\dham(\I,\I')\leq \Lham$, and $T,T'\in\Omega(n\log n)$.
%
The single-sample dynamic Gibbs sampling algorithm (\Cref{alg-dynamic-Gibbs}) does the followings:
\begin{itemize}
\item (\textbf{space cost})
The algorithm maintains an explicit copy of a sample $\X\in Q^V$ for the current instance $\I$, and also a data structure using  $O(T)$ memory words, each of $O(\log T)$ bits, for representing an initial state $\X_0\in Q^V$ and an execution-log $\Exelog(\I,T)=\exelog{X}{v_t}{T}$ for the Gibbs sampling $(\X_t)_{t=0}^T$ on $\I$ generating sample $\X=\X_T$.
\item (\textbf{correctness})
Assuming that Condition~\ref{exe-log-invariant} holds for $\X_0$ and $\Exelog(\I,T)$ for the Gibbs sampling on $\I$,
upon each update 
that modifies $\I$ to $\I'$, the algorithm updates $\X$ to an explicit copy of a sample $\Y\in Q^{V'}$ for the new instance $\I'$, and correspondingly updates the $\X_0$ and $\Exelog(\I,T)$ represented by the data structure to a $\Y_0\in Q^{V'}$ and $\Exelog(\I',T')=\exelog{Y}{v_t'}{T'}$ for the Gibbs sampling $(\Y_t)_{t=0}^{T'}$ on $\I'$ generating the new sample $\Y=\Y_{T'}$, 
where $\Y_0$ and $\Exelog(\I',T')$ satisfy Condition~\ref{exe-log-invariant} for the Gibbs sampling on $\I'$,
therefore, 
\[
\DTV{\Y}{\mu_{\I'}}\le\epsilon(n').
\]
\item (\textbf{time cost})
Assuming Condition~\ref{exe-log-invariant} for $\X_0$ and $\Exelog(\I,T)$ for the Gibbs sampling on $\I$,
the expected time complexity for resolving an update is
\begin{align*}
O\left(\Delta n  + \Delta\left(|T-T'|+ \frac{T_{\max}(\Lham +\Lgra) }{n} + \E{\Rham} + \E{\Rgra} \right)\log^2 T_{\max}\right),
\end{align*}
where $\Delta = \max\{\Delta_G,\Delta_{G'}\}$, $\Delta_G$,$\Delta_{G'}$ denote the maximum degrees of  $G=(V,E)$ and $G'=(V',E')$, $\Rham$ is defined in~\eqref{eq-def-R-Ising} 
for the subroutine $\UpdateHamiltonian$ in \Cref{alg-dynamic-Gibbs}, 
and $\Rgra$ is defined in~\eqref{eq-def-R} for the subroutine $\UpdateEdge$ in \Cref{alg-dynamic-Gibbs}.
\end{itemize}
\end{lemma}
We remark that the $O(\Delta n)$ in time cost is necessary because the update from $\I$ to $\I'$ may change all the potentials of vertices and edges. One can reduce the $O(\Delta n)$ from the time cost if we further restrict that one update can only change constant number of vertices, edges, and potentials. 

The following result is a corollary from \Cref{lemma-dynamic-gibbs-sampling}.
\begin{corollary}
\label{corollary-dynamic-Gibbs}
Assume $\epsilon: \mathbb{N}^+ \to (0,1)$ in \Cref{lemma-dynamic-gibbs-sampling} satisfies the bounded difference condition in Definition~\ref{definition-estimator-dynamic}.
Assume $\I$ and $\I'$ in \Cref{lemma-dynamic-gibbs-sampling} both satisfy Dobrushin-Shlosman condition (\Cref{condition-Dobrushin}) with constant $\delta > 0$.
The single-sample dynamic Gibbs sampling algorithm (\Cref{alg-dynamic-Gibbs}) uses  $O(n\log n)$ memory words, each of $O(\log n)$ bits to maintain the sample for current instance $\I$, and resolves the update from $\I$ to $\I'$ with expected time cost $O\tp{\Delta n + \Delta^2(\Lgra + \Lham )  \log^3 n }$.
\end{corollary}



%

\begin{proof}[Proof of \Cref{lemma-dynamic-gibbs-sampling}]
The dynamic Gibbs sampling algorithm is implemented as follows.
The algorithm uses the dynamic data structure in Theorem~\ref{theorem-DS} to maintain the initial state $\X_0$ and execution-log $\Exelog(\I,T)=\exelog{X}{v_t}{T}$. Besides, the algorithm maintains the explicit copy of the sample $\X \in Q^V$ by a deterministic dynamic dictionary, with $(v, X(v))$ for vertices $v \in V$ as the key-value pairs. The lemma is proved as follows. 

\vspace{4pt}
\noindent \textbf{Space cost:}
Note that $T = \Omega(n \log n), |V| = n$ and $|Q| = O(1)$. We have $O(n) = O(T)$ and 
$O(\log T + \log |V| + \log |Q|) = O(\log T)$.
The dynamic dictionary for sample $\X$ uses $O(n)$ memory words, each of $O(\log |V| + \log |Q|)$ bits.
Combining with Theorem~\ref{theorem-DS}, we have the algorithm uses $O(T)$ memory words to maintain the initial state, execution-log and the random sample, each word is of $O(\log T + \log |V| + \log |Q|) = O(\log T)$ bits.

\vspace{4pt}
\noindent \textbf{Correctness:}
The invariants for execution-log (Condition~\ref{exe-log-invariant}) are preserved by the coupling simulated by the algorithm. The correctness holds as a consequence.

\vspace{4pt}
\noindent \textbf{Time cost:}
Consider the update that modifies $\I$ to $\I'$. 
We divide the algorithm into two stages.
\begin{itemize}
\item \textbf{Preparation stage}: construct the updated  instances $\I'$ and other middle instances $\Imid,\I_1,\I_2$ in~\eqref{eq-def-Imid},~\eqref{eq-def-I-1},~\eqref{eq-def-I-2};  compute $\pup_v$ in~\eqref{eq-def-Ising-up} for all $v \in V$ and construct the random set $\Isingset \subseteq [T]=\{1,2,\ldots,T\}$ in~\eqref{eq-def-Ising-set}.
\item  \textbf{Update stage}: given $\Isingset$ and $\pup_v$ for all $v \in V$, update the initial state $\X_0$ to $\Y_0$, the execution-log $\Exelog(\I,T)=\exelog{X}{v_t}{T}$ to $\Exelog(\I',T')=\exelog{Y}{v'_t}{T'}$, and the sample $\X$ to $\Y$.
\end{itemize}

We make the following two claims.
\begin{claim}
\label{claim-prep-single}
The expected
running time of the preparation stage is 
\begin{align*}
\E{\Tpre} = O\left( \Delta n + \E{\abs{\Isingset}} \log^2 T_{\max}\right),
\end{align*}
and the expected size of $\Isingset$ is at most $\frac{4T_{\max} \Lham}{n}$.
\end{claim}

\begin{claim}
\label{claim-time-cost}
The expected running time of the update stage is 
\begin{align}
\label{eq-time}
\E{\Tupd}  = O\left(\Delta\left(|T-T'|+ \frac{T_{\max}\Lgra}{n}  + \E{\Rham} + \E{\Rgra} \right)\log^2 T_{\max} \right), 
\end{align}
 $\Rham$ is defined in~\eqref{eq-def-R-Ising} for the subroutine $\UpdateHamiltonian$ in~\Cref{alg-dynamic-Gibbs}, 
and $\Rgra$ is defined in~\eqref{eq-def-R} for the subroutine $\UpdateEdge$ in~\Cref{alg-dynamic-Gibbs}.
\end{claim}
By the linearity of expectation,
the expected time cost of the algorithm is 	
$\E{\Tpre}+\E{\Tupd}$.
This proves the time cost.

\end{proof}

We introduce the following technique lemma to prove \Cref{corollary-dynamic-Gibbs}.

\begin{lemma}
\label{lemma-smooth-epsilon}
Let $\epsilon: \mathbb{N}^+ \to (0,1)$ be a function  such that there exists a constant $C > 0$ such that
\begin{align*}
\forall	n \in \mathbb{N}^+: \quad \abs{\epsilon(n+1) - \epsilon(n)} \leq \frac{C}{n}\epsilon(n).
\end{align*}
Then the function $N$ has the following properties
\begin{itemize}
\item for any $n \in \mathbb{N}^+$, it holds that $\epsilon(n) \geq \frac{1}{\mathrm{poly}(n)}$;
\item let $\alpha \geq 1$ be a constant, given any $n,n' \in \mathbb{N}^+$ such that $\frac{1}{\alpha}\leq  \frac{n'}{n}\leq \alpha$,
\begin{align*}
\abs{n\log \frac{n}{\epsilon(n)} - n'\log \frac{n'}{\epsilon(n')}} = C'\abs{n' - n}\log n.
\end{align*}

where $C'$ is a constant that depends only on $\alpha, C$ and $\epsilon(3\lceil C\rceil)$.
\end{itemize}
\end{lemma}
\begin{proof}
By the condition, we have $\epsilon(t) \leq \left(1 + \frac{C}{t-C}\right)\epsilon(t+1)$ for all $t > \lceil C+1\rceil$. 
Thus for all $n > l = 3\lceil C\rceil$,
\begin{align}
\label{eq-epsilon-upper}
\epsilon(l) \leq \prod_{i=l}^{n-1}\left( 1 + \frac{C}{i-C} \right)  \epsilon(n)\leq \epsilon(n) \exp\left(C \sum_{i=2}^{n-1}\frac{1}{i} \right) \leq \epsilon(n)\exp(C\ln n) = \epsilon(n)n^C.
\end{align}
Thus, we have $\epsilon(n) \geq \frac{1}{\mathrm{poly}(n)}$.

We then prove the second property. It is lossless to assume that $\min\{n,n'\}\geq l$, 
since otherwise we can choose $C'$ sufficiently large so that the second property holds.
Firstly, we prove for the case $n>n'$. We have
$\abs{\log \frac{n}{n'}} \leq \frac{n-n'}{n'}$.
By $\epsilon(t) \leq \left(1 + \frac{C}{t-C}\right)\epsilon(t+1)$ for all $t > \lceil C+1\rceil$,
we also have
\begin{align*}
\epsilon(n') \leq \prod_{i=n'}^{n-1}\left( 1 + \frac{C}{i-C} \right)  \epsilon(n)\leq \epsilon(n) \exp\left( \frac{C(n-n')}{n'-C} \right).
\end{align*}
Thus, 
\begin{align}
\label{eq-proof-epsilon}
&\abs{\log \frac{n}{\epsilon(n)} - \log \frac{n'}{\epsilon(n')}}
\leq \abs{\log \frac{n}{n'}} + \abs{\log \frac{\epsilon(n)}{\epsilon(n')}}
 \leq \frac{n-n'}{n'}+ \frac{C(n-n')}{n'-C} 
\leq \frac{(2C+1)(n-n')}{n'}.
\end{align}
The last equality is due to $2(n'-C) \geq n'+l -2C \geq n'$.
Let $C' = 2 + \abs{\log \epsilon(l)} + 3C$.
We have 
\begin{align*}
\abs{n\log \frac{n}{\epsilon(n)} - n'\log \frac{n'}{\epsilon(n')}}
\leq \abs{(n' - n)\log \frac{n}{\epsilon(n)}} + \abs{n'\left(\log \frac{n}{\epsilon(n)} - \log \frac{n'}{\epsilon(n')}\right)}\leq C'\abs{n' - n}\log n.
\end{align*}
The last inequality is due to (\ref{eq-epsilon-upper}) and (\ref{eq-proof-epsilon}).
Similarly, we can also prove the lemma if $n<n'$.
\end{proof}

\begin{proof}[Proof of \Cref{corollary-dynamic-Gibbs}]
By $\Lgra = o(n)$, we have $n' =\Theta(n)$.
Since  $\I$ and $\I'$  both satisfy Dobrushin-Shlosman condition (\Cref{condition-Dobrushin}) with constant $\delta > 0$, we can set  $T,T'$ as in \eqref{eq-upper-bound-T-mix} such that 
\begin{align*}
T &= \left\lceil \frac{n}{\delta} \log \frac{n}{\eps(n)} \right\rceil = \Theta(n \log n)\\
T'	&= \left\lceil \frac{n'}{\delta} \log \frac{n'}{\eps(n')} \right\rceil = \Theta(n \log n).
\end{align*}
The equations  hold because $n' =\Theta(n)$ and 
the error function $\epsilon$ satisfies $\epsilon(\ell) \geq \frac{1}{\mathrm{poly}(\ell)}$ by~\Cref{lemma-smooth-epsilon}.
Thus, we have
\begin{align}
\label{eq-dobrushin-T-max}
T_{\max} = \max\{T, T'\} = O(n \log n).	
\end{align}
By~\Cref{lemma-smooth-epsilon} and $|n'-n|\leq \Lgra = o(n)$, we have
\begin{align}
\label{eq-dobrushin-T}
\abs{T - T'} = 	  O(\Lgra \log n ).
\end{align}

Let $\Imid = (V,E,Q,\Phi^{\mathsf{mid}})$ be the middle instance constructed as in~\eqref{eq-def-Imid}. 
In \Cref{alg-dynamic-Gibbs}, we call the subroutine  $\UpdateHamiltonian$ for instances $\I$ and $\Imid$. Since $\I$ satisfies the Dobrushin-Shlosman condition, by \Cref{lemma-upper-bound-R-Ham} and $\d(\I,\Imid)\leq d(\I,\I') \leq \Lham$,
we have
\begin{align}
\label{eq-dobrushin-R1}
\E{\Rham} = O \left( \frac{\Delta T \Lham}{\delta n} \right)	 = O(\Delta \Lham \log n),
\end{align}
where $\Rham$ is defined in~\eqref{eq-def-R-Ising} for the subroutine $\UpdateHamiltonian$.

We also call the subroutine  \UpdateGraph{} for instances $\Imid$ and $\I'$ in \Cref{alg-dynamic-Gibbs}.
The subroutine is shown in~\Cref{alg-update-graph}. We first add isolated vertices to update $\Imid$ to $\I_1$, then update edges to update $\I_1$ to $\I_2$, finally delete isolated vertices to update $\I_2$ to $\I'$.
Since $\I'$ satisfies Dobrushin-Shlosman condition and 
the only difference between $\I_2$ and $\I'$ is that
$\I_2$ contains extra isolated vertices, 
it is easy to verify that $\I_2$ also satisfies  Dobrushin-Shlosman condition.
In~\Cref{alg-update-graph}, the subroutine $\UpdateEdge$ is called for $\I_1$ and $\I_2$.
By \Cref{lemma-upper-bound-R}, we have 
\begin{align}
\label{eq-dobrushin-R2}
\E{\Rgra} = O\left( \frac{\Delta T \Lgra}{\Delta n} \right)	= O(\Delta \Lgra \log n).
\end{align}
where $\Rgra$ is defined in~\eqref{eq-def-R} for the subroutine $\UpdateEdge$.

Combining~\eqref{eq-dobrushin-T-max},~\eqref{eq-dobrushin-T},~\eqref{eq-dobrushin-R1},~\eqref{eq-dobrushin-R2} with~\Cref{lemma-dynamic-gibbs-sampling}, we have the expected time cost is
\begin{align*}
 \E{T_{\mathsf{cost}}} &= O\left(\Delta n+ \Delta\left(|T-T'|+ \frac{T_{\max}(\Lham+ \Lgra) }{n} + \E{\Rham} + \E{\Rgra} \right)\log^2 T_{\max}\right)\\
& = O\left(\Delta n + \Delta^2 (\Lgra + \Lham )  \log^3 n\right).\qedhere
\end{align*}
\end{proof}

\subsection{Multi-sample dynamic Gibbs sampling algorithm}
\label{section-multi-sample-alg}
In this section, we give an \emph{Multi-sample dynamic Gibbs sampling algorithm} that maintains multiple independent random samples for the current MRF instance.
\Cref{theorem-sample-MRF} follows immediately from the following lemma.
%
%
%
%

\begin{lemma}[\textbf{multi-sample dynamic Gibbs sampling algorithm}]\label{lemma-dynamic-gibbs-sampling-multi}
Let $N:\mathbb{N}^+ \to \mathbb{N}^+$ and $\epsilon: \mathbb{N}^+ \to (0,1)$ be two functions satisfying the bounded difference condition in Definition~\ref{definition-estimator-dynamic}.
Let $\I = (V, E, Q, \Phi)$ be an MRF instance with $n=|V|$ and $\I'= (V', E', Q, \Phi')$ the updated instance with $n' = |V'|$.
Assume that  $\I$ and $\I'$ both satisfy Dobrushin-Shlosman condition with constant $\delta > 0$,  $\dgraph(\I,\I')\leq \Lgra =o(n)$ and $\dham(\I,\I') \leq \Lham$.
Denote $T=\lceil \frac{n}{\delta}\log \frac{n}{\epsilon(n)}\rceil$, $T' = \lceil \frac{n'}{\delta}\log \frac{n'}{\epsilon(n')}\rceil$.

There is  an algorithm which does the followings:
\begin{itemize}
\item (\textbf{space cost})
The algorithm maintains $N(n)$ explicit copies of independent samples
$\X^{(1)},\ldots,\X^{(N(n))}$, where $\X^{(i)}\in Q^V$ for all $1\leq i \leq N(n)$,
for the current instance $\I$, and also a data structure using  $O(nN(n)\log n)$ memory words, each of $O(\log n)$ bits, for representing the initial state $\X^{(i)}_0\in Q^V$ and the execution-log $\Exelog^{(i)}(\I,T)=\exelog{X^{(i)}}{v^{(i)}_t}{T}$ for  $1\leq i \leq N(n)$ such that each Gibbs sampling $(\X^{(i)}_t)_{t=0}^T$ on $\I$ generating an independent sample $\X^{(i)}=\X^{(i)}_T$.
\item (\textbf{correctness})
Assuming that Condition~\ref{exe-log-invariant} holds for each $\X^{(i)}_0$ and $\Exelog^{(i)}(\I,T)$ for the Gibbs sampling on $\I$,
upon each update 
that modifies $\I$ to $\I'$, the algorithm updates $\X^{(1)},\X^{(2)},\ldots,\X^{(N(n))}$ to $N(n')$ explicit copies of independent samples 
$\Y^{(1)},\Y^{(2)},\ldots,\Y^{(N(n'))}\in Q^{V'}$
for the new instance $\I'$, and correspondingly updates the data represented by the data structure to $\Y^{(i)}_0\in Q^{V'}$ and $\Exelog^{(i)}(\I',T')=\exelog{Y^{(i)}}{u^{(i)}_t}{T'}$ for $1\leq i \leq N(n')$ such that each Gibbs sampling chain $(\Y^{(i)}_t)_{t=0}^{T'}$ on $\I'$ generating a new sample $\Y^{(i)}=\Y^{(i)}_{T'}$, 
where each $\Y^{(i)}_0$ and $\Exelog^{(i)}(\I',T')$ satisfy Condition~\ref{exe-log-invariant} for the Gibbs sampling on $\I'$,
therefore, 
\[
\DTV{\Y^{(i)}}{\mu_{\I'}}\le\epsilon(n').
\]
\item (\textbf{time cost})
Assuming Condition~\ref{exe-log-invariant} for each $\X^{(i)}_0$ and $\Exelog^{(i)}(\I,T)$ for the Gibbs sampling on $\I$,
the time complexity for resolving an update is 
$$O\left(\Delta^2( \Lham + \Lgra )N(n)\cdot\log ^ 3 n +\Delta n \right),$$
where $\Delta= \max\{\Delta_G,\Delta_{G'}\}$, and $\Delta_G,\Delta_{G'}$ denote the maximum degree of $G=(V,E)$ and $G'=(V',E')$.  
\end{itemize}
\end{lemma}

The following technique lemma will be used to prove \Cref{lemma-dynamic-gibbs-sampling-multi}.
\begin{lemma}
\label{lemma-smooth-function}
Let $N: \mathbb{N}^+\to \mathbb{N}^+$ be a function  such that there exists a constant $C > 0$ such that
\begin{align*}
\forall	n \in \mathbb{N}^+: \quad \abs{N(n+1) - N(n)} \leq \frac{C}{n}N(n).
\end{align*}
Then the function $N$ has the following properties
\begin{itemize}
\item for any $n \in \mathbb{N}^+$, it holds that $N(n) \leq \mathrm{poly}(n)$;
\item let $\alpha \geq 1$ be a constant, given any $n,n' \in \mathbb{N}^+$ such that $\frac{1}{\alpha}\leq  \frac{n'}{n}\leq \alpha$,
\begin{align*}
\abs{N(n) - N(n')} = C'(\alpha,C) \cdot \frac{\abs{n-n'}}{n}N(n),
\end{align*}
where $C'(\alpha,C)$ is a constant that depends only on $\alpha$ and $C$.
\end{itemize}
\end{lemma}
\begin{proof}
By the condition, we have $N(n+1) \leq \left(1 + \frac{C}{n}\right)N(n)$. Thus for all $n \in \mathbb{N}^+$,
\begin{align*}
N(n) \leq N(1)\prod_{i=1}^{n-1}\left( 1 + \frac{C}{i} \right)\leq N(1) \exp\left(C \sum_{i=1}^{n-1}\frac{1}{i} \right) = N(1)\exp(\Theta(\ln n)) =\mathrm{poly}(n).
\end{align*}
We then prove the second property. 
Note that $\frac{\abs{n-n'}}{n} \leq \alpha$, it suffices to prove
\begin{align}
\label{eq-proof-N}
\abs{\frac{N(n')}{N(n)}	-1} \leq C'(\alpha,C) \cdot \frac{\abs{n-n'}}{n}.
\end{align}
Assume that $\min\{n,n'\}\leq 2 C \alpha$. 
Then, we have $\max\{n,n'\}\leq 2 C \alpha^2$.
We can choose $C'(\alpha,C)$ sufficiently large so that~\eqref{eq-proof-N} holds.
Assume $n' > n > 2\alpha C$. Note that $\frac{\abs{n-n'}}{n} \leq \alpha$. We have
\begin{align*}
1 - \frac{C\abs{n-n'}}{n}  \leq \left( 1 - \frac{C}{n} \right)^{\abs{n-n'}} \leq \frac{N(n')}{N(n)} \leq \left( 1 + \frac{C}{n} \right)^{\abs{n-n'}} \leq 1+\frac{C\exp(\alpha C)\abs{n-n'}}{n},
\end{align*}
which implies~\eqref{eq-proof-N} holds if $C'(\alpha,C) \geq C\exp(\alpha C)$.
Assume $n > n' > 2\alpha C$. Note that $\frac{\abs{n-n'}}{n}\leq \alpha$ and $n' \geq \frac{n}{\alpha}$. We  have
\begin{align*}
 1 - \frac{\alpha C \abs{n-n'} }{n}  \leq \left( 1 - \frac{\alpha C}{n} \right)^{\abs{n-n'}} \leq \frac{N(n')}{N(n)} \leq \left( 1 + \frac{\alpha C}{n} \right)^{\abs{n-n'}}  \leq 1+\frac{C\alpha\exp(\alpha^2 C)\abs{n-n'}}{n}.
\end{align*}
which implies~\eqref{eq-proof-N} holds if $C'(\alpha,C) \geq C\alpha\exp(\alpha^2 C)$.
\end{proof}

\begin{proof}
The main idea of the multi-sample dynamic Gibbs sampling algorithm is to use single-sample dynamic Gibbs sampling algorithm (\Cref{alg-dynamic-Gibbs}) to maintain each sample $\X^{(i)} \in Q^V$ for $1\leq i \leq N(n)$. We need a careful implementation of the algorithm to guarantee the time cost in \Cref{lemma-dynamic-gibbs-sampling-multi}. 

\noindent \textbf{Space cost:}
Note that $T = \left\lceil \frac{n}{\delta} \log \frac{n}{\eps(n)} \right\rceil = \Theta(n \log n)$ due to \Cref{lemma-smooth-epsilon} and $N(n) \leq \mathrm{poly}(n)$ due to \Cref{lemma-smooth-function}. The dynamic dictionary for each sample $\X^{(i)}$ uses $O(n)$ memory words, each of $O(\log n)$ bits.
Hence, the algorithm uses $O(T \cdot N(n)) = O(n N(n) \log n)$ memory words to maintain all the initial states, execution-logs and the random samples due to Theorem~\ref{theorem-DS}.

\vspace{4pt}
\noindent \textbf{Correctness:}
The invariants for execution-log (Condition~\ref{exe-log-invariant}) are preserved by the coupling simulated by the algorithm. The correctness holds as a consequence.

\vspace{4pt}
\noindent \textbf{Time cost:}
Define $N_{\min} \triangleq \min \{N(n),N(n')\}$.
Fix $1\leq k \leq N_{\min}$. We  use the \Cref{alg-dynamic-Gibbs} to update the sample $\X^{(k)}$ to $\Y^{(k)}$.
Let $\Isingset_k \subseteq [T]$ denote the set defined in~\eqref{eq-def-Ising-set} 
for the subroutine $\UpdateHamiltonian$ in~\Cref{alg-dynamic-Gibbs}.
The  multi-sample dynamic Gibbs sampling has the following three stages.
\begin{itemize}
\item \textbf{Preparation stage}: construct the updated instances $\I'$ and other middle instances $\Imid,\I_1,\I_2$ in~\eqref{eq-def-Imid},~\eqref{eq-def-I-1},~\eqref{eq-def-I-2}; compute $\pup_v$ in~\eqref{eq-def-Ising-up} for all $v \in V$; and construct the random sets $\Isingset_1, \Isingset_2,\ldots,\Isingset_{N_{\min}}$.
\item \textbf{Update stage}: given the $(\pup_v)_{v \in V}$ and $(\Isingset_i)_{1\leq i\leq N_{\min}}$, for each $1\leq i \leq N_{\min}$, use \Cref{alg-dynamic-Gibbs} to update the initial state $\X_0^{(i)}$ to $\Y_0^{(i)}$, the execution-log $\Exelog^{(i)}(\I,T)=\exelog{X^{(i)}}{v_t^{(i)}}{T}$ to $\Exelog^{(i)}(\I',T')=\exelog{Y^{(i)}}{u^{(i)}}{T'}$, and the sample $\X^{(i)}$ to $\Y^{(i)}$.
\item \textbf{Completion stage}: If $N(n') < N(n)$, for each $N(n')< i\leq N(n)$, remove the sample $\X^{(i)}$, the initial state $\X^{(i)}_0$ and the  execution-log $\Exelog^{(i)}(\I,T)=\exelog{X^{(i)}}{v_t^{(i)}}{T}$ from the data; if $N(n') > N(n)$, for each  $N(n) < i\leq N(n')$, construct an independent Gibbs sampling chain $(\Y^{(i)}_t)_{t=0}^{T'}$ on instance $\I'$, write the sample $\Y^{(i)}=\Y^{(i)}_{T'}$, the initial state $\Y^{(i)}_0$ and the  execution-log $\Exelog^{(i)}(\I',T')=\exelog{Y^{(i)}}{u_t^{(i)}}{T'}$ into the data. 
\end{itemize}
Let $\Tprem,\Tupdm$ and $\Trefine$ denote the running time of the corresponding stages. Note that the update stage of the multi-sample dynamic sampling algorithm repeats the update stage of the single-sample algorithm for $N_{\min}$ times. 
Also note that both $\I$ and $\I'$ satisfies Dobrushin-Shlosman condition.
Combining~\eqref{eq-time},~\eqref{eq-dobrushin-T-max},~\eqref{eq-dobrushin-T},~\eqref{eq-dobrushin-R1}, and~\eqref{eq-dobrushin-R2}, we have
\begin{align}
\label{eq-multi-T2}
\E{\Tupdm} = \sum_{i=1}^{N_{\min}}\E{T_{\mathsf{update}}^{\mathsf{single},(i)}} &=   O(N_{\min} \Delta^2 (\Lgra + \Lham )  \log^3 n )\notag\\
(\text{by } N_{\min} \leq N(n))\quad&= O(N(n) \Delta^2 (\Lgra + \Lham )  \log^3 n )
\end{align}
where $T_{\mathsf{update}}^{\mathsf{single},(i)}$ is the running time of the update stage of the \Cref{alg-dynamic-Gibbs} that updates the $i$-th sample.

In completion stage, we either remove the chains from the data structure, or generate the new chains and write them into data structure. 
It is easy to see the running time of the completion stage satisfies
\begin{align*}
\E{\Trefine} &= O(\abs{N(n)-N(n')}T_{\max} \log T_{\max}) =O(n\abs{N(n)-N(n')}\log^2 n)\\
(\text{by \Cref{lemma-smooth-function}})\quad&= O(\abs{n-n'}N(n) \log^2 n) =O(\Lgra N(n)\log^2 n),
\end{align*}
where $T_{\max}=\max\{T,T'\} = O(n\log n)$ since $n' = \Theta(n)$ and 
$\epsilon(n') \geq \frac{1}{\mathrm{poly}(n')}$(by $\Lgra = o(n)$ and \Cref{lemma-smooth-epsilon}). 

We make the following claim about the preparation stage. 
\begin{claim}
\label{claim-prep-multi}
The expected running time of the preparation stage is 
\begin{align*}
\E{\Tprem} = O\left( \Delta n + \log^2 n  \sum_{i=1}^{N_{\min}}\E{|\Isingset_i|} \right),
\end{align*}
and the expected size of $\Isingset_i$ is at most $\frac{4 T_{\max} \Lham}{n}$ for each $1\leq i \leq N_{\min}$.
\end{claim}
By \Cref{claim-prep-multi}, we have
\begin{align*}
\E{\Tprem} = O\left( \Delta n + N(n)\Lham \log^3 n \right).
\end{align*}
By the linearity of expectation,
the expected time cost of the algorithm is 	
$\E{\Tprem}+\E{\Tupdm} + \E{\Trefine}$.
This proves the time cost.

\end{proof}

\section{Proofs for dynamic Gibbs sampling}
\label{sec-proof-dynamic}
\subsection{Analysis of the couplings}
\label{section-proof-couple}
We analysis the couplings in dynamic Gibbs sampling algorithm. 
In~\Cref{section-proof-ham}, we analysis the coupling for Hamiltonian update.
In~\Cref{section-proof-gra},  we analysis the coupling for graph update.
 
\subsubsection{Proofs for the coupling for Hamiltonian update}
\label{section-proof-ham}
In this section, we prove \Cref{lemma-valid-coupling}, \Cref{lemma-correct-up-bound}, and \Cref{lemma-upper-bound-R-Ham}.

\paragraph{The validity of the coupling (proof of \Cref{lemma-valid-coupling})}
We first prove that the distribution $\nu^{\tau}_{\I_v,\I'_v}(\cdot)$ in~\eqref{eq-correct-dist} is valid. 
We draw samples from  $\nu^{\tau}_{\I_v,\I'_v}(\cdot)$ only if the result of coin flipping is HEADS, which implies $\mu_{v,\I}(x \mid \tau) > \mu_{v,\I'}(x \mid \tau)$ for some $x \in Q$. 
Thus, the two distributions $\mu_{v,\I}(\cdot \mid \tau)$ and $\mu_{v,\I'}(\cdot\mid \tau)$ are not identical,
and
\begin{align*}
\sum_{x\in Q}\max \left\{0, \mu_{v, \I}(x \mid \tau ) - \mu_{v,\I'}(x \mid \tau)\right\} > 0.	
\end{align*}
Hence, the denominator of $\nu^{\tau}_{\I_v,\I'_v}(\cdot)$ is positive.
Besides, since both $\mu_{v,\I}(\cdot \mid \tau)$ and $\mu_{v,\I'}(\cdot\mid \tau)$ are distributions over $Q$, we have
\begin{align*}
\sum_{x\in Q}\max \left\{0, \mu_{v, \I'}(x \mid \tau ) - \mu_{v,\I}(x \mid \tau)\right\} = \sum_{x\in Q}\max \left\{0, \mu_{v, \I}(x \mid \tau ) - \mu_{v,\I'}(x \mid \tau)\right\}.
\end{align*}
Thus we have 
$\sum_{x\in Q}\nu_{\I_{v},\I'_{v}}^{\tau}(x)  = 1.$
Hence, $\nu^{\tau}_{\I_v,\I'_v}(\cdot)$ a valid distribution.

We next prove the coupling $D_{\I_v,\I'_v}^{\sigma,\tau}(\cdot,\cdot)$ in \Cref{definition-Ising-coupling-step} is a valid coupling between  $\mu_{v,\I}(\cdot \mid \tau)$ and $\mu_{v,\I'}(\cdot\mid \tau)$.
If $\mu_{v,\I}(\cdot \mid \tau)$ and $\mu_{v,\I'}(\cdot\mid \tau)$ are identical, the result holds trivially. We may assume $\mu_{v,\I}(\cdot \mid \tau)$ and $\mu_{v,\I'}(\cdot\mid \tau)$ are not identical, thus the distribution $\nu^{\tau}_{\I_v,\I'_v}(\cdot)$ is well-defined.

The coupling $D_{\I_v,\I'_v}^{\sigma,\tau}(\cdot,\cdot)$ in \Cref{definition-Ising-coupling-step}
returns a pair $(c,c') \in Q^2$. It is easy to see $c$ follows the law $\mu_{v,\I}(\cdot \mid \sigma)$.
We prove that $c'$ follows the law $\mu_{v,\I'}(\cdot \mid \sigma)$. By the definition of $D_{\I_v,\I'_v}^{\sigma,\tau}(\cdot,\cdot)$, $c' \in Q$ is generated by the following procedure:
\begin{itemize}
\item sample $a \in Q$ from the distribution $\mu_{v,\I}(\cdot \mid \tau)$;
\item sample $b \in Q$ from the distribution $\nu^{\tau}_{\I_v,\I_v'}$ defined in~\eqref{eq-correct-dist}, set 
\begin{align*}
c' = \begin{cases}
 b  &\text{with probability } p^\tau_{\I_{v_t},\I'_{v_t}}(a)\\
 a &\text{with probability } 1-p^\tau_{\I_{v_t},\I'_{v_t}}(a).
 \end{cases}
 \end{align*}
\end{itemize}
Note that $a$ follows the law $\mu_{v,\I}(\cdot\mid\tau)$. 
We have for each $x \in Q$, 
\begin{align*}
\Pr[c' = x] &= \Pr[a = x]\cdot (1- p^{\tau}_{\I_v,\I'_v}(x)) + \sum_{y \in Q}\Pr[a = y]\cdot p^{\tau}_{\I_v,\I'_v}(y) \cdot \nu_{\I_{v},\I'_{v}}^{\tau}(x)\\
&= \mu_{v, \I}(x\mid \tau)\cdot (1- p^{\tau}_{\I_v,\I'_v}(x)) + \nu_{\I_{v},\I'_{v}}^{\tau}(x)\sum_{y \in Q}\mu_{v, \I}(y\mid \tau)\cdot p^{\tau}_{\I_v,\I'_v}(y).
\end{align*}
By the definition of $p^{\tau}_{\I_v,\I'_v}(y)$ in~\eqref{eq-correct-prob}, we have
\begin{align*}
\forall y \in Q,\quad \mu_{v, \I}(y\mid \tau)\cdot p^{\tau}_{\I_v,\I'_v}(y) = \begin{cases}
 0 &\text{if } \mu_{v, \I}(y \mid \tau) \leq \mu_{v,\I'}(y \mid \tau)\\
 \mu_{v,\I}(y \mid \tau) - \mu_{v,\I'}(y\mid \tau) &\text{otherwise}.	
 \end{cases}
\end{align*}
This implies $\mu_{v, \I}(y\mid \tau)\cdot p^{\tau}_{\I_v,\I'_v}(y) =\max \left\{0, \mu_{v, \I}(y \mid \tau ) - \mu_{v,\I'}(y \mid \tau)\right\}$. We have
\begin{align*}
&\nu_{\I_{v},\I'_{v}}^{\tau}(x)\sum_{y \in Q}\mu_{v, \I}(y\mid \tau)\cdot p^{\tau}_{\I_v,\I'_v}(y)\\
=&\, \frac{\max \left\{0, \mu_{v, \I'}(x \mid \tau ) - \mu_{v,\I}(x \mid \tau)\right\}}{\sum_{y\in Q}\max \left\{0, \mu_{v, \I}(y \mid \tau ) - \mu_{v,\I'}(y \mid \tau)\right\}}\sum_{y \in Q}\max \left\{0, \mu_{v, \I}(y \mid \tau ) - \mu_{v,\I'}(y \mid \tau)\right\}\\
=&\, \max \left\{0, \mu_{v, \I'}(x \mid \tau ) - \mu_{v,\I}(x \mid \tau)\right\}.
\end{align*}
Hence, we have
\begin{align*}
\Pr[c' = x] = 	\mu_{v, \I}(x\mid \tau)\cdot (1- p^{\tau}_{\I_v,\I'_v}(x))  + \max \left\{0, \mu_{v, \I'}(x \mid \tau ) - \mu_{v,\I}(x \mid \tau)\right\}.
\end{align*}
Suppose $\mu_{v,\I}(x \mid \tau) \leq \mu_{v, \I'}(x \mid \tau)$, then we have $p^{\tau}_{\I_v,\I'_v}(x) = 0$. In this case, we have
\begin{align*}
\Pr[c' = x] = \mu_{v, \I}(x\mid \tau) + \mu_{v, \I'}(x \mid \tau ) - \mu_{v,\I}(x \mid \tau) = \mu_{v, \I'}(x \mid \tau ).
\end{align*}
Suppose $\mu_{v,\I}(x \mid \tau) > \mu_{v, \I'}(x \mid \tau)$, then we have
\begin{align*}
\Pr[c' = x] &=\mu_{v, \I}(x\mid \tau)\cdot (1- p^{\tau}_{\I_v,\I'_v}(x))= \mu_{v, \I'}(x \mid \tau ).
\end{align*}
Combining these two cases proves that $c'$ follows the law $\mu_{v,\I'}(\cdot \mid \tau)$.
\qed

\vspace{1em}
\paragraph{The upper bound of the probability $p_{\I_{v},\I'_{v}}^{\cdot}(\cdot)$ (proof of \Cref{lemma-correct-up-bound}) } It suffices to prove that for any two instances $\I=(V,E,Q,\Phi)$ and $\I'=(V,E,Q,\Phi')$ of MRF model, and any $v\in V,c\in Q$ and $\sigma \in Q^{\Gamma_G(v)}$,
\begin{align}
\label{eq-proof-decay}
\mu_{v,\I}(c\mid\sigma)-\mu_{v,\I'}(c \mid \sigma) \leq 	2\mu_{v,\I}(c\mid\sigma)\left( \Vert\phi_v-\phi'_v\Vert_1 + \sum_{e=\{u,v\}\in E}\Vert \phi_e-\phi'_e \Vert_1 \right).
\end{align}
Note that if $\mu_{v,\I}(c\mid\sigma) =0$, then $p_{\I_{v},\I'_{v}}^{\tau}(c) = 0$; otherwise $p_{\I_{v},\I'_{v}}^{\tau}(c) = \max\left\{0,\frac{\mu_{v,\I}(c\mid\sigma)-\mu_{v,\I'}(c \mid \sigma) }{\mu_{v,\I}(c\mid\sigma)}\right\}$. Hence, inequality~\eqref{eq-proof-decay} proves the lemma.

We now prove~\eqref{eq-proof-decay}. 
Suppose $\mu_{v,\I}(c\mid \sigma) = 0$. Then the LHS of~\eqref{eq-proof-decay} $\leq 0$. Since the RHS $\geq 0$, the inequality holds.

We next assume $\mu_{v,\I}(c\mid \sigma) > 0$. Then it suffices to prove
\begin{align*}
\frac{\mu_{v,\I}(c\mid\sigma)-\mu_{v,\I'}(c \mid \sigma)}{\mu_{v,\I}(c\mid\sigma)} = 1 - \frac{\mu_{v,\I'}(c \mid \sigma)}{\mu_{v,\I}(c\mid\sigma)}\leq 	2\left( \Vert\phi_v-\phi'_v\Vert_1 + \sum_{e=\{u,v\}\in E}\Vert \phi_e-\phi'_e \Vert_1 \right).
\end{align*}
By the definitions of $\phi_v,\phi'_v,\phi_e,\phi'_e$, we can write the ratio as
\begin{align*}
\frac{\mu_{v,\I'}(c \mid \sigma)}{\mu_{v,\I}(c\mid\sigma)} = 	\frac{\exp\left(\phi'_v(c)+ \sum_{u \in \Gamma_v}\phi'_{uv}(\sigma_u,c)\right)}{\exp\left(\phi_v(c)+ \sum_{u \in \Gamma_v}\phi_{uv}(\sigma_u,c)\right)}\frac{\sum_{a \in Q}\exp\left(\phi_v(a)+ \sum_{u \in \Gamma_v}\phi_{uv}(\sigma_u,a)\right)}{\sum_{a \in Q}\exp\left(\phi'_v(a)+ \sum_{u \in \Gamma_v}\phi'_{uv}(\sigma_u,a)\right)},
\end{align*}
where $\Gamma_v$ denotes the neighborhood of $v$ in $G$.
Next, we assume that 
\begin{equation}
\label{eq-assume-both}
\begin{split}
\forall c \in Q:\quad &\phi_v(c) = -\infty \quad \Longleftrightarrow \quad 	\phi'_v(c) = -\infty\\
\forall u \in \Gamma_v, c,c'\in Q:\quad &\phi_{uv}(c,c') = -\infty \quad \Longleftrightarrow \quad \phi'_{uv}(c,c') = -\infty.
\end{split}
\end{equation}
Otherwise, it must hold that the RHS of~\eqref{eq-proof-decay} is $\infty$, then ~\eqref{eq-proof-decay} holds trivially. Thus we can define the set
\begin{align*}
Q' \triangleq \left\{a \in Q \mid \phi_v(a) + \sum_{u \in \Gamma_v}\phi_{uv}(\sigma_u, a) \neq -\infty \right\}	 =  \left\{a \in Q \mid \phi'_v(a) + \sum_{u \in \Gamma_v}\phi'_{uv}(\sigma_u, a) \neq -\infty \right\}.
\end{align*}
Since $\exp(-\infty) = 0$, we have
\begin{align*}
\frac{\mu_{v,\I'}(c \mid \sigma)}{\mu_{v,\I}(c\mid\sigma)} = 	\frac{\exp\left(\phi'_v(c)+ \sum_{u \in \Gamma_v}\phi'_{uv}(\sigma_u,c)\right)}{\exp\left(\phi_v(c)+ \sum_{u \in \Gamma_v}\phi_{uv}(\sigma_u,c)\right)}\frac{\sum_{a \in Q'}\exp\left(\phi_v(a)+ \sum_{u \in \Gamma_v}\phi_{uv}(\sigma_u,a)\right)}{\sum_{a \in Q'}\exp\left(\phi'_v(a)+ \sum_{u \in \Gamma_v}\phi'_{uv}(\sigma_u,a)\right)}.
\end{align*}
We then show that
\begin{equation}
\label{eq-any-a}
\begin{split}
\forall a \in Q':\quad &\frac{\exp\left(\phi_v(a)+ \sum_{u \in \Gamma_v}\phi_{uv}(\sigma_u,a)\right)}{\exp\left(\phi'_v(a)+ \sum_{u \in \Gamma_v}\phi'_{uv}(\sigma_u,a)\right)} \geq \exp\left(- \Vert\phi_v-\phi'_v\Vert_1 - \sum_{e=\{u,v\}\in E}\Vert \phi_e-\phi'_e \Vert_1 \right)\\
\forall a \in Q':\quad &\frac{\exp\left(\phi'_v(a)+ \sum_{u \in \Gamma_v}\phi'_{uv}(\sigma_u,a)\right)}{\exp\left(\phi_v(a)+ \sum_{u \in \Gamma_v}\phi_{uv}(\sigma_u,a)\right)} \geq \exp\left(- \Vert\phi_v-\phi'_v\Vert_1 - \sum_{e=\{u,v\}\in E}\Vert \phi_e-\phi'_e \Vert_1 \right)
\end{split}
\end{equation}
We first use~\eqref{eq-any-a} to prove the~\eqref{eq-proof-decay}. 
Since $\mu_{v,\I}(c\mid \sigma) > 0$, then we have $c \in Q'$.
By~\eqref{eq-any-a}, we have
\begin{align*}
1 - \frac{\mu_{v,\I'}(c \mid \sigma)}{\mu_{v,\I}(c\mid\sigma)} &\leq 1 - \exp\left(- 2\Vert\phi_v-\phi'_v\Vert_1 - 2\sum_{e=\{u,v\}\in E}\Vert \phi_e-\phi'_e \Vert_1 \right)\\
&\leq 2\left( \Vert\phi_v-\phi'_v\Vert_1 + \sum_{e=\{u,v\}\in E}\Vert \phi_e-\phi'_e \Vert_1 \right).
\end{align*}
This proves the lemma.

We now prove~\eqref{eq-any-a}. For any $a \in Q'$, it holds that
\begin{align*}
\frac{\exp\left(\phi_v(a)+ \sum_{u \in \Gamma_v}\phi_{uv}(\sigma_u,a)\right)}{\exp\left(\phi'_v(a)+ \sum_{u \in \Gamma_v}\phi'_{uv}(\sigma_u,a)\right)} = \exp\left(\phi_v(a) - \phi'_v(a) + \sum_{u \in \Gamma_v}\phi_{uv}(\sigma_u,a)- \sum_{u \in \Gamma_v}\phi'_{uv}(\sigma_u,a)\right).
\end{align*}
Then~\eqref{eq-any-a} holds because 
\begin{align*}
\phi_v(a) - \phi'_v(a) &\geq -\sum_{c \in Q}|\phi_v(c)-\phi'_v(c)| =- \Vert\phi_v-\phi'_v\Vert_1;\\
\sum_{u \in \Gamma_v}\phi_{uv}(\sigma_u,a)- \sum_{u \in \Gamma_v}\phi'_{uv}(\sigma_u,a) &\geq - \sum_{e=\{u,v\}\in E}\sum_{c,c' \in Q}|\phi_e(c,c')-\phi'_e(c,c')|=- \sum_{e=\{u,v\}\in E}\Vert \phi_e-\phi'_e \Vert_1.
\end{align*}
The lower bound of $\frac{\exp\left(\phi'_v(a)+ \sum_{u \in \Gamma_v}\phi'_{uv}(\sigma_u,a)\right)}{\exp\left(\phi_v(a)+ \sum_{u \in \Gamma_v}\phi_{uv}(\sigma_u,a)\right)}$ can be proved in a similar way.
\qed

\vspace{1em}
\paragraph{The cost of the coupling for \UpdateHamiltonian{} (proof of \Cref{lemma-upper-bound-R-Ham})}
By the definition of the indicator random variable $\gamma_t$ in~\eqref{eq-def-R-Ising}, we have
\begin{align*}
\Pr[\gamma_t = 1 \mid \mathcal{D}_{t-1}] &\leq \Pr\left[t \in \Isingset \mid \mathcal{D}_{t-1}\right] + \Pr\left[ v_t \in \Gamma_G^+(\mathcal{D}_{t-1}) \mid \mathcal{D}_{t-1}\right] \\
&\leq\frac{(\Delta+1)|\mathcal{D}_{t-1}|}{n} + \sum_{v \in V}\frac{\pup_v}{n}.
\end{align*}
By the definition of $\pup_v$ in ~\eqref{eq-def-Ising-up} and $\dham(\I,\I') =  \sum_{v\in V}\norm{\phi_v-\phi'_v}_1 + \sum_{e\in E}\norm{\phi_e- \phi'_e}_1 \leq L$, we have
\begin{align*}
\Pr[\gamma_t = 1 \mid \mathcal{D}_{t-1}] \leq \frac{(\Delta+1)|\mathcal{D}_{t-1}|}{n} + \frac{4L}{n}. 	
\end{align*}
By the definition of $\Rham \triangleq \sum_{t=1}^T \gamma_t$, we have
\begin{align}
\label{eq-bound-R-Ising}
\E{\Rham} = \sum_{t=1}^T\E{\gamma_t} =  \sum_{t=1}^{T}{\E{\E{\gamma_t \mid \mathcal{D}_{t-1}}}}\leq \sum_{t=1}^{T}\left(\frac{(\Delta+1)\E{|\mathcal{D}_{t-1}|}}{n} + \frac{4L}{n}\right).
\end{align}

Next, we bound the expectation $\E{|\D_t|}$.
Recall that 
the one-step local coupling for Hamiltonian update (\Cref{def-one-step-local-coupling-dynamic}) 
is implemented as follows. 
We first construct the random set $\Isingset \subseteq V$ in~\eqref{eq-def-Ising-set}.
In the $t$-th step, where $1\leq t \leq T$, given any $\X_{t-1}$ and $\Y_{t-1}$, the $\X_t$ and $\Y_t$ is generated as follows.
\begin{itemize}
\item Let $X'(u) = X_{t-1}(u)$ and $Y'(u) = Y_{t-1}(u)$ for all $u \in V \setminus \{v_t\}$, sample $(X'(v_t),Y'(v_t)) \in Q^2$ jointly from the optimal coupling $D^{\sigma,\tau}_{\opt,\I_{v_t}}$ of the marginal distributions $\mu_{v_t,\I}(\cdot\mid \sigma)$ and $\mu_{v_t,\I}(\cdot\mid \tau)$, where $\sigma = X_{t-1}(\Gamma_G(v_t))$  and $\tau = Y_{t-1}(\Gamma_G(v_t))$.
\item Let $\X_t = \X'$ and $\Y_t = \Y'$. If $t \in \Isingset$, update the value of $Y_t(v_t)$ using~\eqref{eq-couple-resample}.
\end{itemize}
Hence, for any vertex $v \in V$, $X_t(v)\neq Y_t(v)$ only if one of the following two events occurs (1) $X'(v) \neq Y'(v)$; (2) $v_t = v$ and $t \in \Isingset$. Then for any $v \in V$, we have
\begin{align}
\label{eq-bound-single}
\Pr[X_t(v) \neq Y_t(v) \mid \X_{t-1}, \Y_{t-1}] &\leq \Pr[X'(v) \neq Y'(v) \mid \X_{t-1},\Y_{t-1}] + \Pr[v=v_t \land t \in \Isingset \mid \X_{t-1},\Y_{t-1}]\notag\\
&=\Pr[X'(v) \neq Y'(v) \mid \X_{t-1},\Y_{t-1}] + \Pr[v=v_t \land t \in \Isingset],
\end{align}
where the equation holds because $v=v_t \land t \in \Isingset$ is independent of $\X_{t-1},\Y_{t-1}$.
Given $\X_{t-1},\Y_{t-1}$, the random pair $\X',\Y'$ are obtained by the one-step optimal coupling for Gibbs sampling on instance $\I$~(\Cref{def-one-step-coupling-static-Gibbs}). 
Since $\I$ satisfies the  Dobrushin-Shlosman condition with constant $0<\delta < 1$,
by \Cref{proposition-mixing-Gibbs}, we have
\begin{align}
\label{eq-use-mix}
\E{H(\X',\Y') \mid \X_{t-1},\Y_{t-1}} \leq 	\left(1 - \frac{\delta}{n}\right)H(\X_{t-1},\Y_{t-1}) = \left(1 - \frac{\delta}{n}\right)|\D_{t-1}|.
\end{align}
where $H(\X, \Y) = |\{v \in V \mid X(v) \neq Y(v)\}|$ denote the Hamming distance.
Combining~\eqref{eq-bound-single} and~\eqref{eq-use-mix}, 
\begin{align*}
\E{|\mathcal{D}_t| \mid \mathcal{D}_{t-1}} &\leq \sum_{v \in V} \Pr[X'(v) \neq Y'(v) \mid \D_{t-1} ] +\sum_{v \in V} \Pr[t \in \Isingset \land v = v_t \mid \mathcal{D}_{t-1}]\\
&\leq \left(1 - \frac{\delta}{n}\right)|\mathcal{D}_{t-1}| + \sum_{v \in V}\frac{\pup_v}{n}\\
(\text{by}~\eqref{eq-def-Ising-up} )\quad&\leq \left(1 - \frac{\delta}{n}\right)|\mathcal{D}_{t-1}| + \frac{2}{n}\sum_{v \in V} \left( \Vert\phi_v-\phi'_v\Vert_1 + \sum_{e=\{u,v\}\in E}\Vert \phi_e-\phi'_e \Vert_1 \right)\\
(\text{by}~\dham(\I,\I')\leq L)\quad&\leq \left(1 - \frac{\delta}{n}\right)|\mathcal{D}_{t-1}| + \frac{4L}{n}.
\end{align*}
Thus, we have
\begin{align*}
\E{|\mathcal{D}_{t}|} \leq  \left(1 - \frac{\delta}{n}\right)\E{|\mathcal{D}_{t-1}|} + \frac{4L}{n}.
\end{align*}
Note that $|\D_0| = 0$.
This implies
\begin{align}
\label{eq-bound-D-Ising}
\E{|\mathcal{D}_t|} \leq \frac{8L}{\delta}.  
\end{align}
Thus, by~\eqref{eq-bound-R-Ising}, we have
\begin{align*}
\E{\Rham} \leq \frac{20\Delta T L}{\delta n} =  O\left( \frac{\Delta T L}{\delta n} \right). 
\end{align*}
\qed

\subsubsection{Proofs for the coupling for graph update}
\label{section-proof-gra}
In this section, we prove  \Cref{lemma-upper-bound-R}.
\vspace{1em}
\paragraph{Cost of the coupling for \UpdateEdge{} (Proof of \Cref{lemma-upper-bound-R})}
By the definition of $\Rgra$ in~\eqref{eq-def-R} and the linearity of the expectation, we have
\begin{align*}
\E{\Rgra} &= \sum_{t= 1}^{T}\E{\gamma_t}=\sum_{t=1}^{T} \E{\E{\gamma_t \mid \D_{t-1}}}.
\end{align*}
Recall $\gamma_t = \one{v_t \in \mathcal{S} \cup \Gamma^+_G(\D_{t-1})}$ and $v_t \in V$ is uniformly at random given $\D_{t-1}$. Note that $|\Gamma^+_G(\D_{t-1})| \leq (\Delta+1)|\D_{t-1}|$ and $|\mathcal{S}| \leq 2\vert E \oplus E' \vert \leq 2L$. We have
\begin{align}
\label{eq-ER-upper}
\E{\Rgra} \leq \sum_{t=1}^{T}\E{\frac{(\Delta+1)|\D_{t-1}| + 2L}{n}} = \frac{(\Delta+1)}{n}\sum_{t=1}^{T}\E{|\D_{t-1}|} + \frac{2LT}{n}.
\end{align}
Suppose $\I'$ satisfies Dobrushin-Shlosman condition (\Cref{condition-Dobrushin}) with the constant $\delta > 0$, we claim
\begin{align}
\label{eq-D-upper}
\forall\, 0\leq t \leq T:\quad \E{|\D_{t}|} \leq \frac{8L}{\delta}. 	
\end{align}
Combining~\eqref{eq-ER-upper} and~\eqref{eq-D-upper}, we have
\begin{align*}
\E{\Rgra} \leq \frac{18\Delta L T}{ \delta n}  = 	O\left( \frac{\Delta LT}{n}\right).
\end{align*}
This proves the lemma. 

We now prove~\eqref{eq-D-upper}.
Let $(\X_t, \Y_t)_{t \geq 0}$ be the one-step local coupling for updating edges (Definition~\ref{def-one-step-local-coupling-graph}). 
 We claim the following result
\begin{align}
\label{eq-claim-decay}
\forall\,\sigma,\tau \in \Omega:\qquad \E{\,H(\X_t,\Y_t)\mid \X_{t-1} =\sigma \land \Y_{t-1}=\tau\,}
 \leq\left(1-\frac{\delta}{n}\right)\cdot H(\sigma, \tau) + \frac{4L}{n},
\end{align}
where $H(\sigma,\tau) = |\{v \in V \mid \sigma(v) \neq \tau(v) \}|$ denotes the Hamming distance.
Assume~\eqref{eq-claim-decay} holds. Taking expectation over $\X_{t-1}$ and $\Y_{t-1}$, we have
\begin{align}
\label{eq-recur}
\E{H(\X_t , \Y_t)}	\leq \left(1-\frac{\delta}{n}\right)\E{H(\X_{t-1}, \Y_{t-1})} + \frac{4L}{n}.
\end{align}
Note that $\X_0 = \Y_0$, we have
\begin{align}
\label{eq-init}
H(\X_0, \Y_0) = 0.	
\end{align}
Combining~\eqref{eq-recur} with~\eqref{eq-init} implies
\begin{align}
\label{eq-Dt-edge}
\forall\,  0\leq t\leq T:\quad	\E{|\D_t|} = \E{H(\X_t , \Y_t)}\leq \frac{8L}{\delta}.
\end{align}
This proves the claim in~\eqref{eq-D-upper}.

We finish the proof by proving the claim in~\eqref{eq-claim-decay}.
The main idea is to compare 
the one-step local coupling for updating edges (Definition~\ref{def-one-step-local-coupling-graph}) 
with the one-step optimal coupling for Gibbs sampling on instance $\I'$ (Definition~\ref{def-one-step-coupling-static-Gibbs}). 
Let $(\X'_t, \Y'_t)_{t \geq 0}$ be the coupling for Gibbs sampling on $\I'$. 
Since $\I'$ satisfies Dobrushin-Shlosman condition, by Proposition~\ref{proposition-mixing-Gibbs}, we have
\begin{align}
\label{eq-contra-Gibbs}
\forall\,\sigma, \tau \in \Omega = Q^V:\quad \E{\,H(\X'_t,\Y'_t)\mid \X'_{t-1} = \sigma \land \Y'_{t-1}=\tau\,} \leq \left(1-\frac{\delta}{n}\right)\cdot H(\sigma, \tau).
\end{align}
According to the coupling, we can rewrite the expectation in~\eqref{eq-contra-Gibbs} as follows:
\begin{align}
\label{eq-expect-1}
\E{H(\X'_t , \Y'_t) \mid \X'_{t-1} =\sigma \land \Y'_{t-1} =\tau} = \frac{1}{n}\sum_{v \in V}\E{H\left(\sigma^{v \gets C_v^{X'} }, \tau^{v \gets C_v^{Y'}}\right)}	,
\end{align}
where $(C^{X'}_v, C^{Y'}_v) \sim D^{\sigma,\tau}_{\opt,\I'_{v}}$, $D^{\sigma,\tau}_{\opt,\I'_{v}}$ is the optimal coupling between $\mu_{v,\I'}(\cdot\mid \sigma)$ and $\mu_{v,\I'}(\cdot\mid \tau)$, and the configuration $\sigma^{v \gets C_v^{X'} } \in Q^V$ is defined as
\begin{align*}
\sigma^{v \gets C_v^{X'} }(u) \triangleq\begin{cases}
C^{X'}_v &\text{if } u = v\\
\sigma(u) &\text{if } u \neq v	
\end{cases}
 \end{align*}
and the configuration $\tau^{v \gets C_v^{Y'}} \in Q^V$ is defined in a similar way.

Similarly, we can rewrite the expectation in~\eqref{eq-claim-decay} as follows:
\begin{align}
\label{eq-expect-2}
\E{H(\X_t , \Y_t) \mid \X_{t-1} =\sigma \land \Y_{t-1} =\tau} = \frac{1}{n}\sum_{v \in V}\E{H\left(\sigma^{v \gets C_v^{X} }, \tau^{v \gets C_v^{Y}}\right)}	,
\end{align}
where $(C^{X}_v, C^{Y}_v) \sim D_{\I_v,\I_v'}^{\sigma,\tau}$, where $D_{\I_v,\I_v'}^{\sigma,\tau}$ is the local coupling defined in~\eqref{eq:local-coupling-D-for-R}.

The following two properties hold for~\eqref{eq-expect-1} and~\eqref{eq-expect-2}.
\begin{itemize}
\item If $v \not \in \mathcal{S}$, by the definition of $D_{\I_v,\I'_v}^{\sigma,\tau}(\cdot,\cdot)$ in~\eqref{eq:local-coupling-D-for-R}, it holds that $D_{\I_v,\I'_v}^{\sigma,\tau} = D_{\opt,\I_v}^{\sigma,\tau}$. Hence
\begin{align*}
\forall v \not\in \mathcal{S}:\quad 	\E{H\left(\sigma^{v \gets C_v^{X'} }, \tau^{v \gets C_v^{Y'}}\right)} = \E{H\left(\sigma^{v \gets C_v^{X} }, \tau^{v \gets C_v^{Y}}\right)}.
\end{align*}
\item If $v \in \mathcal{S}$, then it holds that  $H(\sigma^{v \gets C_v^{X}},\sigma^{v \gets C_v^{X'}} ) \leq 1$ and $H(\tau^{v \gets C_v^{Y'}},\tau^{v \gets C_v^{Y}}) \leq 1$. By the triangle inequality of the Hamming distance, we have
\begin{align*}
H\left(\sigma^{v \gets C_v^{X} }, \tau^{v \gets C_v^{Y}}\right) &\leq 	H\left(\sigma^{v \gets C_v^{X}},\sigma^{v \gets C_v^{X'}} \right) +H\left(\sigma^{v \gets C_v^{X'} }, \tau^{v \gets C_v^{Y'}}\right)+ H\left(\tau^{v \gets C_v^{Y'}},\tau^{v \gets C_v^{Y}}\right) \\
&\leq H\left(\sigma^{v \gets C_v^{X'} }, \tau^{v \gets C_v^{Y'}}\right) + 2.
\end{align*}
This implies
\begin{align*}
\forall v \in \mathcal{S}:\quad
\E{H\left(\sigma^{v \gets C_v^{X} }, \tau^{v \gets C_v^{Y}}\right)} \leq \E{H\left(\sigma^{v \gets C_v^{X'} }, \tau^{v \gets C_v^{Y'}}\right)} + 2.
\end{align*}
\end{itemize}
Combining above two properties with~\eqref{eq-expect-1} and~\eqref{eq-expect-2}, we have for any $\sigma \in , \tau \in \Omega$,
\begin{align*}
&\E{H(\X_t , \Y_t) \mid \X_{t-1} =\sigma \land \Y_{t-1} =\tau}\\
=&\, \frac{1}{n}\sum_{v \in V}\E{H\left(\sigma^{v \gets C_v^{X} }, \tau^{v \gets C_v^{Y}}\right)}\\
\leq&\, \frac{1}{n}\sum_{v \not\in \mathcal{S}}\E{H\left(\sigma^{v \gets C_v^{X'} }, \tau^{v \gets C_v^{Y'}}\right)} + \frac{1}{n}\sum_{v \in \mathcal{S}}\left(\E{H\left(\sigma^{v \gets C_v^{X'} }, \tau^{v \gets C_v^{Y'}}\right)} + 2 \right)\\
(\ast)\quad\leq&\, \E{H(\X'_t , \Y'_t) \mid \X'_{t-1} =\sigma \land \Y'_{t-1} =\tau} + \frac{4L}{n}\\
 \leq&\, \left(1-\frac{\delta}{n}\right)\cdot H(\sigma, \tau) + \frac{4L}{n},
\end{align*}
where $(\ast)$ holds due to $|\mathcal{S}| \leq 2L$. This proves the claim in~\eqref{eq-claim-decay}.
\qed

\subsection{Implementation of the algorithms}
\label{section-proof-implement}
In this section, we prove the \Cref{claim-prep-single}, \Cref{claim-time-cost} and \Cref{claim-prep-multi} by giving the implementation of the algorithms.


\subsubsection{Proofs of \Cref{claim-prep-single} and \Cref{claim-prep-multi}}
We prove \Cref{claim-prep-multi}, then \Cref{claim-prep-single} can be proved in a similar way.

It is easy to verify the updated sample $\I'$, all the probabilities $(\pup_v)_{v \in V}$ in~(\ref{eq-def-Ising-up}), all middle instances  $\Imid,\I_1,\I_2$ in~\eqref{eq-def-Imid},~\eqref{eq-def-I-1},~\eqref{eq-def-I-2} can be computed with time cost $O(\Delta n)$. We focus on constructing $\Isingset_i$ for $1\leq i \leq N_{\min}$.

The multi-sample dynamic Gibbs sampling algorithm use the data structure in \Cref{theorem-DS} to maintain $N(n)$ independent Gibbs sampling chain on instance $\I$ represented by $\X^{(i)}_0$ and $\Exelog\left(\I, T \right)= \exelog{X^{(i)}}{v^{(i)}_t}{T} $ . To construct the random sets $\Isingset_i$ for $1\leq i \leq N_{\min}$, we need an additional data structure to maintain the following data. Define the set $H_v$ as 
\begin{align*}
H_v \triangleq \{(i,t) \in [N(n)]\times [T] \mid v^{(i)}_t = v\}.	
\end{align*}
$H_v$ contains all the transition steps in $N(n)$ independent chains that picks the vertex $v$. The algorithm uses an extra data structure $\+H$ to maintain all $(H_v)_{v \in V}$. The data structure $\+H$ contains $n$ balanced binary search trees $(\+H_v)_{v \in V}$, where each $\+H_v$ maintains the set $H_v$ in a similar way as in the main data structure in \Cref{theorem-DS}. Since $T = O(n \log n), N(n) \leq \mathrm{poly(n)}$, the space cost of $\+H$ is $O(nN(n)\log n)$ memory words, each of $O(\log n)$ bits, which is dominated by the space cost in \Cref{lemma-dynamic-gibbs-sampling-multi}. And the time cost of adding, deleting, and searching a transition step in $\+H$ is $O(\log^2 n)$. We need to update $\+H$ when $\I$ is updated to $\I'$. One can verify that such time cost is dominated by the time cost in \Cref{lemma-dynamic-gibbs-sampling-multi}.

Then for each $v \in V$, we pick each element in $H_v$ with probability $\pup_v$ to construct   the set
\begin{align*}
\+B_v \subseteq H_v.	
\end{align*}
This is the standard Bernoulli process. With the data structure $\+H_v$, the time complexity of constructing the set $\+B_v$ is $O(\abs{\+B_v}\log^2 n)$.
 Given all the  sets $\+B_v $, it is easy to construct  all the sets $\Isingset_i$. 
 Hence, 
 \begin{align*}
 \Tprem = O\left(\Delta n + \sum_{v \in V}\abs{\+B_v} \log^2 n\right)=	 O\left(\Delta n + \sum_{i=1}^{N_{\min}}\abs{\Isingset_i} \log^2 n\right).
 \end{align*}
 In the preparation stage of multi-sample dynamic Gibbs sampling algorithm, we first construct the $\Imid =(V,E,Q,\Phi^{\mathsf{mid}})$ as in~\eqref{eq-def-Imid}, 
and each $\Isingset_i$ ($1\leq i\leq N_{\min}$) is constructed with respect to $\I$ and $\Imid$.
Note that $\dham(\I,\Imid ) \leq \dham(\I,\I') \leq \Lham$. 
By (\ref{eq-def-Ising-up}), we have for each $1\leq i\leq N_{\min}$,
\begin{align*}
\E{\abs{\Isingset_i}} \leq \sum_{t = 1}^T \sum_{v \in V} \frac{\pup_v}{n} \leq \frac{4T\Lham}{n}.	
\end{align*}
This proves the claim.
\qed

\subsubsection{Proof of \Cref{claim-time-cost}}
We give the implementation of the update stage of the single-sample dynamic Gibbs sampling algorithm (\Cref{alg-dynamic-Gibbs}). The algorithm updates the MRF instance from $\I$ to $\I'$ as follows,
\begin{align*}
\I \quad \to \quad \Imid \quad \to \quad \I_1 \quad \to \quad \I_2 \quad \to \quad \I', 	
\end{align*}
where $\Imid$ is defined in~\eqref{eq-def-Imid}, $\I_1 = \I_1(\Imid,\I')$ is defined in~\eqref{eq-def-I-1}, and $\I_2 = \I_2(\Imid,\I')$ is defined in~\eqref{eq-def-I-2}.
Then the algorithm calls \LengthFix{} to modifies the length of the execution log from $T$ to $T'$. 

The preparation stage computes all probabilities $(\pup_v)_{v \in V}$ in~(\ref{eq-def-Ising-up}), the set $\Isingset$ in~(\ref{eq-def-Ising-set}), and all instances $\Imid,\I_1,\I_2$. Consider the time cost of the update stage. In the update from $\Imid$ to $\I_1$, we only add isolated vertices in $V' \setminus V$, using the data structure in \Cref{theorem-DS}, the expected time cost is
\begin{align*}
\E{T_{\Imid \to \I_1}} = O\left( \frac{\abs{V' \setminus V}}{\abs{V}}T_{\max}\log^2{T_{\max}}	\right) = O\left( \frac{\Lgra}{n}T_{\max}\log^2{T_{\max}}	\right).
\end{align*}
In the update from $\I_2$ to $\I'$, we only delete isolated vertices in $V \setminus V'$, thus
\begin{align*}
\E{T_{\Imid \to \I_1}} = O\left( \frac{\abs{V \setminus V'}}{\abs{V \cup V'}}T_{\max}\log^2{T_{\max}}	\right) = O\left( \frac{\Lgra}{n}T_{\max}\log^2{T_{\max}}	\right).
\end{align*}
It is also easy to observe that the expected time cost of \LengthFix{} is
\begin{align*}
\E{T_{\LengthFix}} = O\tp{ \Delta \abs{T-T'}\log^2 T_{\max} }.	
\end{align*}
We then prove that
\begin{align}
\E{T_{\I \to \Imid}} &= O\tp{\Delta\E{\Rham}\log^2T_{\max}}\label{eq-proof-claim-1}\\
\E{T_{\I_1 \to \I_2}} &= O\tp{\Delta\E{\Rgra}\log^2T_{\max}}\label{eq-proof-claim-2}.
\end{align}
Combining all the running time together proves~\Cref{claim-time-cost}.

We give the implementation of \Cref{alg-ham-update} to prove~\eqref{eq-proof-claim-1}. 
The \Cref{alg-constraint-update} can be implemented in a similar way to prove~\eqref{eq-proof-claim-2}.
Since $(\pup_v)_{v \in V}$ and $\Isingset$ are given, the running time of \Cref{alg-ham-update} is dominated by the while-loop. We implement \Cref{alg-ham-update} such that after each execution of the while-loop, the first $t_0$ transition steps of the Gibbs sampling on instance $\I$ is updated to the first $t_0$ transition steps of the Gibbs sampling on instance $\I'$, namely, $(\X_t)_{t=0}^{t_0}$ is updated to $(\Y_t)_{t=0}^{t_0}$, where $t_0$ is the variable in \Cref{alg-ham-update}. 
Recall the sets $\+D$ and $\+P$ in~\Cref{alg-ham-update}.
We need some temporary data structures:
\begin{itemize}
\item a balanced binary search tree $\+T$ to maintain the set $\+D$ and the configuration $X_{t_0 - 1}(\+D)$;
\item a heap $\+H_1$ to maintain the set $\+P$;
\item a heap $\+H_2$ such that once a vertex $v$ is added into $\+D$, the update times $\datasuccessor(t_0,u)$ for all $u \in \Gamma_G(v) \cup \{v\}$  are added into $\+H_2$, where $\datasuccessor$ is the operation of the data structure in \Cref{theorem-DS}.
\end{itemize}
\Cref{line-next-ham} can be implemented using $\+H_1,\+H_2,\+T$. And \Cref{line-sample-ham-1} and \Cref{line-flip-ham} can be implemented using $\+T$ and the main data structure in \Cref{theorem-DS}. Note that the time cost of each operation of $\+T$ is $O(\log n) = O(\log T_{\max})$. Also note that at most $\Delta \Rham$ elements can be added into $\+H_2$. Hence, all the time cost contributed by $\+H_2$ is $O(\Delta \Rham \log (\Delta \Rham)) = O(\Delta \Rham \log T_{\max})$. One can verify that the total running time is
\begin{align*}
T_{\I \to \Imid} = 	 O\tp{\Delta \Rham\log^2T_{\max}}.
\end{align*}
This proves~\eqref{eq-proof-claim-1}.
\qed

\subsection{Dynamic Gibbs sampling  for specific models}
\label{section-proof-model}
In this section, we apply our algorithm on Ising model, graph $q$-coloring, and hardcore model. We prove the following theorem.

\begin{theorem}
\label{theorem-model}
There exist dynamic sampling algorithms as stated in Theorem~\ref{theorem-sample-MRF} with the same space cost $O\left(n N(n) \log n\right)$, and expected time cost $O\left(\Delta^2( \Lgra + \Lham)N(n)\log ^ 3 n +\Delta n \right)$ for each update,  if the input instance $\I$ with $n$ vertices and the updated instance $\I'$ satisfying $\dgraph(\I,\I') \leq \Lgra = o(n), \dham(\I,\I') \leq \Lham$ both are: 
\begin{itemize}
\item Ising models with temperature $\beta$ and arbitrary local fields where $\exp(-2|\beta|) \geq 1 - \frac{2-\delta}{\Delta + 1}$;
\item proper $q$-colorings with $q\ge(2+\delta)\Delta$;
\item hardcore models with fugacity $\lambda\le\frac{2-\delta}{\Delta-2}$, but with an alternative time cost for each update 
\begin{align}
\label{eq-running-time-hardcore}
O\left(\Delta^3( \Lgra + \Lham)N(n)\log ^ 3 n +\Delta n \right),
\end{align}
\end{itemize}	
where $\delta>0$ is a constant, $\Delta = \max\{\Delta_{G},\Delta_{G'}\}$, $\Delta_G$ denotes the maximum degree of the input graph, and  $\Delta_{G'}$ denotes the maximum degree of the updated graph.
\end{theorem}

In \Cref{theorem-model}, the regime for Ising model and $q$-coloring match the Dobrushin-Shlosman condition, thus the results are corollaries of \Cref{theorem-sample-MRF}. The regime for hardcore model is better than the Dobrushin-Shlosman condition. We give the proof for hardcore model. 

We use $\I = (V, E, \lambda)$ to specify the hardcore model on graph $G=(V, E)$ with fugacity $\lambda$. 
A configuration of hardcore model is $\sigma \in \{0,1\}^V$, where $\sigma_v = 1$ indicates $v$ is occupied, $\sigma_v = 0$ indicates $v$ is unoccupied. If $\sigma$ forms an independent set, then $\mu_{\I}(\sigma) \propto \lambda^{\Vert \sigma \Vert}$; otherwise, $\mu_{\I}(\sigma) = 0$.
We need the following lemma proved by Vigoda's coupling technique~\cite{vigoda1999fast}.
\begin{lemma}
\label{lemma-hardcore}
Let $\delta >0$ be a constant.
Let $\I = (V, E, \lambda)$ be a hardcore instance, where $n = |V|$, and $\Omega_{\I}\triangleq\{ \sigma\in \{0,1\}^V\mid \mu_{\I}(\sigma)>0\}$.
Assume $\lambda \leq \frac{2-\delta}{\Delta - 2}$, where $\Delta$ is the maximum degree of $G=(V,E)$.
There exist a potential function $\rho_{\I}:\Omega_{\I} \times \Omega_{\I}\to\mathbb{R}_{\ge 0}$, where $\forall\sigma,\tau \in \Omega_{\I}$,
$\rho_\I(\sigma,\tau) = 0$ if $\sigma=\tau$ and $\rho_\I(\sigma,\tau) \ge  1$ if $\sigma\neq\tau$, and  $\mathrm{Diam}_{\I}\triangleq\max_{\sigma,\tau\in\Omega_\I}\rho_\I(\sigma,\tau)\le \Delta n$, such that the one-step optimal coupling (Definition~\ref{def-one-step-coupling-static-Gibbs})  $(\X_t,\Y_t)_{t\geq 0}$ of Gibbs sampling on $\I$ satisfies
\begin{enumerate}
\item (\textbf{step-wise decay}) for the coupling $(\X_t,\Y_t)_{t\geq 0}$ of Gibbs sampling, it holds that
\begin{align}
\forall\,\sigma, \tau \in \Omega_{\I}:\quad \E{\,\rho_{\I}(\X_t,\Y_t)\mid \X_{t-1} = \sigma \land \Y_{t-1}=\tau\,} \leq \mbox{$\left(1-\frac{\beta}{n}\right)$}\cdot\rho_{\I}(\sigma, \tau),\label{eq-contra-coupling}
\end{align}
where $\beta = \frac{1}{96\delta}$, which implies $\tau_{\mathsf{mix}}(\I, \epsilon) \leq \lceil \frac{n}{\beta}\log \frac{\mathrm{Diam}_{\I}}{\epsilon} \rceil = O(n\log \frac{n}{\epsilon})$.
\item (\textbf{up-bound to Hamming}) for all $\sigma, \tau \in \Omega_{\I}$, $H(\sigma,\tau) \leq \rho_{\I}(\sigma,\tau)$, where $H(\sigma, \tau)$ denotes the Hamming distance between $\sigma$ and $\tau$.
\item (\textbf{Lipschitz}) 
function $\rho_\I(\cdot,\cdot)$, seen as a function of $2n$ variables, is $K$-Lipschitz, that is, 
$$\max_{\sigma,\sigma',\tau,\tau'\in\Omega_{\I}}\left|\rho_{\I}(\sigma,\tau)-\rho_{\I}(\sigma',\tau')\right|\le K\cdot H(\sigma\tau,\sigma'\tau'),$$
where $K = 12 \Delta$.
\end{enumerate}	
\end{lemma}

Compared with \Cref{proposition-mixing-Gibbs}, the step-wise decay property in~\eqref{eq-contra-coupling} holds only for feasible configurations $\sigma$ and $\tau$, and the decay property is established on the potential function $\rho_{\I}$ rather than the Hamming distance $H$. We first use \Cref{lemma-hardcore} to prove \Cref{theorem-model}, then we prove \Cref{lemma-hardcore} in the end of this section. 

Recall that the error function $\epsilon$ satisfies $\epsilon(\ell) \geq \frac{1}{\mathrm{poly}(\ell)}$ by~\Cref{lemma-smooth-epsilon}.
Recall $\Delta = \max\{\Delta_G,\Delta_{G'}\}$.
By \Cref{lemma-hardcore} and $n' = \Theta(n)$ (since $\Lgra = o(n)$), we can set 
\begin{align*}
T &= T(\I) = \left\lceil\frac{96n}{\delta}\log \frac{n \Delta}{\epsilon(n)} \right\rceil = O\left(n\log n\right)\\
T' &= T(\I') = \left\lceil\frac{96n'}{\delta}\log \frac{n' \Delta}{\epsilon(n')} \right\rceil = O\left(n\log n\right).
\end{align*}

We modify~\Cref{alg-dynamic-Gibbs} for the hardcore model as follows. Suppose the current instance is $\I=(V,E,\lambda)$, we set the initial configuration $\X_0$ as
\begin{align*}
\forall v \in V, \quad X_0(v) = 0.	
\end{align*}
Thus $\X_0$ is feasible.
Suppose the instance $\I=(V,E,\lambda)$ is updated to $\I'=(V',E',\lambda')$. 
We divide the update into the following steps
\begin{align*}
\I \quad \to \quad \Imid \quad \to \quad \I_1 \quad \to \quad \I_{2} \quad \to \quad \I_3 \quad \to \quad \I',	
\end{align*}
\begin{itemize}
\item change fugacity to update $\I = (V,E,\lambda)$ to $\Imid=(V,E,\lambda')$ using \UpdateHamiltonian;
\item add isolated vertices in $V' \setminus V$ to update $\Imid = (V,E,\lambda')$ to $\I_1=(V\cup V',E,\lambda')$ using \AddVertex;
\item delete edges in $E \setminus E'$ to update $\I_1=(V\cup V',E,\lambda')$ to  $\I_2=(V\cup V',E\cap E' ,\lambda')$ using \UpdateEdge; 
\item add edges in $E' \setminus E$ to update $\I_2=(V\cup V',E \cap E' ,\lambda')$ to  $\I_3=(V\cup V',E',\lambda')$ using \UpdateEdge; 
\item delete isolated vertices in $V' \setminus V$ to update  $\I_3=(V\cup V',E',\lambda')$ to $\I'=(V',E',\lambda')$;
\item fix the length of the execution log from $T$ to $T'$.
\end{itemize}
Compared to~\Cref{alg-dynamic-Gibbs}, we further divide the update of edges into two steps: at first delete edges, then add edges. 
Thus, we have the following observation.
\begin{observation}
\label{observation-hardcore}
The following results holds:
\begin{itemize}
\item $\Omega_{\I} = \Omega_{\Imid}$, $\Omega_{\I_1} \subseteq \Omega_{\I_2}$ and $\Omega_{\I_3} \subseteq \Omega_{\I_2}$, where $\Omega_{\+J}$ is the set of feasible configurations for any instance $\+J$.
\item the instances $\I,\I_2,\I_3,\I'$ all satisfy 	 $\lambda \leq \frac{2-\delta}{\Delta - 2}$, where $\lambda$ and $\Delta$ are the fugacity and maximum degree of the corresponding instance. 
\end{itemize}
\end{observation}
\noindent
By the observation, we know that $\Omega_{\I} = \Omega_{\Imid}$, $\Omega_{\I_1} \subseteq \Omega_{\I_2}$ and $\Omega_{\I_3} \subseteq \Omega_{\I_2}$, thus we can use \Cref{lemma-hardcore}, because the step-wise decay property~\eqref{eq-contra-coupling} is established only on feasible configurations.

We need to analyze $\Rham$ and $\Rgra$ defined in~\eqref{eq-def-R-Ising} and~\eqref{eq-def-R} for the hardcore model. 
We prove the following two lemmas for hardcore model.

\begin{lemma}
\label{lemma-hardcore-1}
Consider $\UpdateHamiltonian\left(\I,\I',\*X_0, \exelog{X}{v_t}{T} \right)$.
Let $\I=(V,E,\lambda)$ be the current instance and $\I'=(V,E,\lambda')$ the updated instance.
Assume  $\lambda \leq \frac{2-\delta}{\Delta - 2}$, where $\delta > 0$ is a constant and $\Delta$ is the maximum degree of $G=(V,E)$. Also assume $\dham(\I,\I') =  n\abs{\ln\lambda-\ln\lambda'} \leq L$.  
Then $\E{\Rham} = O\left( \frac{\Delta^2 T L}{n\delta}  \right)$,
where $n= V$, $\Delta$ is the maximum degree of graph $G=(V,E)$.
\end{lemma} 

\begin{lemma}\label{lemma-hardcore-2}
Consider $\UpdateEdge\left(\I,\I',\*X_0, \exelog{X}{v_t}{T} \right)$.
Let $\I=(V,E,\lambda)$ be the current instance and $\I'=(V,E',\lambda)$ the updated instance.
Assume $|E \oplus E'| \leq L$. Also assume one of the following two conditions holds for some constant $\delta > 0$:
\begin{itemize}
\item  $\lambda \leq \frac{2-\delta}{\Delta_G - 2}$ and $\Omega_{\I'} \subseteq \Omega_{\I}$, 	where  $\Delta_G$ is the maximum degree of $G=(V,E)$;
\item  $\lambda \leq \frac{2-\delta}{\Delta_{G'} -2}$ and $\Omega_{\I} \subseteq \Omega_{\I'}$, 	where  $\Delta_{G'}$ is the maximum degree of $G'=(V,E')$.
\end{itemize}
Then $\E{\Rgra} = O\left( \frac{\Delta^2 T L}{n\delta}  \right)$,
where $n= V$, $\Delta = \max\{\Delta_G,\Delta_{G'}\}$.
\end{lemma}

Note that we call the subroutine \UpdateHamiltonian{} for the update modifying $\I$ to $\Imid$. By \Cref{observation-hardcore}, the condition in \Cref{lemma-hardcore-1} holds. 
We call the subroutine \UpdateEdge{} for the update modifying $\I_1$ to $\I_2$ and the update modifying $\I_2$ to $\I_3$. By \Cref{observation-hardcore},
in both two calls of \UpdateEdge{}, the condition in \Cref{lemma-hardcore-2} holds.
Then \Cref{theorem-model} for hardcore can by proved by going through the proof in 
\Cref{sec:agorithm-dynamic-Gibbs}.
 Compared to ~\Cref{lemma-upper-bound-R-Ham} and ~\Cref{lemma-upper-bound-R}, $\E{\Rham},\E{\Rgra}$ in \Cref{lemma-hardcore-1} and \Cref{lemma-hardcore-2} are bounded by  $O\left( \frac{\Delta^2 T L}{n\delta}  \right)$ rather than  $O\left( \frac{\Delta T L}{n\delta}  \right)$. This is why the hardcore model has an alternative running time in~\eqref{eq-running-time-hardcore}.

The proofs of \Cref{lemma-hardcore-1} and \Cref{lemma-hardcore-2} are similar to the proofs of \Cref{lemma-upper-bound-R-Ham} and \Cref{lemma-upper-bound-R}. We give the proofs here for the completeness.
\begin{proof}[Proof of \Cref{lemma-hardcore-1}]
By the definition of the indicator $\gamma_t$ in~\eqref{eq-def-R-Ising}, we have
\begin{align*}
\Pr[\gamma_t = 1 \mid \mathcal{D}_{t-1}] &\leq \Pr\left[t \in \Isingset \right] + \Pr\left[ v_t \in \Gamma_G^+(\mathcal{D}_{t-1}) \right]=\frac{(\Delta+1)|\mathcal{D}_{t-1}|}{n} + \sum_{v \in V}\frac{\pup_v}{n}.
\end{align*}
By the definition of $\pup_v$ in ~\eqref{eq-def-Ising-up} and $\dham(\I,\I') =  n\abs{\ln \lambda - \ln \lambda'} \leq L$, we have
\begin{align*}
\Pr[\gamma_t = 1 \mid \mathcal{D}_{t-1}] \leq \frac{(\Delta+1)|\mathcal{D}_{t-1}|}{n} + \frac{2L}{n}. 	
\end{align*}
By the definition of $\Rham \triangleq \sum_{t=1}^T \gamma_t$, we have
\begin{align}
\label{eq-bound-R-hardcore}
\E{\Rham} = \sum_{t=1}^T\E{\gamma_t} =  \sum_{t=1}^{T}{\E{\E{\gamma_t \mid \mathcal{D}_{t-1}}}}\leq \sum_{t=1}^{T}\left(\frac{(\Delta+1)\E{|\mathcal{D}_{t-1}|}}{n} + \frac{2L}{n}\right).
\end{align}
Next, we bound the expectation $\E{|\D_t|}$.
In our implementation of the one-step local coupling for Hamiltonian update (\Cref{def-one-step-local-coupling-dynamic}), we first construct the random set $\Isingset \subseteq V$ in~\eqref{eq-def-Ising-set}.
In the $t$-th step, where $1\leq t \leq T$, given any $\X_{t-1}$ and $\Y_{t-1}$, the $\X_t$ and $\Y_t$ is generated as follows.
\begin{itemize}
\item Let $X'(u) = X_{t-1}(u)$ and $Y'(u) = Y_{t-1}(u)$ for all $u \in V \setminus \{v_t\}$, sample $(X'(v_t),Y'(v_t)) \in \{0,1\}^2$ jointly from the optimal coupling $D^{\sigma,\tau}_{\opt,\I_{v_t}}$ of the marginal distributions $\mu_{v_t,\I}(\cdot\mid \sigma)$ and $\mu_{v_t,\I}(\cdot\mid \tau)$, where $\sigma = X_{t-1}(\Gamma_G(v_t))$  and $\tau = Y_{t-1}(\Gamma_G(v_t))$.
\item Let $\X_t = \X'$ and $\Y_t = \Y'$. If $t \in \Isingset$, update the value of $Y_t(v_t)$ using~\eqref{eq-couple-resample}.
\end{itemize}
Note that $\Omega_{\I} = \Omega_{\I'}$.
Since $\I$ satisfies $\lambda \leq \frac{2-\delta}{\Delta - 2}$ with constant $\delta > 0$,
by \Cref{lemma-hardcore}, for any feasible $\X_{t-1}  ,\Y_{t-1}  \in \Omega_{\I} = \Omega_{\I'}$, we have
\begin{align}
\label{eq-use-mix-hardcore}
\E{\rho_{\I}(\X',\Y') \mid \X_{t-1},\Y_{t-1}} \leq 	\left(1 - \frac{\delta}{96n}\right)\rho_{\I}(\X_{t-1},\Y_{t-1}). 
\end{align}
By \Cref{lemma-hardcore}, function $\rho_\I(\cdot,\cdot)$, seen as a function of $2n$ variables, is $12\Delta$-Lipschitz. Let $\+F$ indicates whether $t \in \+P$. We flip the value of $Y_t(v_t)$ only if $\+F$ occurs.
By~\eqref{eq-use-mix-hardcore}, we have 
\begin{align*}
\E{\rho_{\I}(\X_t,\Y_t) \mid \X_{t-1},\Y_{t-1}} &\leq \E{\rho_{\I}(\X',\Y')  + 12\Delta \+F \mid \X_{t-1},\Y_{t-1}}\\
&= \E{\rho_{\I}(\X',\Y') \mid \X_{t-1},\Y_{t-1}}  + \E{12\Delta \+F \mid \X_{t-1},\Y_{t-1}}\\
(\text{$\+F$ is independent with $\X_{t-1},\Y_{t-1}$})\quad&\leq 	\left(1 - \frac{\delta}{96n}\right)\rho_{\I}(\X_{t-1},\Y_{t-1}) + 12\Delta \E{\+F}\\
&\leq \left(1 - \frac{\delta}{96n}\right)\rho_{\I}(\X_{t-1},\Y_{t-1})  + 12\Delta\sum_{v \in V}\frac{\pup_v}{n}\\
(\text{by}~\eqref{eq-def-Ising-up} )\quad&\leq  \left(1 - \frac{\delta}{96n}\right)\rho_{\I}(\X_{t-1},\Y_{t-1})  + \frac{24\Delta}{n}\sum_{v \in V} \abs{\ln \lambda - \ln \lambda'} \\
(\text{by}~\dham(\I,\I')\leq L)\quad&\leq  \left(1 - \frac{\delta}{96n}\right)\rho_{\I}(\X_{t-1},\Y_{t-1})  + \frac{24L\Delta}{n}.
\end{align*}
Note that $\rho_{\I}(\X_0,\Y_0) = 0$ and $\X_0(v) = \Y_0(v) = 0$ for all $v \in V$, the configurations $\X_t,\Y_t$ are feasible for all $t \geq 0$.
Thus, we have
\begin{align*}
\E{\rho_{\I}(\X_t,\Y_t) } \leq  \left(1 - \frac{\delta}{96n}\right)\E{\rho_{\I}(\X_{t-1},\Y_{t-1}) } + \frac{24L\Delta}{n}.
\end{align*}
Thus $\E{\rho_{\I}(\X_t,\Y_t)} \leq \frac{5000L\Delta}{\delta}. $
By the up-bound to Hamming in \Cref{lemma-hardcore}, we have
\begin{align*}
\E{|\mathcal{D}_t|} \leq \frac{5000L\Delta}{\delta}.
\end{align*}
Thus, by~\eqref{eq-bound-R-hardcore}, we have
\begin{align*}
\E{\Rham} \leq \frac{50000\Delta^2 T L}{\delta n} =  O\left( \frac{\Delta^2 T L}{\delta n} \right). 
\end{align*}	
\end{proof}

\begin{proof}[Proof of \Cref{lemma-hardcore-2}]
By the definition of $\Rgra$ in~\eqref{eq-def-R} and the linearity of the expectation, we have
\begin{align*}
\E{\Rgra} &= \sum_{t= 1}^{T}\E{\gamma_t}=\sum_{t=1}^{T} \E{\E{\gamma_t \mid \D_{t-1}}}.
\end{align*}
Recall $\gamma_t = \one{v_t \in \mathcal{S} \cup \Gamma^+_G(\D_{t-1})}$ and $v_t \in V$ is uniformly at random given $\D_{t-1}$. Note that $|\Gamma^+_G(\D_{t-1})| \leq (\Delta+1)|\D_{t-1}|$ and $|\mathcal{S}| \leq 2\vert E \oplus E' \vert \leq 2L$. We have
\begin{align}
\label{eq-ER-upper-harecore}
\E{\Rgra} \leq \sum_{t=1}^{T}\E{\frac{(\Delta+1)|\D_{t-1}| + 2L}{n}} = \frac{(\Delta+1)}{n}\sum_{t=1}^{T}\E{|\D_{t-1}|} + \frac{2LT}{n}.
\end{align}
Suppose $\lambda \leq \frac{2-\delta}{\Delta_G - 2}$ and $\Omega_{\I'} \subseteq \Omega_{\I}$. The other condition follows from symmetry. We claim that
\begin{align}
\label{eq-D-upper-hardcore}
\forall\, 0\leq t \leq T:\quad \E{|\D_{t}|} \leq \frac{10000\Delta L}{\delta}. 	
\end{align}
Combining~\eqref{eq-ER-upper-harecore} and~\eqref{eq-D-upper-hardcore}, we have
\begin{align*}
\E{\Rgra} \leq \frac{100000\Delta L T}{n\delta}  = 	O\left( \frac{\Delta^2 LT}{ n\delta}\right).
\end{align*}
This proves the lemma. 

We now prove~\eqref{eq-D-upper-hardcore}.
Let $(\X_t, \Y_t)_{t \geq 0}$ be the one-step local coupling for updating edges (Definition~\ref{def-one-step-local-coupling-graph}). 
 We claim the following result
\begin{align}
\label{eq-claim-decay-hardcore}
\forall\,\sigma \in \Omega_{\I}, \tau \in \Omega_{\I'} \subseteq \Omega_{\I}, \E{\,\rho_{\I}(\X_t,\Y_t)\mid \X_{t-1} =\sigma \land \Y_{t-1}=\tau\,}
 \leq\left(1-\frac{\delta}{96n}\right)\cdot \rho_{\I}(\sigma, \tau) + \frac{48\Delta L}{n},
\end{align}
where $\rho_{\I}$ is the potential function in \Cref{lemma-hardcore}.
Assume~\eqref{eq-claim-decay-hardcore} holds.
Since $\X_0 = \Y_0 = \{0\}^V$ and $\Omega_{\I'} \subseteq \Omega_{\I}$, we must have $\X_{t-1},\Y_{t-1} \in \Omega_{\I}$.
Taking expectation over $\X_{t-1}$ and $\Y_{t-1}$, we have
\begin{align}
\label{eq-recur-hardcore}
\E{\rho_{\I}(\X_t , \Y_t)}	\leq \left(1-\frac{\delta}{96n}\right)\E{\rho_{\I}(\X_{t-1}, \Y_{t-1})} + \frac{48\Delta L}{n}.
\end{align}
Note that $\X_0 = \Y_0$, we have
\begin{align}
\label{eq-init-hardcore}
\rho_{\I}(\X_0, \Y_0) = 0.	
\end{align}
Combining~\eqref{eq-recur-hardcore}, ~\eqref{eq-init-hardcore} and upper-bound Hamming property in \Cref{lemma-hardcore}  implies
\begin{align*}
\forall\,  0\leq t\leq T:\quad \E{\abs{\+D_t}} \leq \E{\rho_{\I}(\X_t , \Y_t)}\leq \frac{10000\Delta L}{\delta}.
\end{align*}
This proves the claim in~\eqref{eq-D-upper-hardcore}.

We finish the proof by proving the claim in~\eqref{eq-claim-decay-hardcore}.
Let $(\X'_t, \Y'_t)_{t \geq 0}$ be the one-step optimal coupling for Gibbs sampling on instance $\I$ (Definition~\ref{def-one-step-coupling-static-Gibbs}). 
Since $\I$ satisfies $\lambda \leq \frac{2-\delta}{\Delta_G - 2}$, by \Cref{lemma-hardcore}, we have
\begin{align}
\label{eq-contra-Gibbs-hardcore}
\forall\,\sigma , \tau \in \Omega_{\I} :\quad \E{\,\rho_{\I}(\X'_t,\Y'_t)\mid \X'_{t-1} = \sigma \land \Y'_{t-1}=\tau\,} \leq \left(1-\frac{\delta}{96n}\right)\cdot \rho_{\I}(\sigma, \tau).
\end{align}
According to the coupling, we can rewrite the expectation in~\eqref{eq-contra-Gibbs-hardcore} as follows:
\begin{align}
\label{eq-expect-1-hardcore}
\E{\rho_{\I}(\X'_t , \Y'_t) \mid \X'_{t-1} =\sigma \land \Y'_{t-1} =\tau} = \frac{1}{n}\sum_{v \in V}\E{\rho_{\I}\left(\sigma^{v \gets C_v^{X'} }, \tau^{v \gets C_v^{Y'}}\right)}	,
\end{align}
where $(C^{X'}_v, C^{Y'}_v) \sim D^{\sigma,\tau}_{\opt,\I_{v}}$, $D^{\sigma,\tau}_{\opt,\I_{v}}$ is the optimal coupling between $\mu_{v,\I}(\cdot\mid \sigma)$ and $\mu_{v,\I}(\cdot\mid \tau)$, and the configuration $\sigma^{v \gets C_v^{X'} } \in Q^V$ is defined as
\begin{align*}
\sigma^{v \gets C_v^{X'} }(u) \triangleq\begin{cases}
C^{X'}_v &\text{if } u = v\\
\sigma(u) &\text{if } u \neq v	
\end{cases}
 \end{align*}
and the configuration $\tau^{v \gets C_v^{Y'}} \in Q^V$ is defined in a similar way.

Similarly, we can rewrite the expectation in~\eqref{eq-claim-decay-hardcore} as follows:
\begin{align}
\label{eq-expect-2-hardcore}
\E{\rho_{\I}(\X_t , \Y_t) \mid \X_{t-1} =\sigma \land \Y_{t-1} =\tau} = \frac{1}{n}\sum_{v \in V}\E{\rho_{\I}\left(\sigma^{v \gets C_v^{X} }, \tau^{v \gets C_v^{Y}}\right)}	,
\end{align}
where $(C^{X}_v, C^{Y}_v) \sim D_{\I_v,\I_v'}^{\sigma,\tau}$, where $D_{\I_v,\I_v'}^{\sigma,\tau}$ is the local coupling defined in~\eqref{eq:local-coupling-D-for-R}.

The following two properties hold for~\eqref{eq-expect-1-hardcore} and~\eqref{eq-expect-2-hardcore}.
\begin{itemize}
\item If $v \not \in \mathcal{S}$, by the definition of $D_{\I_v,\I'_v}^{\sigma,\tau}(\cdot,\cdot)$ in~\eqref{eq:local-coupling-D-for-R}, it holds that $D_{\I_v,\I'_v}^{\sigma,\tau} = D_{\opt,\I_v}^{\sigma,\tau}$. Hence
\begin{align*}
\forall v \not\in \mathcal{S}:\quad 	\E{\rho_{\I}\left(\sigma^{v \gets C_v^{X'} }, \tau^{v \gets C_v^{Y'}}\right)} = \E{\rho_{\I}\left(\sigma^{v \gets C_v^{X} }, \tau^{v \gets C_v^{Y}}\right)}.
\end{align*}
\item If $v \in \mathcal{S}$, then it holds that  $H(\sigma^{v \gets C_v^{X}},\sigma^{v \gets C_v^{X'}} ) \leq 1$ and $H(\tau^{v \gets C_v^{Y'}},\tau^{v \gets C_v^{Y}}) \leq 1$, where $H$ is the Hamming distance. Since $\Omega_{\I'} \subseteq \Omega_{\I}$, it holds that $\sigma^{v \gets C_v^{X'}},\sigma^{v \gets C_v^{X}}, \tau^{v \gets C_v^{Y}}, \tau^{v \gets C_v^{Y'}} \in \Omega_{\I}$. Since the function $\rho_{\I}$ is $12\Delta$-Lipschitz, we have
\begin{align*}
\forall v \in \mathcal{S}:\quad
\E{\rho_{\I}\left(\sigma^{v \gets C_v^{X} }, \tau^{v \gets C_v^{Y}}\right)} \leq \E{\rho_{\I}\left(\sigma^{v \gets C_v^{X'} }, \tau^{v \gets C_v^{Y'}}\right)} + 24\Delta.
\end{align*}
\end{itemize}
Combining above two properties with ~\eqref{eq-contra-Gibbs-hardcore},~\eqref{eq-expect-1-hardcore} and~\eqref{eq-expect-2-hardcore}, we have for any $\sigma \in , \tau \in \Omega$,
\begin{align*}
&\E{\rho_{\I}(\X_t , \Y_t) \mid \X_{t-1} =\sigma \land \Y_{t-1} =\tau}\\
=&\, \frac{1}{n}\sum_{v \in V}\E{\rho_{\I}\left(\sigma^{v \gets C_v^{X} }, \tau^{v \gets C_v^{Y}}\right)}\\
\leq&\, \frac{1}{n}\sum_{v \not\in \mathcal{S}}\E{\rho_{\I}\left(\sigma^{v \gets C_v^{X'} }, \tau^{v \gets C_v^{Y'}}\right)} + \frac{1}{n}\sum_{v \in \mathcal{S}}\left(\E{\rho_{\I}\left(\sigma^{v \gets C_v^{X'} }, \tau^{v \gets C_v^{Y'}}\right)} + 24\Delta \right)\\
(\ast)\quad\leq&\, \E{\rho_{\I}(\X'_t , \Y'_t) \mid \X'_{t-1} =\sigma \land \Y'_{t-1} =\tau} + \frac{48L\Delta}{n}\\
 \leq&\, \left(1-\frac{\delta}{96n}\right)\cdot \rho_{\I}(\sigma, \tau) + \frac{48L\Delta}{n},
\end{align*}
where $(\ast)$ holds due to $|\mathcal{S}| \leq 2L$. This proves the claim in~\eqref{eq-claim-decay-hardcore}.
\end{proof}

Finally, we prove \Cref{lemma-hardcore}. This proof is based on the coupling technique in~\cite{vigoda1999fast}. 
\begin{proof}[Proof of \Cref{lemma-hardcore}]
We give a potential function $\rho_{\I}$ for the hard core instance $\I$. 
%
%
We mainly use Vigoda's potential function in~\cite{vigoda1999fast}.
However, we need to slightly modify Vigoda's potential function to handle the isolated vertices.

Recall that for hard core model, $Q = \{0,1\}$. For any $\sigma \in Q^V$, $\sigma(v) = 1$ represents $v$ is occupied and $\sigma(v) = 0$ represents $v$ is unoccupied.
For each vertex $v \in V$, we use $\mathrm{deg}(v)$ to denote the degree of $v$ in graph $G=(V, E)$.
We divide the graph $G = (V, E)$ into two graphs $G_1=(V_1,E_1)$ and $G_2=(V_2,E_2)$ such that
\begin{align*}
&V_1 = \{v\in V\mid\deg(v) = 0\},\quad  E_1 = \emptyset,\\
&V_2	= V \setminus V_1,\quad  E_2 = E.
\end{align*}
Thus $G_1$ is an empty graph and $G_2$ contains no isolated vertex.
The potential function $\rho_{\I}$ is defined as
\begin{align*}
\forall \sigma, \tau \in \Omega_{\I}:\quad \rho_{\I}(\sigma,\tau) \triangleq 4\rho_1(\sigma(V_1),\tau(V_1)) + 4\rho_2(\sigma(V_2),\tau(V_2)).
\end{align*}
Here, $\rho_1$ is the potential function on $G_1$, which is the Hamming distance:
\begin{align*}
\rho_1(\sigma(V_1),\tau(V_1)) = \sum_{v \in V_1}\one{\sigma(v) \neq \tau(v)}.	
\end{align*}
And $\rho_2(\sigma(V_2),\tau(V_2))$ is the Vigoda's potential function~\cite{vigoda1999fast} on the graph $G_2$. Formally, let $D = \{v \in V_2 \mid \sigma(v) \neq \tau(v)\}$. For each $v \in V_2$, let $d_v = |D \cap \Gamma_{G_2}(v)|$. Let $c = \frac{\Delta\lambda}{\Delta\lambda + 2}$, where $\Delta$ is the maximum degree of graph $G$. Note that the maximum degree of graph $G_2$ is also $\Delta$. The potential  function $\rho_2(\sigma(V_2),\tau(V_2))$ is defined as
\begin{align*}
\alpha_v = \begin{cases}
 \deg(v) &\text{if } v \in D\\
 0 &\text{otherwise};	
 \end{cases}\quad
\beta_v = \begin{cases}
-c d_v &\text{if $\exists\,w \in \Gamma_{G_2}(v)$ such that $\sigma(w) = \tau(w) = 1$}\\
-c (d_v - 1) &\text{if there is no such $w$ and $d_v>1$ }\\
0& \text{otherwise};
 \end{cases}
\end{align*}
\begin{align*}
\rho_2(\sigma(V_2),\tau(V_2)) = \sum_{v \in V_2}(\alpha_v + \beta_v).	
\end{align*}
It is easy to see $\rho_{\I}(\sigma,\sigma) = 0$ and $\max_{\sigma,\tau \in \Omega_{\I}}\rho_{\I}(\sigma, \tau) = \Delta n$.
We then verify other properties for $\rho_{\I}$.

At first, we prove the upper-bound to Hamming property. For function $\rho_1$, it holds that 
\begin{align*}
\rho_1(\sigma(V_1),\tau(V_1)) = H(\sigma(V_1),\tau(V_1)). 	
\end{align*}
For function $\rho_2$, it holds that 
\begin{align*}
\rho_2(\sigma(V_2),\tau(V_2)) = \sum_{v \in V_2}(\alpha_v + \beta_v) = \sum_{v \in D}\alpha_v + \sum_{v \in V_2}\beta_v	 \geq \sum_{v \in D}\sum_{w \in \Gamma_{G_2}(v)}(1-c),
\end{align*}
where the last inequality holds due to $\sum_{v \in V_2}\beta_v \geq -\sum_{v \in V_2}cd_v = -c\sum_{v \in D}\deg(v)$.
Since graph $G_2$ contains no isolated vertex, then $|\Gamma_{G_2}(v)| = \deg(v) \geq 1$ for all $v \in D$. Note $c < 1$. Thus
\begin{align*}
\rho_2(\sigma(V_2),\tau(V_2)) \geq |D|(1-c)	= |D|\frac{2}{\Delta\lambda + 2} \geq \frac{|D|}{4} =\frac{1}{4}H(\sigma(V_2),\tau(V_2)),
\end{align*}
where $\frac{2}{\lambda \Delta + 2} \geq \frac{1}{4}$ is because $\lambda < \frac{2}{\Delta-2}$ and $\Delta \geq 3$.
Combining together we have
\begin{align*}
\rho_{\I}(\sigma, \tau) = 	4\rho_1(\sigma(V_1),\tau(V_1)) + 4\rho_2(\sigma(V_2),\tau(V_2)) \geq H(\sigma, \tau).
\end{align*}
This also implies $\rho_{\I}(\sigma, \tau) \geq \one{\sigma \neq \tau}$.

Next, we show the function $\rho_{\I}$ is  $12\Delta$-Lipschitz.  
Recall $V_1 \cap V_2 = \emptyset$, $V_1 \cup V_2 = V$ and
\begin{align*}
\rho_{\I}(\sigma, \tau) = 4\rho_1(\sigma(V_1),\tau(V_1)) + 4\rho_2(\sigma(V_2), \tau(V_2)).	
\end{align*}
Since $\rho_1$ is the Hamming distance, it is easy to see $\rho_1$ is $1$-Lipschitz. 
To give the Lipschitz constant for $\rho_2$. We extend the function $\rho_2$ as follows.
Suppose the function $\rho_2$ is defined over $Q^{V_2} \times Q^{V_2}$, where $Q = \{0, 1\}$. 
For any $x,y,x',y' \in Q^{V_2}$ such that $H(xy,x'y') = 1$, it is easy to verify the extended function $\rho_2$ satisfies
\begin{align*}
|\rho_2(x,y) - \rho_2(x',y')| \leq 3\Delta.
\end{align*}
This implies the original function $\rho_2$ is $3\Delta$-Lipschitz.
Hence, the function $\rho_{\I}$ is  $12\Delta$-Lipschitz.

Finally, we prove the step-wise decay property.
Let $(\X_t^{(1)})_{t\geq 0},(\Y_t^{(1)})_{t \geq 0}$ be the Gibbs sampling chains for hard core model on graph $G_1$. Since $G_1$ is a graph consisting of isolated vertices, then the one step optimal coupling  $(\X_t^{(1)}, \Y_t^{(1)})_{t \geq 0}$ satisfies
\begin{align*}
\forall \sigma, \tau \in \Omega_{\I}:\, \E{\rho_1\left(\X^{(1)}_t,\Y^{(1)}_t\right)\mid \X^{(1)}_{t-1} = \sigma(V_1) \land \Y^{(1)}_{t-1}=\tau(V_1)} \leq \left(1-\frac{1}{|V_1|}\right)\rho_1(\sigma(V_1), \tau(V_1)).
\end{align*}
Let $(\X_t^{(2)})_{t\geq 0},(\Y_t^{(2)})_{t \geq 0}$ be the Gibbs sampling chains for hard core model on graph $G_2$. If $\lambda \leq \frac{2-\delta}{\Delta-2}= \frac{2(1-\delta/2)}{\Delta-2}$, then due to Vigoda's proof \footnote{It can be verified that in Vigoda's proof~\cite{vigoda1999fast}, the Markov chain for sampling hard core is indeed the Gibbs sampling and the coupling for analysis is indeed the one step-optimal coupling for Gibbs sampling.}, the one step optimal coupling  $(\X_t^{(2)}, \Y_t^{(2)})_{t \geq 0}$ satisfies:
\begin{align*}
\forall \sigma, \tau \in \Omega_{\I}:\, \E{\rho_2\left(\X^{(2)}_t,\Y^{(2)}_t\right)\mid \X^{(2)}_{t-1} = \sigma(V_2) \land \Y^{(2)}_{t-1}=\tau(V_2)} \leq \left(1-\frac{\delta}{96|V_2|}\right)\rho_2(\sigma(V_2), \tau(V_2)).	
\end{align*}
Let $(\X_t)_{t\geq 0},(\Y_t)_{t \geq 0}$ be the Gibbs sampling chains for hard core model on graph $G$. If $\lambda \leq \frac{2-\delta}{\Delta-2}$, then the one step optimal coupling  $(\X_t, \Y_t)_{t \geq 0}$ satisfies:
\begin{align*}
\forall \sigma, \tau \in \Omega_{\I}:\quad &\E{\rho_{\I}\left(\X_t,\Y_t\right)\mid \X_{t-1} = \sigma \land \Y_{t-1}=\tau}\\
=&\, \frac{|V_1|}{n}\left(\left(1-\frac{1}{|V_1|}\right)4\rho_1(\sigma(V_1), \tau(V_1))+ 4\rho_2(\sigma(V_2), \tau(V_2))\right)\\
&+ \frac{|V_2|}{n}\left(4\rho_1(\sigma(V_1), \tau(V_1))+ \left(1-\frac{\delta}{96|V_2|}\right)4\rho_2(\sigma(V_2), \tau(V_2))\right)\\
\leq&\, \left(1 - \frac{\min\{\delta/96,1\}}{n}\right)\rho_{\I}(\sigma, \tau).
\end{align*}
Thus, the potential function $\rho_{\I}$ satisfies the step-wise decay property.
\begin{align*}
\forall \sigma, \tau \in \Omega_{\I}:\quad &\E{\rho_{\I}\left(\X_t,\Y_t\right)\mid \X_{t-1} = \sigma \land \Y_{t-1}=\tau}\leq \left(1 - \frac{\delta/96}{n}\right)\rho_{\I}(\sigma, \tau).
\end{align*}
This proves the lemma.
\end{proof}

\section{Proofs for dynamic inference}
\label{sec-proof-inference}
\subsection{Proof of the main theorem}
\label{section-proof-main}
Our dynamic inference algorithm is given as follows.
For each MRF instance $\I=(V,E,Q,\Phi)$, where $n = |V|$, our dynamic inference algorithm maintains $N(n)$ independent samples $\*X^{(1)},\*X^{(2)},\ldots,\*X^{(N(n))} \in Q^V$ satisfying each $\DTV{\mu_{\I}}{\*X^{(i)}} \leq \epsilon(n)$ and the estimator 
$\hat{\*\theta}(\I) = \mathcal{E}(\X^{(1)}, \X^{(2)},\ldots,\X^{(N(n))})$ for $\*\theta(\I)$.
Given an update that modifies $\I$ to $\I'=(V',E',Q,\Phi')$ where $n' = |V'|$, 
our algorithm does as follows.
\begin{itemize}
\item \emph{Update the sample sequence.} Update $\*X^{(1)},\*X^{(2)},\ldots,\*X^{(N(n))}$ to $N(n')$ independent random samples $\*Y^{(1)},\*Y^{(2)},\ldots,\*Y^{(N(n'))} \in Q^{V'}$ such that each $\DTV{\mu_{\I'}}{\*Y^{(i)}} \leq \epsilon(n')$ and output the difference between two sample sequences. 
\item \emph{Update the estimator.} Given the difference between two sample sequences $\*X^{(1)},\*X^{(2)},\ldots,\*X^{(N(n))}$ and $\*Y^{(1)},\*Y^{(2)},\ldots,\*Y^{(N(n'))}$, update $\hat{\*\theta}(\I)$ to $\hat{\*\theta}(\I') =\mathcal{E}_{\*\theta}(\Y^{(1)},\Y^{(2)},\ldots,\Y^{(N(n'))})$ using the black-box algorithm in \Cref{definition-estimator-dynamic}.
\end{itemize}
Obviously, $\hat{\*\theta}(\I')$ is an $(N,\epsilon)$-estimator for $\*\theta(\I')$.

The sample sequence is maintained and updated by the dynamic sampling algorithm in~\Cref{theorem-sample-MRF}.
By ~\Cref{theorem-sample-MRF}, we have the space cost for maintaining the sample sequence is  $O\left(nN(n)\log n\right)$  memory words, each of $O(\log n)$ bits.
By following the proof of~\Cref{theorem-sample-MRF}, 
it is easy to verify that the expected time cost for each update is 
$O\left(\Delta^2 L N(n)\log ^ 3 n +\Delta n \right)$.

The estimator is maintained and updated by the black-box algorithm in~\Cref{definition-estimator-dynamic}.
By~\Cref{lemma-smooth-function}, we have $N(n) \leq \mathrm{poly}(n)$.
Combining with ~\Cref{definition-estimator-dynamic}, we have 
the space cost for maintaining the estimator is $(n\cdot N(n) +K)\mathrm{polylog}(n)$ bits.
Let $\+D$ be the size of the difference between two sample sequences as defined in~\eqref{eq-def-diff-sample}.
We can follow the proof of~\Cref{theorem-sample-MRF} to bound the expectation of $\+D$. 
Let $T = \left\lceil\frac{n}{\delta}\log \frac{n}{\epsilon(n)}\right\rceil$ and $T' = \left\lceil\frac{n'}{\delta}\log \frac{n'}{\epsilon(n')}\right\rceil$.
Since $\abs{n-n'}\leq L = o(n)$, we have $\abs{T-T'} = O(L\log n)$~(due to \Cref{lemma-smooth-epsilon}).
Combining~\eqref{eq-bound-D-Ising},~\eqref{eq-Dt-edge} and~\eqref{eq-upper-bound-T-mix} yields 
\begin{align*}
\E{\+D} &=\abs{N(n) - N(n')}\cdot \max\{n, n'\} +  O\left(L + \abs{T - T'}\right) \cdot N(n) = O(LN(n)\log n),
\end{align*}
where the last equation holds because $N(n) - N(n') = O(\frac{N(n)}{n})$ (due to \Cref{lemma-smooth-function}).
Combining with ~\Cref{definition-estimator-dynamic}, we have the expected time cost for 
updating the estimator is $LN(n)\mathrm{polylog}(n)$.

In summary, our dynamic inference algorithm
 maintains an estimator for the current MRF instance $\I$,  
 using  extra  $\widetilde{O}\left(nN(n)+K\right)$  memory words, each of $O(\log n)$ bits,
such that when $\I$ is updated  to $\I'$, the algorithm updates the estimator 
within expected time cost
\begin{align*}
\E{T_{\mathsf{cost}}} &= \E{T_{\mathsf{sample}}} + \E{T_{\mathsf{estimator}}} \\
&= O\left(\Delta^2 L N(n)\log ^ 3 n +\Delta n \right) + LN(n)\mathrm{polylog}(n)\\
&= \widetilde{O}\left(\Delta^2 L N(n) +\Delta n \right).
\end{align*}

\subsection{Dynamic inference on specific models}
Applying our dynamic inference algorithm on Ising model, $q$-coloring and hardcore model yields the following result.
\begin{theorem}
\label{theorem-inf-model}
There exist dynamic inference algorithms as stated in Theorem~\ref{theorem-infer-MRF} with the same space cost $\widetilde{O}\left(n N(n) + K\right)$, and expected time cost $\widetilde{O}\left(\Delta^2 L N(n) +\Delta n \right)$ for each update, if the input instance $\I$ with $n$ vertices and the updated instance $\I'$ with $d(\I,\I')\leq L = o(n)$ both are:
\begin{itemize}
\item Ising models with temperature $\beta$ and arbitrary local fields where $\exp(-2|\beta|) \geq 1 - \frac{2-\delta}{\Delta + 1}$;
\item proper $q$-colorings with $q\ge(2+\delta)\Delta$;
\item hardcore models with fugacity $\lambda\le\frac{2-\delta}{\Delta-2}$, but with an alternative time cost for each update 
\begin{align*}
\widetilde{O}\left(\Delta^3 L N(n) +\Delta n \right),
\end{align*}
\end{itemize}	
where $\delta>0$ is a constant, $\Delta = \max\{\Delta_{G},\Delta_{G'}\}$,
$\Delta_G$ and $\Delta_{G'}$ denote the maximum degree of the input graph and updated graph respectively.
\end{theorem}
With the dynamic sampling algorithm in \Cref{theorem-model}, \Cref{theorem-inf-model} can be proved by going through the same proof in~\Cref{section-proof-main}.

\section{Conclusion}\label{sec:conclusion}
In this paper we study probabilistic inference problem in a graphical model when the model itself is changing dynamically with time.
We study the non-local updates so that two consecutive graphical models may differ everywhere as long as the total amount of their  difference is bounded.
This general setting covers many typical applications.
We give a sampling-based dynamic inference algorithm that maintains an inference solution efficiently  against the dynamic inputs. The algorithm significantly improves the time cost compared to the static sampling-based inference algorithm. 

Our algorithm generically reduces the dynamic inference to dynamic sampling problem. 
Our main technical contribution is a dynamic Gibbs sampling algorithm that maintains random samples for graphical models dynamically changed by non-local updates. 
Such technique is extendable to all single-site dynamics. 
%
This gives us a systematic approach for transforming classic MCMC samplers on static inputs to the sampling and inference algorithms in a dynamic setting. 
Our dynamic  algorithms are efficient as long as the one-step optimal coupling exhibits a step-wise decay, a key property that has been widely used in supporting efficient MCMC sampling in the classic static setting
 and captured by the Dobrushin-Shlosman condition.

Our result is the first one that shows the possibility of 
efficient probabilistic inference 
in dynamically changing graphical models (especially when the graphical models are changed by non-local updates). 
Our dynamic inference algorithm has potentials in speeding up the iterative algorithms 
for learning graphical models, which deserves more theoretical and experimental research.
In this paper, we focus on discrete graphical models and  sampling-based inference algorithms.
Important future directions include considering more general distributions and the dynamic algorithms based on other inference techniques.


\bibliographystyle{alpha}
\bibliography{refs.bib}

\newcommand{\etalchar}[1]{$^{#1}$}
\begin{thebibliography}{NSWN17}

\bibitem[ADK{\etalchar{+}}16]{abraham2016fully}
Ittai Abraham, David Durfee, Ioannis Koutis, Sebastian Krinninger, and Richard
  Peng.
\newblock On fully dynamic graph sparsifiers.
\newblock In {\em FOCS}, 2016.

\bibitem[AQ{\etalchar{+}}17]{anacleto2017dynamic}
Osvaldo Anacleto, Catriona Queen, et~al.
\newblock Dynamic chain graph models for time series network data.
\newblock {\em Bayesian Anal.}, 12(2):491--509, 2017.

\bibitem[BC16]{bernstein2016deterministic}
Aaron Bernstein and Shiri Chechik.
\newblock Deterministic decremental single source shortest paths: beyond the
  $o(mn)$ bound.
\newblock In {\em STOC}, 2016.

\bibitem[BD97]{bubley1997path}
Russ Bubley and Martin Dyer.
\newblock Path coupling: A technique for proving rapid mixing in {Markov}
  chains.
\newblock In {\em FOCS}, 1997.

\bibitem[CLRS09]{cormen2009introduction}
Thomas~H Cormen, Charles~E Leiserson, Ronald~L Rivest, and Clifford Stein.
\newblock {\em Introduction to algorithms}.
\newblock MIT press, 2009.

\bibitem[CW07]{carvalho2007dynamic}
Carlos~M. Carvalho and Mike West.
\newblock Dynamic matrix-variate graphical models.
\newblock {\em Bayesian Anal.}, 2(1):69--97, 2007.

\bibitem[DG00]{dyer2000markov}
Martin Dyer and Catherine Greenhill.
\newblock On {M}arkov chains for independent sets.
\newblock {\em J. Algorithms}, 35(1):17--49, 2000.

\bibitem[DGGP18]{durfee2018fully}
David Durfee, Yu~Gao, Gramoz Goranci, and Richard Peng.
\newblock Fully dynamic effective resistances.
\newblock {\em arXiv preprint arXiv:1804.04038}, 2018.

\bibitem[DGGP19]{durfee2019dynamicsoectral}
David Durfee, Yu~Gao, Gramoz Goranci, and Richard Peng.
\newblock Fully dynamic spectral vertex sparsifiers and applications.
\newblock In {\em STOC}, 2019.

\bibitem[DGJ08]{dyer2008dobrushin}
Martin Dyer, Leslie~Ann Goldberg, and Mark Jerrum.
\newblock Dobrushin conditions and systematic scan.
\newblock {\em Combin. Probab. Comput.}, 17(6):761--779, 2008.

\bibitem[DS85a]{dobrushin1985completely}
Roland~L Dobrushin and Senya~B Shlosman.
\newblock Completely analytical {G}ibbs fields.
\newblock In {\em Statistical Physics and Dynamical Systems}, pages 371--403.
  Springer, 1985.

\bibitem[DS85b]{dobrushin1985constructive}
Roland~Lvovich Dobrushin and Senya~B Shlosman.
\newblock Constructive criterion for the uniqueness of {G}ibbs field.
\newblock In {\em Statistical Physics and Dynamical Systems}, pages 347--370.
  Springer, 1985.

\bibitem[DS87]{dobrushin1987completely}
RL~Dobrushin and SB~Shlosman.
\newblock Completely analytical interactions: constructive description.
\newblock {\em J. Statist. Phys.}, 46(5-6):983--1014, 1987.

\bibitem[DSOR16]{sa2016ensuring}
Christopher De~Sa, Kunle Olukotun, and Christopher R{\'e}.
\newblock Ensuring rapid mixing and low bias for asynchronous {Gibbs} sampling.
\newblock In {\em ICML}, 2016.

\bibitem[FG19]{forster2019dynamic}
Sebastian Forster and Gramoz Goranci.
\newblock Dynamic low-stretch trees via dynamic low-diameter decompositions.
\newblock In {\em STOC}, pages 377--388, 2019.

\bibitem[FVY19]{feng2019dynamic}
Weiming Feng, Nisheeth~K Vishnoi, and Yitong Yin.
\newblock Dynamic sampling from graphical models.
\newblock In {\em STOC}, 2019.

\bibitem[GHP18]{goranci2018dynamic}
Gramoz Goranci, Monika Henzinger, and Pan Peng.
\newblock {Dynamic Effective Resistances and Approximate Schur Complement on
  Separable Graphs}.
\newblock In {\em ESA}, volume 112, 2018.

\bibitem[G{\v{S}}V15]{galanis2015inapproximability}
Andreas Galanis, Daniel {\v{S}}tefankovi{\v{c}}, and Eric Vigoda.
\newblock Inapproximability for antiferromagnetic spin systems in the tree
  nonuniqueness region.
\newblock {\em J. ACM}, 62(6):50, 2015.

\bibitem[G{\v{S}}V16]{galanis2016inapproximability}
Andreas Galanis, Daniel {\v{S}}tefankovi{\v{c}}, and Eric Vigoda.
\newblock Inapproximability of the partition function for the antiferromagnetic
  {Ising} and hard-core models.
\newblock {\em Combin. Probab. Comput.}, 25(04):500--559, 2016.

\bibitem[Hay06]{hayes2006simple}
Thomas~P Hayes.
\newblock A simple condition implying rapid mixing of single-site dynamics on
  spin systems.
\newblock In {\em FOCS}, 2006.

\bibitem[Hin12]{hinton2012practical}
Geoffrey~E Hinton.
\newblock A practical guide to training restricted boltzmann machines.
\newblock In {\em Neural Networks: Tricks of the Trade}, pages 599--619.
  Springer, 2012.

\bibitem[HKN14]{henzinger2014decremental}
Monika Henzinger, Sebastian Krinninger, and Danupon Nanongkai.
\newblock Decremental single-source shortest paths on undirected graphs in
  near-linear total update time.
\newblock In {\em FOCS}, 2014.

\bibitem[HKN16]{henzinger2016dynamic}
Monika Henzinger, Sebastian Krinninger, and Danupon Nanongkai.
\newblock Dynamic approximate all-pairs shortest paths: Breaking the {$O(mn)$}
  barrier and derandomization.
\newblock {\em SIAM J. Comput.}, 45(3):947--1006, 2016.

\bibitem[Jer95]{jerrum1995very}
Mark Jerrum.
\newblock A very simple algorithm for estimating the number of {$k$}-colorings
  of a low-degree graph.
\newblock {\em Random Structures \& Algorithms}, 7(2):157--165, 1995.

\bibitem[JVV86]{JVV86}
Mark Jerrum, Leslie~G. Valiant, and Vijay~V. Vazirani.
\newblock Random generation of combinatorial structures from a uniform
  distribution.
\newblock {\em Theoret. Comput. Sci.}, 43:169--188, 1986.

\bibitem[KFB09]{koller2009probabilistic}
Daphne Koller, Nir Friedman, and Francis Bach.
\newblock {\em Probabilistic graphical models: principles and techniques}.
\newblock MIT press, 2009.

\bibitem[LMV19]{lee2019online}
Holden Lee, Oren Mangoubi, and Nisheeth Vishnoi.
\newblock Online sampling from log-concave distributions.
\newblock In {\em NIPS}, 2019.

\bibitem[LP17]{levin2017markov}
David~A Levin and Yuval Peres.
\newblock {\em Markov chains and mixing times}.
\newblock American Mathematical Soc., 2017.

\bibitem[LV99]{luby1999fast}
Michael Luby and Eric Vigoda.
\newblock Fast convergence of the {G}lauber dynamics for sampling independent
  sets.
\newblock {\em Random Structures \& Algorithms}, 15(3-4):229--241, 1999.

\bibitem[MM09]{mezard2009information}
Marc Mezard and Andrea Montanari.
\newblock {\em Information, physics, and computation}.
\newblock Oxford University Press, 2009.

\bibitem[NR17]{narayanan2017efficient}
Hariharan Narayanan and Alexander Rakhlin.
\newblock Efficient sampling from time-varying log-concave distributions.
\newblock {\em J. Mach. Learn. Res.}, 18(1):4017--4045, 2017.

\bibitem[NSWN17]{nanongkai2017dynamic}
Danupon Nanongkai, Thatchaphol Saranurak, and Christian Wulff-Nilsen.
\newblock Dynamic minimum spanning forest with subpolynomial worst-case update
  time.
\newblock In {\em FOCS}, 2017.

\bibitem[QS93]{queen1993multiregression}
Catriona~M. Queen and Jim~Q. Smith.
\newblock Multiregression dynamic models.
\newblock {\em J. Roy. Statist. Soc. Ser. B}, 55(4):849--870, 1993.

\bibitem[RKD{\etalchar{+}}19]{renggli2019continuous}
Cedric Renggli, Bojan Karla{\v{s}}, Bolin Ding, Feng Liu, Kevin Schawinski,
  Wentao Wu, and Ce~Zhang.
\newblock Continuous integration of machine learning models: A rigorous yet
  practical treatment.
\newblock In {\em SysML}, 2019.

\bibitem[{\v{S}}VV09]{vstefankovivc2009adaptive}
Daniel {\v{S}}tefankovi{\v{c}}, Santosh Vempala, and Eric Vigoda.
\newblock Adaptive simulated annealing: A near-optimal connection between
  sampling and counting.
\newblock {\em J. {ACM}}, 56(3):18, 2009.

\bibitem[SWA09]{smyth2009asynchronous}
Padhraic Smyth, Max Welling, and Arthur~U Asuncion.
\newblock Asynchronous distributed learning of topic models.
\newblock In {\em NIPS}, 2009.

\bibitem[Vig99]{vigoda1999fast}
Eric Vigoda.
\newblock Fast convergence of the {G}lauber dynamics for sampling independent
  sets: Part {II}.
\newblock Technical Report TR-99-003, International Computer Science Institute,
  1999.

\bibitem[WJ08]{wainwright2008graphical}
Martin~J. Wainwright and Michael~I. Jordan.
\newblock {\em Graphical models, exponential families, and variational
  inference}.
\newblock Now Publishers Inc, 2008.

\bibitem[WN17]{wulff2017fully}
Christian Wulff-Nilsen.
\newblock Fully-dynamic minimum spanning forest with improved worst-case update
  time.
\newblock In {\em STOC}, 2017.

\end{thebibliography}

\end{document}